%% file: TCSV-20120826.tex
\documentclass[11pt,draftclsnofoot,peerreviewca,letterpaper,onecolumn]{IEEEtran}

\usepackage{amsmath}
\usepackage{amssymb,theorem}
\usepackage{graphics}
\usepackage{epsfig,subfigure}
\usepackage[pdftex]{color}

\input{macros}

\newtheorem{theorem}{Theorem}

\newtheorem{claim}{Claim}
\newtheorem{definition}{Definition}

\newcommand{\mmse}{ {\mathsf{mmse}} } 

\newcommand{\Asf}{{\sf A}}

\newcommand{\Bsf}{{\sf B}}

\def\E{{\mathbb E}}

\newcommand{\intd}{{\,\normalfont{\text d}}} 

\newcommand{\snr}{{\sf snr}}

\newcommand{\eff}{{ \mathsf{\eta} }} 

\newcommand{\be}{\begin{eqnarray}}
\newcommand{\ee}{\end{eqnarray}}
\newcommand{\bea}{\begin{eqnarray}}
\newcommand{\eea}{\end{eqnarray}}

\begin{document}

\title{Support Recovery with Sparsely Sampled\\ Free Random Matrices}

\author{Antonia Tulino,
Giuseppe Caire,
Sergio Verd\'u and Shlomo Shamai (Shitz)
\thanks{A. Tulino is with the Wireless Communication Theory Research Bell Laboratories,
 Alcatel--Lucent, Holmdel, NJ. (a.tulino@alcatel-lucent.com)}
\thanks{G. Caire is with the Department of Electrical Engineering,
University of Southern California, Los Angeles, CA. (caire@usc.edu)}
\thanks{S. Verd\'u is with the Department of Electrical Engineering, Princeton
University, Princeton NJ. (verdu@ee.princeton.edu)}
\thanks{S. Shamai (Shitz) is with the
Department of Electrical Engineering, TechnionÐ-Israel Institute of Technology, Haifa, Israel. (sshlomo@ee.technion.ac.il)}
\thanks{Part of the results in this paper were presented at the {\em 2011 Int. Symp. on Information Theory}, Saint Petersburg, Russia,
July 31 -- August 5, 2011}
}

\maketitle
\date{\today}

\thispagestyle{empty}

\begin{abstract}
Consider a Bernoulli-Gaussian complex $n$-vector whose components are
$V_i = X_i B_i$, with  $X_i \sim \Cc\Nc(0,\Pc_x)$ and binary $B_i$
mutually independent and iid across $i$.
This random $q$-sparse vector is multiplied by a square random
matrix $\Um$, and a randomly chosen subset, of average size $n p$, $p \in [0,1]$,
of the resulting vector components is then observed in additive Gaussian noise.
We extend the scope of conventional noisy compressive sampling
models where $\Um$  is typically 
a matrix with iid components,
to allow $\Um$ satisfying a certain freeness condition. 
This class of matrices encompasses Haar matrices and other unitarily invariant matrices.  
We use the replica method and the decoupling principle
of Guo and Verd\'u, as well as a number of information theoretic bounds, to study the input-output mutual
information and the support recovery error rate in the limit of $n \to \infty$.
We also extend the scope of the large deviation approach of Rangan, Fletcher and Goyal and 
characterize the performance of a class of estimators encompassing 
thresholded linear MMSE and $\ell_1$ relaxation.
\end{abstract}

\begin{IEEEkeywords}
Compressed Sensing, Random Matrices, Rate-Distortion Theory, Sparse Models, Support Recovery, Free Probability.
\end{IEEEkeywords}

\section{Introduction} \label{intro}

\subsection{Model Setup}  \label{sc:setup}

Consider the $n$-dimensional complex-valued observation model:
\begin{eqnarray} \label{model1}
\yv &=& \Am \Um \Xm \bv + \zv \\
&=& \Am \Um \vv + \zv
\end{eqnarray}
where:
\begin{itemize}
\item $\Xm = \diag(\xv)$, and $\xv$ is an iid complex Gaussian $n$-vector with components $x_i \sim \Cc\Nc(0,\Pc_x)$;
\item $\bv$ is an iid $n$-vector with components $b_i \sim $ Bernoulli-$q$, i.e., $\mathbb{P} [b_{i} = 1] = q = 1 - \mathbb{P} [b_{i} = 0]$;
\item $\vv = \Xm \bv$ is a Bernoulli-Gaussian vector, with components $v_i = x_i b_i$;
\item $\Am$ is an $n \times n$ diagonal  matrix with iid diagonal elements $[\Am]_{i,i} \sim $ Bernoulli-$p$, i.e.,
$\mathbb{P} [[\Am]_{i,i} = 1] = p = 1 - \mathbb{P} [[\Am]_{i,i} = 0]$;
\item $\Um$ is an  $n \times n$ random matrix such that\footnote{Superscript $^\dagger$ indicates Hermitian transpose.}
\begin{eqnarray} \label{RR}
\Rm = \Um^\dagger \Am^\dagger \Am \Um
\end{eqnarray}
is \textit{free} from any deterministic Hermitian matrix (see \cite{fnt} and references therein).
\item $\zv$ is an  iid complex Gaussian $n$-vector 
with components $z_i \sim \Cc\Nc(0,1)$;
\item $\Am$, $\Um$, $\Xm$, $\bv$ and $\zv$ are mutually independent.
\item The 
signal-to-noise ratio (SNR) of the observation model (\ref{model1}) is defined as 
\begin{eqnarray}
\SNR = \frac{\EE[\|\vv\|^2]}{\EE[\|\zv\|^2]} = q \Pc_x .
\end{eqnarray}
\end{itemize}
The non-zero elements of $\bv$ define the support of the Bernoulli-Gaussian vector $\vv$, whose
``sparsity'' (average fraction of non-zero elements) equal to $q$. 
The non-zero diagonal elements of $\Am$ define the components of
the product $\Um \vv$ for which a noisy 
measurement is acquired. 
In the literature,
the number of non-zero diagonal elements of $\Am$ is commonly referred to as the number
of measurements. 
The ``sampling rate''
(average fraction of observed components) of the observation model (\ref{model1}) is equal to $p$.
The sensing matrix $\Am \Um$ is known to the signal processor, 
the goal of which
is to detect the support of $\vv$, i.e., to find the position of the non-zero components of $\bv$.

In this paper we are interested in the optimal performance of the recovery of the sparse signal support.  Denoting the recovered support 
by $\widehat{\bv} = (\widehat{b}_1, \ldots, \widehat{b}_n)^\transp$, with $\widehat{b}_i \in \{0,1\}$,
the objective is to minimize  the \textit{support recovery error rate}:
\begin{eqnarray} \label{def:d}
D^{(n)} (p, q, \Pc_x) =  \frac{1}{n} \sum_{i=1}^n \PP[ b_i \neq \widehat{b}_i ],
\end{eqnarray}
where the expectation is with respect to $\Am$, $\Um$, $\Xm$, $\bv$, and $\zv$. 
In particular, this works focuses on the large 
$n$ regime 
\begin{eqnarray} \label{limod}
D (p, q, \Pc_x) = \lim_{n \rightarrow \infty} D^{(n)} (p, q, \Pc_x)
\end{eqnarray}
under the optimal
{\em Maximum A Posteriori Symbol-By-Symbol} (MAP-SBS) estimator, as well as under some popular suboptimal but practically 
implementable estimation algorithms.

\subsection{Existing results}

Recovery of the sparsity pattern with vanishing error probability is studied
in a number of recent works such as \cite{aeron, tarokh-CS, icaap, rad, wainwright-IT, wwk}.
When $k = \sum_{i=1}^n b_i$, the number of nonzero coefficients in $\vv$, is known beforehand\footnote{Note that in our model, the number of nonzero coefficients
is not known a priori but $\frac{k}{n} \to q$.} and their magnitude is bounded away from zero, 
exact support recovery requires that the number of measurements grow as $k \log n$ \cite{icaap,wwk}. 
If the support recovery error rate is allowed to be non-vanishing, fewer measurements are necessary.
Under various assumptions, \cite{aeron,tarokh-CS,reevesphd} show that a number of measurements
growing proportionally to $k \log \frac{n}{k}$ suffices. A more refined analysis is given by Reeves and Gastpar in \cite{reevesphd,reeves-ISIT2010,reeves-journal1,reeves-journal2}, assuming that the entries of the measurement matrix are iid but without requiring the 
signal vector $\xv$ to be Gaussian. They find tight bounds on the behavior 
of the 
proportionality constant 
as a function of SNR and the target support recovery error rate.
In particular, \cite{reeves-journal1} upper bounds the required difference $p-q$ when using an ML estimator of the support. The comparison given in \cite{reeves-journal1,reeves-journal2} of computationally efficient algorithms such as linear MMSE estimation and Approximate Message 
Passing (AMP) to information theoretic bounds reveals that the suboptimality of those algorithms increases with SNR.
In contrast to \eqref{def:d}, \cite{reeves-journal2} considers a distortion measure which is the maximum of the false-alarm
and missed detection probability.

The recent work \cite{montanari-dense}
gives results for iid Gaussian measurement matrices, based on the
analysis of a message passing algorithm rather than the replica
method. A full rigorization of the decoupling principle introduced in \cite{guo-verdu}  has been recently
announced in \cite{donohoplenary} for compressive sensing applications with iid measurement matrices.
Another rigorous justification of previous replica-based results is given in \cite{wuverdunoisy}
which shows that iid Gaussian sensing matrices incur no penalty on the phase transition threshold 
with respect to an optimal nonlinear encoding.

It is of considerable interest to explore the degree of improvement afforded
by dropping the assumption that the measurement matrix has iid coefficients.
Randomly sampled Discrete Fourier Transform (DFT) 
matrices (where rows/columns are deleted independently) 
e.g. \cite{TCSV}
are one example of such matrices. The model considered in Section \ref{sc:setup} allows a relevant
generalization of the iid measurement model, which is analytically tractable.

\subsection{Organization}

Section \ref{section:main} gives expressions  for the input-output mutual information rate, 
and shows how to use it in order to lower bound the support recovery error rate. 
We write the mutual information of interest as the difference of two mutual information rates. The first term is
obtained using the heuristic replica-method, previously applied in various problems involving iid
matrices, e.g. \cite{guo-verdu,Tanaka,rangan-replica,guo-baron-shamai}. The second term is given 
rigorously, using free probability and large random matrix theory. 

Upper and lower bounds on the input-output 
mutual information  corroborating 
the replica analysis
are developed in Section \ref{bounds}. We also give
a converse result that shows that \eqref{limod} is bounded away for zero if $p \leq q$. 
Numerical examples illustrate the tightness of the bounds. 

Section \ref{decoupling} extends the {\em decoupling principle} \cite{guo-verdu} to the model in (\ref{model1})
and provides the analysis of three support estimators:  optimal MAP-SBS, thresholded 
linear MMSE  and $\ell_1$ relaxation (Lasso).  

Proofs and other technical details are given in the Appendices.

\section{Mutual information rate}\label{section:main}

In this section we are concerned with the mutual information rate
\begin{eqnarray} \label{inforate}
\Ic & \eqdef & \lim_{n \to \infty} \frac{1}{n} I(\bv; \yv | \Am,\Um) = \Ic_1 - \Ic_2
\end{eqnarray}
where
\begin{eqnarray}
\Ic_1 &\eqdef& \lim_{n \to \infty} \frac{1}{n} I(\vv; \yv | \Am,\Um) \label{inforate1} \\
\Ic_2 &\eqdef& \lim_{n \to \infty} \frac{1}{n} I(\xv; \yv | \Am,\Um, \bv). \label{inforate2}
\end{eqnarray}
and the right-most equality in (\ref{inforate}) follows from
\begin{eqnarray} \label{info-diff}
I(\bv; \yv | \Am,\Um) &=& I(\xv,\bv; \yv | \Am,\Um) - I(\xv; \yv | \Am,\Um,\bv) \\
&=&  I(\Xm\bv; \yv | \Am,\Um) - I(\xv; \yv | \Am,\Um,\bv)
\end{eqnarray}

\subsection{Error rate lower bound via mutual information}\label{section:errorate}

We can bound the minimal support recovery error rate  $D(p,q,\Pc_x)$ defined in \eqref{def:d}
in terms of $\Ic$ using the following simple result.

\begin{theorem} \label{th:totallyelementary}
Given a joint distribution $P_{XY} $ on $\mathcal{X} \times \mathcal{Y}$, a reconstruction
alphabet $\widehat{\mathcal{X}}$  and a distortion measure $ \mathsf{d}\colon \mathcal{X} \times \widehat{\mathcal{X}} \mapsto [0, \infty)$, let
\begin{eqnarray}
R (d) &\eqdef& \inf_{P_{\widehat{X}|X}\colon \mathbb{E} [ \mathsf{d} (X, \widehat{X})] \leq d} I (X;\widehat{X}) \label{poju1}
\end{eqnarray}
Then
\begin{eqnarray}
R ( \inf \mathbb{E} [\mathsf{d} ( X, \widehat{X}) ]) \leq I (X;Y) \label{ysh}
\end{eqnarray}
where the infimum is over all conditional probability assignments 
$P_{\widehat{X}|Y}$ such that $P_{XY\widehat{X}} =   P_XP_{Y|X} P_{\widehat{X}|Y}$.
\end{theorem}

\begin{IEEEproof} See Appendix \ref{proof:totallyelementary} \end{IEEEproof}

Since $R(d)$ is a monotonically decreasing function, (\ref{ysh})
gives an information theoretic lower bound on the non-information-theoretic quantity
$\inf \mathbb{E}[\mathsf{d} ( X, \widehat{X} )]$. 
In our case, using the rate-distortion function of a
Bernoulli-$q$ source with Hamming distortion, given by
$R (d) = \max \{ h(q) - h(d) , 0 \}$,  Theorem \ref{th:totallyelementary} results in
\begin{eqnarray}
D ( p, q, \Pc_x) \geq h^{-1} ( h (q) - \mathcal{I} ) \label{love}
\end{eqnarray}
where $h(x) = x \log \frac{1}{x} + (1 - x) \log \frac{1}{1-x}, \; x \in [0,1]$ denotes the binary entropy function, and where
we assume $q \leq \frac12$ (notice that $\mathcal{I} \leq h(q)$ by definition (\ref{inforate})).

\subsection{Mutual information rate ${\Ic}_1$ via replica method}\label{replica}

For any $(X,Y) \sim P_{XY}$, we denote the minimum mean-square error for estimating $X$ from $Y$ as
\begin{eqnarray}
{\sf mmse}(X|Y) \eqdef \EE[|X - \EE[X|Y]|^2].
\end{eqnarray}
With this definition, we have the following claim dependent on the validity of the replica method:

\begin{claim} \label{th:x}
Let $B_0, X_0, Z$ be independent random variables, with $B_0 \sim$ Bernoulli-$q$,
$X_0 \sim \Cc\Nc(0,\Pc_x)$, and $Z \sim \Cc\Nc(0,1)$, and define  $V_0 = X_0 B_0$.
Let $\mathcal{R}_{\Rm}(\cdot)$ denote the R-transform \cite{fnt} of the random matrix
$\Rm$ defined in (\ref{RR}). Then,
\begin{eqnarray}
\mathcal{I}_1 =  I\left ( V_0 ;  V_0  + \eta^{-\frac12} Z \right )
 +   \int_0^{\chi} \!  \! \left ( \Rc_{\Rm} (-w)  - \eta  \right) dw \,  \log e  \label{1l}, 
\end{eqnarray}
where $\eta$ and $\chi$ are the non-negative solutions of the system of equations:
\begin{subequations}
\begin{eqnarray}
\eta &=&\mathcal{R}_{\Rm}(-\chi) \label{e:fix-pointeq1}  \\
\chi &=& {\sf mmse} \left ( V_0  | V_0 + \eta^{-\frac12}  Z \right ). \label{e:fix-pointeq}
\end{eqnarray}
\end{subequations}
If the solution of (\ref{e:fix-pointeq1}) -- (\ref{e:fix-pointeq}) is not unique, then we select the solution
that minimizes $\Ic_1$ given in (\ref{1l}), which corresponds to the ``free energy'' (up to an irrelevant additive constant) of 
a physical system with ``quenched disorder parameters'' $\yv, \Am, \Um$, ``state'' $\vv \sim p_{\vv}(\vv)$ 
and  unnormalized Boltzman distribution $p_{\yv|\vv, \Am, \Um}(\yv|\vv, \Am, \Um) p_{\vv}(\vv)$, where
\begin{eqnarray} \label{transition-model1} 
p_{\yv|\vv, \Am, \Um}(\yv|\vv, \Am, \Um) = \frac{1}{\pi^n} \exp \left ( - \| \yv - \Am \Um \vv \|^2 \right ) 
\end{eqnarray}
is the conditional transition probability density of the observation model (\ref{model1}), given $\Am, \Um$. 
\end{claim}

\begin{IEEEproof} See Appendix \ref{proof:x}. \end{IEEEproof}

The efficient calculation of $I\left ( V_0 ;  V_0  + \eta^{-\frac12} Z \right )$ and of
${\sf mmse} \left ( V_0  | V_0 + \eta^{-\frac12}  Z \right )$ is addressed in Appendix \ref{app:formulas}.

\subsection{Mutual information rate $\Ic_2$ via freeness}\label{section:TAT2}

\begin{theorem}  \label{th:y}
Let $\Vc_{\Rm}(\cdot)$ and $\eta_{\Rm}(\cdot)$ denote the Shannon transform and $\eta$-transform 
(see \cite{fnt} and definitions in Appendix \ref{proof:y}) of $\Rm$ defined in (\ref{RR}).
Then,
\begin{eqnarray} \label{2I}
\mathcal{I}_2 = \Vc_{\Rm}(\alpha\,   \Pc_x) + q \log \left( 1 + \nu \Pc_x \right)  - \log ( 1 + \alpha\,  \nu \Pc_x)
\end{eqnarray}
where  $\alpha$ and $\nu$ are the unique non-negative solutions of the system of equations
\begin{eqnarray}
\eta_{\Rm}(\alpha\,  \Pc_x)  = \frac{1}{1
+ \alpha\,  \nu \Pc_x} =  \frac{q}{ 1 + \nu \Pc_x} + 1-q \label{e:fix-pointeq-alpha-nu}
\end{eqnarray}
\end{theorem}

\begin{IEEEproof} See Appendix \ref{proof:y}. \end{IEEEproof}

\subsection{Special Cases} \label{section:specialcases}

\subsubsection{$\Um$ is an iid random matrix}
Assuming $\Um$ has iid entries with mean zero and variance
$\frac{1}{n}$, according to \cite[Theorem 2.39]{fnt} 
the $\eta$-transform of $\Rm$ satisfies the relation
\begin{eqnarray} \label{etaR}
1 = \frac{1 - \eta_{\Rm}(x)}{1 - \eta_{\Tm}(x \eta_{\Rm}(x))}
 \end{eqnarray}
with  $\Tm = \Am^\dagger \Am$. Using the fact that $\Am$ is diagonal with Bernoulli-$p$ iid diagonal elements,
\begin{eqnarray} \label{etaA}
\eta_{\Tm}(x) = \eta_{\Am}(x)  = 1 - p + \frac{p}{1+x}
\end{eqnarray}
Using this in (\ref{etaR}), we have that $\eta_{\Rm}(x)$ is the positive solution of the quadratic equation
\begin{eqnarray} \label{etaR1}
x\eta^2 - ((1-p)x - 1) \eta - 1 = 0,
\end{eqnarray}
which corresponds to the $ \eta$-transform of a random matrix of the form
$\Hm \Hm^\dagger$, with $\Hm$ of dimension $n \times pn$ and iid elements with zero mean and variance $1/n$.
The R-transform of such matrix is well-known (see \cite[Example 2.27]{fnt}) and takes on the form
\begin{eqnarray} \label{R-transfromGV}
\mathcal{R}_{\Rm} ( z ) = \frac{p}{1 - z}.
\end{eqnarray}
Hence, the fixed point equations (\ref{e:fix-pointeq1}) -- (\ref{e:fix-pointeq}) reduce to
\begin{eqnarray} \label{matched-mmse-fixed-point}
\frac{1}{\eta} =  \frac{1}{p} \left ( 1 +  \displaystyle{{\sf mmse} \left ( V_0  | V_0 + \eta^{-\frac12} Z \right ) } \right ),
\end{eqnarray}
and  \eqref{1l} takes on the form
\begin{eqnarray} \label{th:dpiid}
\mathcal{I}_1 = I \left( V_0;  V_0 + \eta^{-\frac12}  Z  \right) + p \left ( \log \left (\frac{p}{\eta} \right ) + \left( \frac{\eta}{p} - 1\right) \log e \right).
\end{eqnarray}
This is obtained from (\ref{1l}) using (\ref{R-transfromGV})  for the R-transform and the identity
$\frac{1}{\eta} = \frac{1}{p} \left ( 1 + \chi  \right )$, from  (\ref{e:fix-pointeq1}).
We notice that  when $p=1$ (\ref{th:dpiid}) coincides with the result in \cite{guo-verdu}.
The formula provided by Claim \ref{th:x} does not coincide with  the result in \cite{guo-verdu,rangan-replica} 
for general $p$ since in the model considered by \cite{guo-verdu,rangan-replica}  the ``channel matrix''  $\Am \Um$ is normalized
such that  the columns (and not the non-zero rows, as in our setting)  have unit average squared norm conditioned on $\Am$.
Instead, our formulas are consistent with those in \cite{reeves-journal1}, which uses the same row-energy normalization as in this paper.

In order to calculate $\mathcal{I}_2$, we use  (\ref{e:fix-pointeq-alpha-nu}) and obtain
\begin{eqnarray} \label{sosad2}
\alpha \Pc_x  =  \frac{1}{\nu} \left (\frac{1}{\eta_{\Rm}(\alpha \Pc_x)}-1 \right).
\end{eqnarray}
Using the definition of S-transform (see Definition \ref{def:S} In Appendix \ref{proof:y}), 
we have that 
\begin{eqnarray} \label{funny-beppe}
\alpha \Pc_x  =  \Sigma_{\Rm}(\eta_{\Rm}(\alpha \Pc_x) - 1) \left (\frac{1}{\eta_{\Rm}(\alpha \Pc_x)}-1 \right), 
\end{eqnarray}
from which, identifying terms, we obtain
\begin{eqnarray}
\label{sosad3}
\nu  =  \frac{1}{\Sigma_{\Rm}\left (\eta -1 \right)} = \eta -1+p,
\end{eqnarray}
where for simplicity we let $\eta = \eta_{\Rm}(\alpha \Pc_x)$ and
where the rightmost equality follows from the well-known explicit expression $\Sigmam_{\Rm}(z) = \frac{1}{z + p}$, valid when 
$\Um$ is an iid matrix.   Replacing  (\ref{sosad3}) in the equality $\eta = \frac{q}{ 1 + \nu \Pc_x} + 1-q$ in (\ref{e:fix-pointeq-alpha-nu}), we obtain
\begin{eqnarray}
\label{sosad4}
\eta = \frac{q}{1+(\eta-1+p)\Pc_x} + 1-q.
\end{eqnarray}
Defining $\Gc = \nu /p$ we can rewrite (\ref{sosad4}) as
\begin{eqnarray}
\label{sosad5}
\mathcal{G} = 1-  \frac{q}{p} + \frac{\frac{q}{p}}{1+p  \mathcal{G} \Pc_x }.
\end{eqnarray}
Hence, $\mathcal{G}$ is seen to satisfy a well-known fixed-point equation yielding $\Gc = \eta_{\widetilde{\Hm} \widetilde{\Hm}^\dagger}(p \,\Pc_x)$, 
where $\widetilde{\Hm}$ is a $pn \times qn$ matrix with  iid with variance $1/(pn)$  (see \cite[Eq. (2.120)]{fnt}). 
Using \cite[Eq. (2.121)]{fnt}, $\mathcal{G}$ can be obtained in closed form as
\begin{eqnarray} 
\mathcal{G} =1 - \frac{\Fc \left (p \, \Pc_x, \frac{q}{p} \right)}{4 \Pc_x}, 
\end{eqnarray}
where
\begin{eqnarray} 
\Fc(x,y) = \left ( \sqrt{x(1 + \sqrt{y})^2 + 1} - \sqrt{x(1 - \sqrt{y})^2 + 1} \right )^2,  \end{eqnarray}
and the corresponding Shannon transform yields the desired $\Ic_2$, in the form
\begin{eqnarray}  \label{I2iid}
\mathcal{I}_2 & = & 
q\log\left(1+ p \, \Pc_x - \frac{1}{4} \Fc\left (p\, \Pc_x,\frac{q}{p} \right ) \right)  +
p \log \left(1+ q \, \Pc_x  - \frac{1}{4}  \Fc\left (p\, \Pc_x,\frac{q}{p}\right) \right) \nonumber \\
& & - \frac{1}{4\Pc_x} \Fc\left (p\, \Pc_x,\frac{q}{p}\right ) \log e .
\end{eqnarray}
In passing, we remark that the  ``large SNR'' (i.e., large $\Pc_x$) behavior 
of (\ref{I2iid}) is 
\begin{eqnarray}
\Ic_2 = \min\{p,q\} \log(1 + |p - q| \Pc_x) + O(1)
\end{eqnarray}
showing that the pre-log of $\Ic_2$ is the asymptotic almost sure normalized 
rank of the matrix $\Am \Um \diag(\bv)$, as expected.

\subsubsection{$\Um$ is Haar-distributed}\label{subs:haar}

If $\Um$ is Haar-distributed, i.e., uniformly distributed on the manifold of $n \times n$ unitary matrices,
the eigenvalue distribution of $\Rm$ coincides with that of $\Am \Am^\dagger = \Am$,
i.e., with the Bernoulli-$p$ distribution. Using (\ref{etaA}) and the relation between the $\eta$-transform and the R-transform
in \cite[Eq. 2.74]{fnt}, we obtain
\begin{eqnarray} \label{cool2}
\Rc_{\Rm} ( z ) = \Rc_{\Am} ( z ) =  \frac{z -1 + \sqrt{(z - 1)^2 + 4zp}}{2 z}.
\end{eqnarray}
This allows for the calculation of \eqref{1l} with the corresponding fixed point equations
(\ref{e:fix-pointeq1}) and (\ref{e:fix-pointeq}).

As far as $\Ic_2$ is concerned, we use
\[ \eta_{\Rm}(\alpha\Pc_x) = \eta_{\Am}(\alpha\Pc_x) = \frac{p}{1 + \alpha \Pc_x} + 1 - p \]
in (\ref{e:fix-pointeq-alpha-nu}) and solve for $\alpha$ using the first equality, obtaining
\begin{eqnarray} \label{alpha-ziofa} 
\alpha = \frac{p-\nu}{\nu \Pc_x (1 - p)}. 
\end{eqnarray}
Replacing in the second equality in (\ref{e:fix-pointeq-alpha-nu}), we obtain explicitly $\nu$ as
\begin{eqnarray} \label{fixpointsolve1-ziofa}
\nu = \frac{\Pc_x (p - q) - 1 + \sqrt{(\Pc_x (p - q) - 1)^2 + 4 \, p\,  \Pc_x(1 - q)}}{2\Pc_x(1 - q)}.
\end{eqnarray}
It can be checked that $0 < \nu \leq p$ for any $\Pc_x > 0$ and $p,q$ in $[0,1]$. 
Using (\ref{alpha-ziofa}) and (\ref{fixpointsolve1-ziofa}) (\ref{2I}), we obtain
\begin{eqnarray} \label{tvcs-bernoulli-bernoulli}
\Ic_2 = q \log \left( 1 + \nu \, \Pc_x \right) + d ( p || \nu )
\end{eqnarray}
where
\begin{eqnarray} \label{sanvalentino}
d ( a || b ) = a \log \frac{a}{b} + (1 - a ) \log \frac{1-a}{1-b}
\end{eqnarray}
is the binary relative entropy. The expression (\ref{tvcs-bernoulli-bernoulli}) coincides with the result given in  \cite{TCSV} 
for the limit of the mutual information rate
\begin{eqnarray} \label{funny-beppe1}
\frac{1}{n}  I(\xv; \Am \Um \Bm \xv + \zv| \Am, \Um, \Bm)  =  \frac{1}{n} \EE \left [ \log \left | \Id + \Pc_x \Um^\dagger \Am \Um \Bm \right | \right ], 
\end{eqnarray}
of a vector Gaussian channel with iid Gaussian input $\xv$, and channel matrix $\Am \Um\Bm$ with $\Bm = \diag(\bv)$.

\subsubsection{$\Am=\Id$, unitary  $\Um$}

In this case, $\Rm = \Um^\dagger \Am^\dagger \Am \Um =\Id$ and  $\mathcal{R}_{\Rm} ( z ) = 1$.
Hence, (\ref{e:fix-pointeq1}) and (\ref{e:fix-pointeq}) become
\begin{eqnarray}
\eff &=&1  \label{e:fix-pointeq1r}  \\
\chi &=& \frac{\Pc_x}{1 + \Pc_x}.
\label{e:fix-pointeqr}
\end{eqnarray}
Since $\Am = \Id$ implies $p=1$,  (\ref{fixpointsolve1-ziofa}) yields $\nu = 1$ and (recalling \eqref{tvcs-bernoulli-bernoulli}), we have
\begin{eqnarray}
\mathcal{I} &=&  I(V_0; V_0 + Z)  -  q \log \left( 1 + \Pc_x \right)\label{bingo-bongo}\\
&=& I (V_0 ; V_0 + Z) - I ( V_0 ; V_0 + Z | B_0 ) \\
&=& I ( B_0 ; V_0 + Z) \label{shlomoknocksatdoor} \\
&=& h(q) - H(B_0 | V_0 + Z), \label{minfo-bound0}
\end{eqnarray}
where  \eqref{shlomoknocksatdoor} follows because $V_0 + Z$ and $B_0$ are independent conditioned on $V_0$.
In fact, in this case, the single-letter expression $\Ic = \frac{1}{n} I(\bv; \yv | \Am=\Id,\Um)$ holds for all $n$, not only in the 
limit of $n \rightarrow \infty$.

\section{Bounds on the Mutual Information Rate} \label{bounds}

\subsection{Upper Bounds}  \label{sec:upbounds} 

We start with the following result, which follows immediately from first principles. 

\begin{theorem} \label{thm:dataprocessing}
If  $\Um$ is unitary, then (\ref{inforate}) satisfies
\begin{eqnarray} \label{minfo-bound2ex}
\mathcal{I} \leq  I(V_0; V_0 + Z)  -  q\log \left( 1 + \Pc_x \right),
\end{eqnarray}
where $Z$ and $V_0$ are as defined in Claim \ref{th:x}.
Equation \eqref{minfo-bound2ex} holds with equality for $\Am=\Id$.
\end{theorem}

\begin{IEEEproof}
It is sufficient to notice that the output $\yv$ in (\ref{model1}) is obtained by sampling the vector 
$\Um \Xm \bv + \zv$ at the positions of the ``1'' elements of the diagonal of $\Am$.
From the data processing inequality and noticing that $\frac{1}{n} I(\bv; \Um \Xm \bv + \zv)$ is given by 
\eqref{bingo-bongo}, the result follows.
\end{IEEEproof}

In the general case, we have the following upper bounds
\begin{theorem} \label{th:bounds}
\begin{align}
\Ic_1 & \leq  \Vc_{\Rm} (q \, \Pc_x) \;\;\;\;\; \label{minfo-bound1} \\
\Ic_1 & \leq  I\left (V_0;   \sqrt{\E \left[ | \mathsf{R} |^2   \right]} V_0 + Z \right )  \;\;\;\;\; \label{minfo-bound3}
\end{align}
where $Z$ and $V_0$ are as defined in Claim \ref{th:x}, and 
where $|{\sf R}|^2$ is a random variable distributed as the limiting spectrum
of $\Rm$.
\end{theorem}

\begin{IEEEproof} See Appendix \ref{proof:bounds} \end{IEEEproof}

\subsection{Lower Bounds}  \label{sec:lowerbounds} 

In order to corroborate the exact result of Claim \ref{th:x} obtained through the heuristic replica method, we also consider a lower bound to
the mutual information. Since $\Ic_2$ is known exactly, it is sufficient to have a lower bound for $\Ic_1$. This is provided by the following result:


\begin{theorem} \label{LB-theorem}
The mutual information rate in
 (\ref{inforate1}) is lower bounded by
\begin{eqnarray} \label{I1LB}
\Ic_1 \geq \int_0^1 I\left ( V_0;  \sqrt{\eta(q \Pc_x; \beta)} \, V_0 + Z \right ) \; d\beta,
\end{eqnarray}
where $Z$ and $V_0$ are as defined in Claim \ref{th:x} and where $\eta(s; \beta)$ is defined by
\begin{eqnarray} \label{eta-beta}
\eta(s; \beta)  =  \lim_{n \rightarrow \infty} \uv_i^\dagger \Am^\dagger \left [ \Id  + s \Am \Um_{i-1} \Um_{i-1}^\dagger \Am^\dagger \right ]^{-1} \Am \uv_i
\end{eqnarray}
where $i = \lfloor n \beta \rfloor$.
\end{theorem}

\begin{IEEEproof} See Appendix \ref{proof:bounds} \end{IEEEproof}

It is interesting to notice that the quantity defined in (\ref{eta-beta}) can be interpreted as the asymptotic (in $n$) {\em multiuser efficiency} of
a CDMA system  $\rv = \Am \Um \vv + \zv$ with input $\vv$, output $\rv$ and spreading codes given by the columns of $\Am\Um$, where the receiver uses linear MMSE detection with successive decoding, and the input symbols $v_{i+1}, \ldots, v_n$ have been already decoded and subtracted from the received signal
(see \cite{fnt,verdu-shamai-random-cdma-fading}).
Hence, the integral in (\ref{I1LB}) can be regarded as the mutual information between the input $\vv$ and the output of a
{\em mismatched} successive interference cancellation receiver that treats the symbols of $\vv$ as if they were Gaussian iid, instead of
Bernoulli-Gaussian.

Explicit expressions for $\eta(s; \beta)$ can be provided in
several cases of interest. For example, when $\Um$ has iid entries,
using \cite[Theorem 2.52]{fnt} we obtain $\eta(s; \beta)
= \eta$, given by the solution of the fixed-point equation
\begin{eqnarray}
\eta = \frac{p}{1 + \beta \frac{s}{1 + s \eta}},
\end{eqnarray}
namely,
\begin{eqnarray}  \label{eta-iid}
\eta (s;\beta) = \frac{(p-\beta)s - 1 +
\sqrt{((p-\beta)s - 1)^2 + 4 ps}}{2 s}
\end{eqnarray}

In the case of Haar-distributed $\Um$, using \cite[Eq.
3.112]{fnt} we obtain $\eta(s; \beta) = \eta$, given by
the solution of the fixed-point equation
\begin{eqnarray}  \label{eta-haar-fixedp}
\frac{\eta}{1 + s \eta} = \frac{p}{1 + \beta s + (1 -
\beta)  s \eta},
\end{eqnarray}
namely,
\begin{eqnarray} \label{eta-haar}
\eta (s;\beta) = \frac{ (p - \beta) s - 1 +
\sqrt{((p-\beta)s - 1) ^2 + 4(1-\beta) p s}}{2(1-\beta)
s}
\end{eqnarray}

Using the mean-value theorem in  (\ref{I1LB}), there exists some $\beta^* \in [0,1]$ such that
\begin{eqnarray} \label{I1LB-meanvalue}
\int_0^1 I\left ( V_0; {\sqrt{\eta(q \Pc_x; \beta)}} \, V_0 +  Z \right ) \; d\beta =
I\left ( V_0; {\sqrt{\eta(q\Pc_x; \beta^*)}} \, V_0 +  Z \right )
\end{eqnarray}
which is in the same form as the upper bound (\ref{minfo-bound3}) 
save for a different  signal-to-noise ratio between the Bernoulli-Gaussian input and the Gaussian noise.

It is also immediate to notice that the upper and lower bounds on $\Ic_1$ hold for any fixed deterministic $\Um$, provided that
the limits exist. For example, in the case of $\Um = \Fm$, a deterministic unitary DFT matrix, \cite{TCSV} shows that
$\eta(s;\beta)$ takes on the same form (\ref{eta-haar}) as well as the exact expression for $\Ic_2$ is still given by
Theorem \ref{th:y}. Hence, it follows that while at the moment we can develop the replica analysis only for $\Um$ random,
satisfying the freeness requirement as said above, the mutual information for a deterministic DFT matrix satisfies the same bounds.
In fact, we have numerical evidence (see Section \ref{sec:dist-results}) 
that leads us to conjecture that the replica result of Claim \ref{th:x} applies also to a DFT sensing matrix, although the proofs
of this paper do not extend to this case.

\subsection{High-SNR Regime}  \label{sec:hiSNR} 

\begin{theorem} \label{th:totallysupercool}
For the observation model (\ref{model1}) and any support estimator, 
$D( p, q, \Pc_x)$ is bounded away from zero for $0 \leq p \leq q$, even in the noiseless case.  
\end{theorem}

\begin{proof}
From (\ref{love}) it is evident that $D(p,q,\Pc_x)$ is bounded away from zero if $\Ic < h(q)$. From the definition of 
the mutual information rate $\Ic$ (see (\ref{inforate})), it is immediate that $\Ic < h(q)$ for any finite $\Pc_x$.
However, in the limit of high SNR, $\Ic$ may or may not converge to $h(q)$ depending on the system parameters $p$ and $q$.  In the remainder of the proof we show that
\begin{eqnarray}
\lim_{\Pc_x\to \infty } \Ic < h (q)
\end{eqnarray}
provided $0 < p \leq q$. The case $p= 0$ is trivial.

Recall from Theorem  \ref{th:y} 
that 
\begin{eqnarray}
\eta_{\Rm}(\alpha\,  \Pc_x)  &=& \frac{1}{1+ \alpha\,  \nu\, \Pc_x} \label{ITA1}\\
&=&  \frac{q}{ 1 + \nu \Pc_x} + 1-q \label{ITA2}
\\
\mathcal{I}_2 (  \Pc_x ) &=& \Vc_{\Rm}(\alpha\,   \Pc_x) + q \log \left( 1 + \nu \Pc_x \right)  - \log ( 1 + \alpha\,  \nu \Pc_x) \label{888}
\end{eqnarray}
where we have made explicit the dependence of $\mathcal{I}_2$ on $\Pc_x$.
For the purposes of the proof it is important to elucidate the behavior of 
$\alpha\,   \Pc_x$, $\nu\,   \Pc_x$, and $\alpha \nu\,   \Pc_x$ as  $\Pc_x\to \infty$,
where $\nu$ and $\alpha$ depend on $\Pc_x$ through \eqref{ITA2}.
In principle, there are nine possibilities:
\begin{enumerate} 
\item 
$  \alpha \Pc_x \to 0$  and $ \nu \Pc_x \to 0$.
\item 
$  \alpha \Pc_x \to 0$  and $ 0 < \lim_{\Pc_x\to \infty } \nu \Pc_x < \infty$.
\item
$  \alpha \Pc_x \to 0$  and $\nu \Pc_x $  diverges.
\item
$  0 < \lim_{\Pc_x\to \infty } \alpha \Pc_x < \infty$  and $ \nu \Pc_x \to 0$.
\item 
$  0 < \lim_{\Pc_x\to \infty } \alpha \Pc_x < \infty$  and $ 0 < \lim_{\Pc_x\to \infty } \nu \Pc_x < \infty$.
\item\label{case6}
$  0 < \lim_{\Pc_x\to \infty } \alpha \Pc_x < \infty$  and $\nu \Pc_x $  diverges.
\item
$\alpha \Pc_x$ diverges  and $ \nu \Pc_x \to 0$.
\item \label{case8}
$\alpha \Pc_x$ diverges  and $ 0 < \lim_{\Pc_x\to \infty } \nu \Pc_x < \infty$.
\item \label{case9}
$\alpha \Pc_x$ diverges  and $\nu \Pc_x $  diverges.
\end{enumerate}
The asymptotic behavior of  \eqref{888} is
\begin{eqnarray} \label{chico}
\lim_{\Pc_x\to \infty } \frac{\mathcal{I}_2 (  \Pc_x )}{\log \Pc_x } &=& p
\end{eqnarray}
since $\frac{1}{n} {\rm rank} ( \Am \Um \Bm ) \rightarrow \min \{p, q\}$ with probabilty 1.

In view of  \eqref{ITA1}, $\alpha \nu\,   \Pc_x$ cannot diverge when $p < 1$, since 
\begin{eqnarray} \label{bacioilly}
1-p \leq 
\eta_{\Rm}(\alpha\,  \Pc_x) \leq  1
\end{eqnarray}
where the lower bound is the limit of $\eta_{\Rm}(\alpha\,  \Pc_x) $ if $\alpha\,  \Pc_x \to \infty$
while the upper bound is the limit of $\eta_{\Rm}(\alpha\,  \Pc_x) $ if $\alpha\,  \Pc_x \to 0$.

\begin{enumerate}
\item Impossible because it would contradict \eqref{chico}.
\item Impossible because it would contradict \eqref{ITA2}.
\item Impossible because it would contradict \eqref{ITA2} since $q > 0$.
\item Impossible because it would contradict \eqref{chico}.
\item Impossible because then $\alpha \nu\,   \Pc_x \to 0$ and \eqref{ITA2} would be contradicted.
\item Impossible if $p < q$ since $\eta_{\Rm}(\alpha\,  \Pc_x) = 1 -q$ would be outside the range established in \eqref{bacioilly}.
If $p=q$ then the lower limit in \eqref{bacioilly} would be achieved at a finite argument of $\eta_{\Rm}$
which is impossible due to the strictly monotonic nature of that function.
\item 
Impossible because it would contradict \eqref{ITA2}.
\item Impossible if $p=q$   because it would contradict \eqref{ITA2}. The case $p < q$ is treated below.
\item Impossible if $p<q$ because it would contradict \eqref{ITA2}. 
The case $p=q$ is treated below.
\end{enumerate}
 We proceed to consider case \ref{case8}) when $p <q$.
 The solution of the fixed-point equation \eqref{ITA1}-\eqref{ITA2} yields
\begin{eqnarray}
\lim_{\Pc_x\to \infty } \frac{q}{ 1 + \nu \Pc_x}  &=& q - p  \label{subi1}\\
\lim_{\Pc_x\to \infty } \frac{1}{ 1 + \alpha \nu \Pc_x}  &=& 1 - p \\
\lim_{\Pc_x\to \infty } \nu S &=& \frac{p}{q-p} \\
\lim_{\Pc_x\to \infty } \alpha &=& \frac{q-p} {1-p} \label{subi4}
\end{eqnarray}
We can proceed to upper bound $\Ic$ using Theorem \ref{th:bounds} and \eqref{subi1}-\eqref{subi4}:
\begin{eqnarray} \label{antonia-ub1}
\Ic &\leq& {\mathcal V}_{\Rm}(q \Pc_x)  -   {\mathcal V}_{\Rm}(\alpha  \Pc_x)  - q \log (1+ \nu \Pc_x)  + \log(1+ \alpha \nu \Pc_x)
\label{ziobeberu} \\
&\rightarrow&
 (1-p) \log \frac{1}{1 - p} - (q - p) \log \frac{q}{q-p}
\\
&=&
 (1-q) \log \frac{1}{1 - p} - (q - p) \log \frac{q(1-p)}{q-p}
 \\
 &<&
 (1-q) \log \frac{1}{1 - q} + (q - p) \log \frac{1}{q} + (q - p) \log \frac{q-p}{1-p}
 \\
 &<&
 (1-q) \log \frac{1}{1 - q} + q \log \frac{1}{q}
 \\
 &=&h(q)
\end{eqnarray}

We now  proceed to consider case \ref{case9}) when $p =q$. 
In this case, the solution of the fixed-point equation \eqref{ITA1}-\eqref{ITA2} yields
\begin{eqnarray}
\lim_{\Pc_x\to \infty } \frac{1}{ 1 + \alpha \nu \Pc_x}  &=& 1 - q
\label{tiam0}
\\
\lim_{\Pc_x\to \infty }  \alpha \nu \Pc_x  &=& \frac{q}{1 - q}\label{tiam9}
\end{eqnarray}

As before, we can now  proceed to upper bound $\Ic$ using Theorem \ref{th:bounds}:
\begin{eqnarray} \label{antonia-ub2}
\Ic &\leq& {\mathcal V}_{\Rm}(q \Pc_x)  -   {\mathcal V}_{\Rm}(\alpha  \Pc_x)  - q \log (1+ \nu \Pc_x)  + \log(1+ \alpha \nu \Pc_x)
\\
&=& \left( {\mathcal V}_{\Rm}(q \Pc_x) - q \log ( 1 + q \Pc_x ) 
- {\mathcal V}_{\Rm}(\alpha \Pc_x) + q \log ( 1 + \alpha \Pc_x ) \right)
 \nonumber\\
 & & + q \log \frac{1 + q \Pc_x}{( 1 + \alpha \Pc_x ) ( 1 + \nu \Pc_x )} + \log(1+ \alpha \nu \Pc_x)\\
&\rightarrow&
 (1-q) \log   \frac{1}{1-q} \label{juju}\\
 &<& h(q)
\end{eqnarray}
where \eqref{juju} follows from \eqref{tiam0}, \eqref{tiam9} and the fact that the first term in the left side
vanishes as $\Pc_x \to \infty$.
\end{proof}

Note that an achievability counterpart to Theorem \ref{th:totallysupercool} in
the noiseless case (under a more general signal model) is given in \cite{wuverdunoiseless},
showing that $p = q$ is the critical sampling rate threshold 
for exact reconstruction.

\subsection{Examples}  \label{sec:info-results}

We provide a few numerical examples illustrating the results developed before.
Figs. \ref{info0db}, \ref{info20db} and \ref{info50db} show the mutual information rate $\Ic$ as a function of the sampling rate $p$, for a
Haar-distributed sensing matrix $\Um$ and a Gaussian-Bernoulli source signal $\vv$ 
with $q = 0.2$ and $\SNR = q\Pc_x$ is equal to
0, 20 and 50 dB, respectively. Each figure show also the corresponding lower and upper bounds provided by Theorems
\ref{thm:dataprocessing}, \ref{th:bounds} and \ref{LB-theorem}. We notice that the lower bound of Theorem \ref{LB-theorem} is close to the
exact value of $\Ic$ for low SNR (in fact, it is tight for $\Pc_x \rightarrow 0$). In contrast, for high SNR, 
the mutual information $\Ic$ is very closely approximated by the minimum of the two 
upper bounds provided by Theorem \ref{thm:dataprocessing} and (\ref{minfo-bound1}) in Theorem \ref{th:bounds}. 

\begin{figure}[ht]
\centerline{\includegraphics[width=12cm]{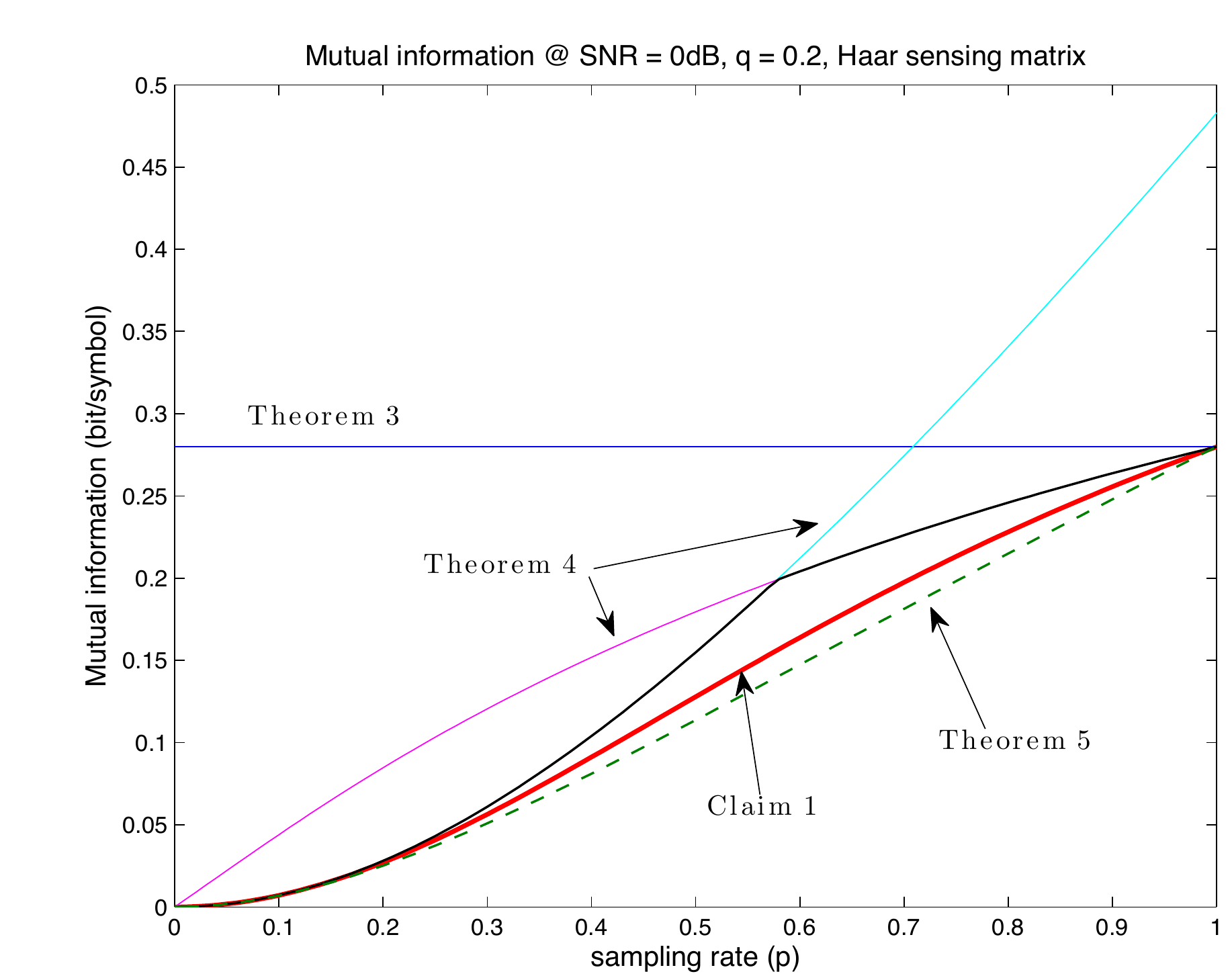}}
\caption{Mutual information rate $\Ic$ versus $p$, for $q = 0.2$ and $\SNR = q\Pc_x = 0$ dB. Upper and lower bounds
are also shown for comparison.}
\label{info0db}
\end{figure}

\begin{figure}[ht]
\centerline{\includegraphics[width=12cm]{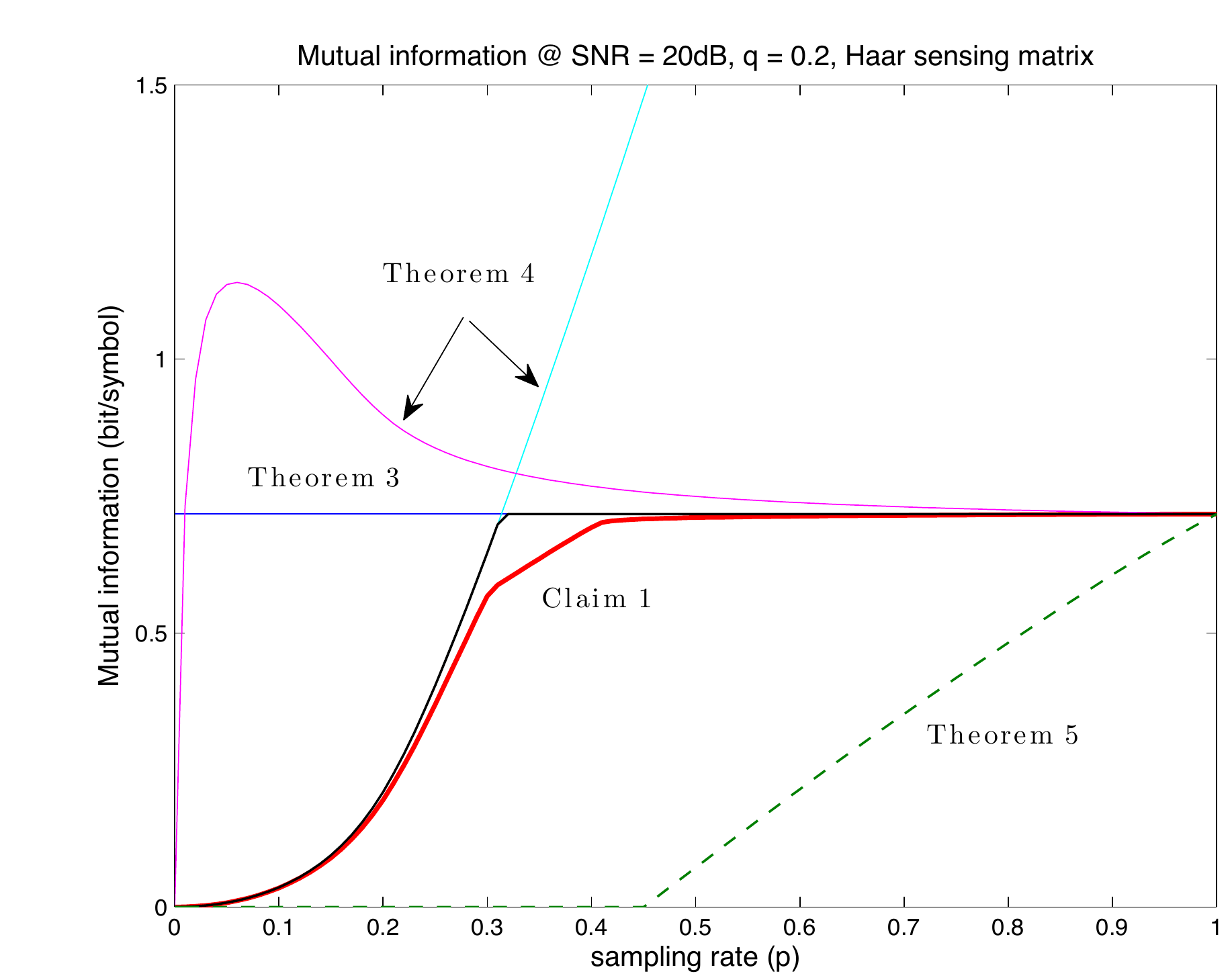}}
\caption{Mutual information rate $\Ic$ versus $p$, for $q = 0.2$ and $\SNR = q\Pc_x = 20$ dB. Upper and lower bounds
are also shown for comparison.}
\label{info20db}
\end{figure}

\begin{figure}[ht]
\centerline{\includegraphics[width=12cm]{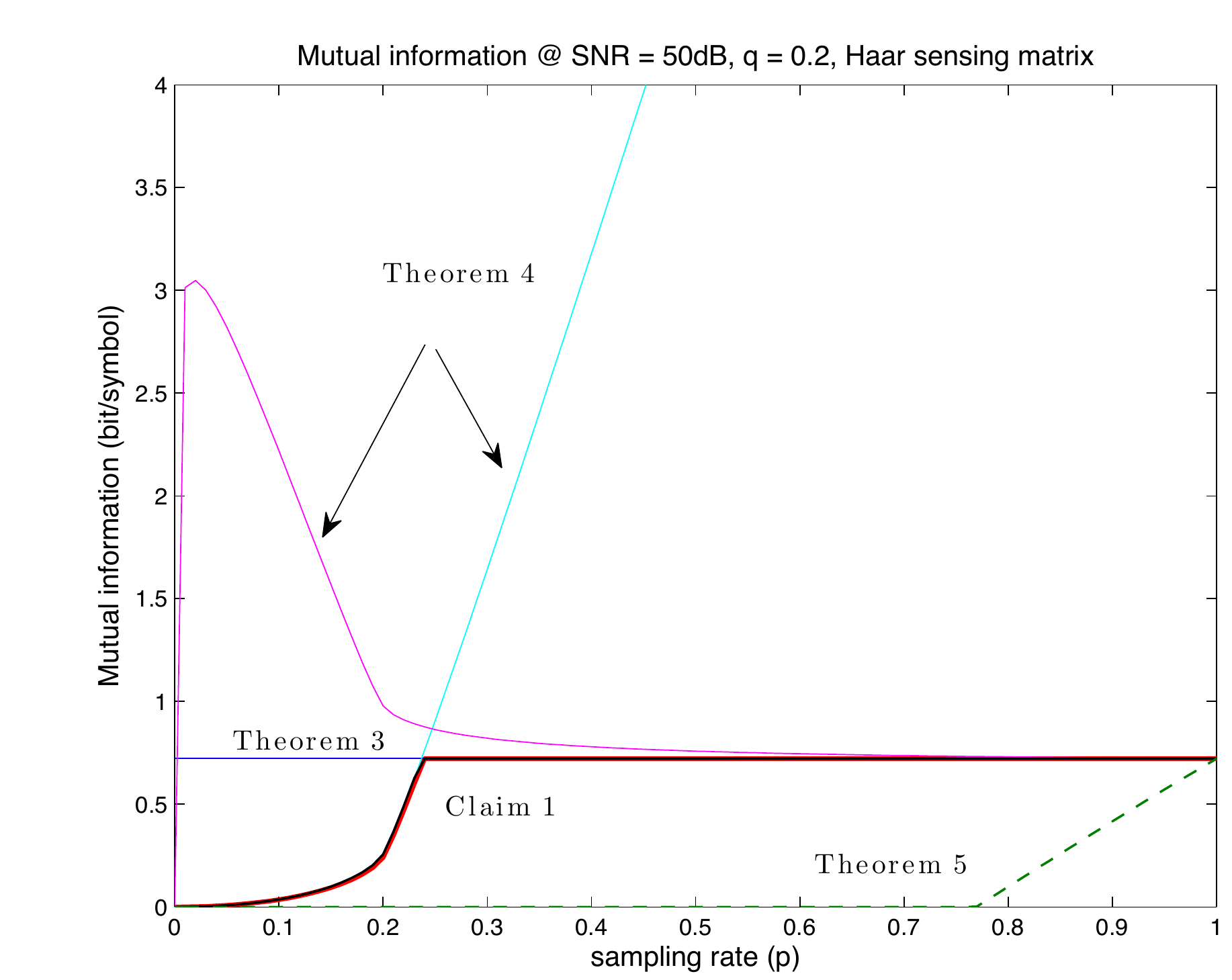}}
\caption{Mutual information rate $\Ic$ versus $p$, for $q = 0.2$ and $\SNR = q\Pc_x = 50$ dB. Upper and lower bounds
are also shown for comparison.}
\label{info50db}
\end{figure}

It is also interesting to observe that the asymptotic regime of vanishing $D(p,q,\Pc_x)$ for any $p > q$ is approached
very slowly, i.e., an impractically high SNR is required. For example, we notice that
at $\SNR = 50$ dB the mutual information $\Ic$ in Fig.~\ref{info50db} achieves the upper upper bound of Theorem \ref{thm:dataprocessing}
(very close to $h(q)$) at $p = 0.24$, which is quite far from the threshold $q = 0.2$. 
Fig.~\ref{ub1-highSNR} shows $\Ic_{\rm ub}$ evaluated at $q = 0.2, p = 0.205$ versus SNR in dB. In order to reach the 
value $h(q) = 0.722$ bits, we need an SNR of about 340 dB. This gives an idea of ``how high'' the high-SNR regime must be, in order to
work closely to the noiseless reconstruction threshold.

\begin{figure}[ht]
\centerline{\includegraphics[width=12cm]{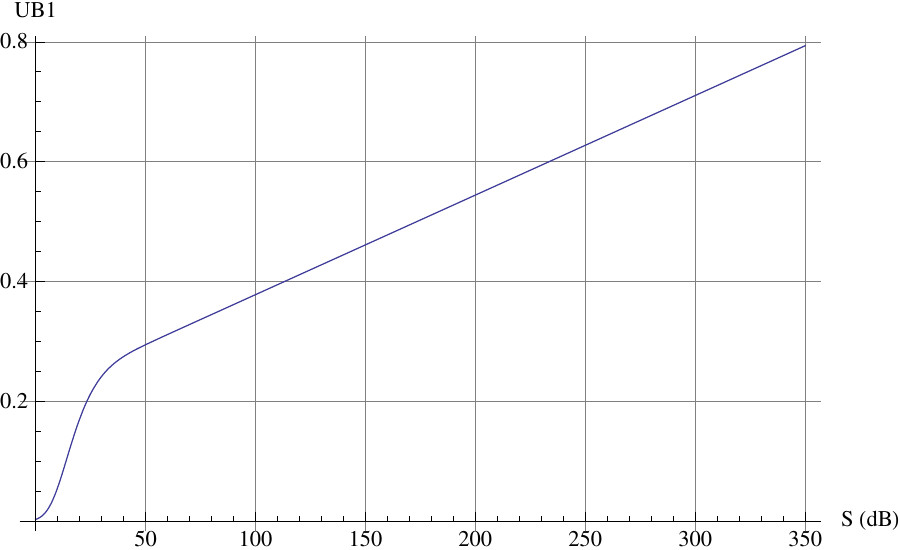}}
\caption{Mutual information upper bound (right-hand side of  (\ref{ziobeberu})) 
versus $\SNR = q\Pc_x$ (dB), for $q = 0.2$ and $p = 0.205$.} 
\label{ub1-highSNR}
\end{figure}

Next, we 
take a closer look at the behavior of the solutions of the fixed-point equation (\ref{e:fix-pointeq1}) -- (\ref{e:fix-pointeq}).
Even in the iid case (in which the equation reduces to \eqref{matched-mmse-fixed-point}) solved in \cite{guo-verdu,rangan-replica}, the question of how to choose among the multiple solutions has not been thoroughly addressed in the literature. 
 Fig.~\ref{mapf023}, \ref{mapf024} and \ref{mapf033} show
the fixed-point mapping function obtained by eliminating $\chi$ from (\ref{e:fix-pointeq1}) -- (\ref{e:fix-pointeq}), and given by 
\begin{eqnarray}
f(1/\eta) = \frac{1}{\mathcal{R}_{\Rm}\left(- {\sf mmse} \left ( V_0  | V_0 + \eta^{-\frac12}  Z \right )\right)},
\end{eqnarray}
given as a function of $1/\eta$, for $q = 0.2$ and $\SNR = 50$ dB. 
The intersections of this function with the main diagonal are the solutions of the equation $1/\eta = f(1/\eta)$.  We explore
the values of $p$ in the vicinity of the ``phase transition'' $p \approx 0.24$, for which the mutual information reaches a value very close to
$h(q)$ (corresponding to $D(p,q,\Pc_x) \approx 0$). 
For $p = 0.23$ (see Fig.~\ref{mapf023}) we have three solutions. Two are stable fixed points and one is an unstable fixed point.
The solution corresponding to the absolute minimum of the free energy $\Ic_1$ is the right-most fixed point (see Fig.~\ref{freeen023}), 
corresponding to a large value of $1/\eta$, which in turn translates into a large support recovery error rate, as we will see in Section \ref{sec:dist-results}. 
For $p = 0.24$ (see Fig.~\ref{mapf024}) we have also three solutions of which two are stable fixed points. 
However, now the solution corresponding to the absolute minimum of the free energy $\Ic_1$ is the left-most fixed point (see Fig.~\ref{freeen024}),  corresponding to a small value of $1/\eta$, i.e., to a very small support recovery error rate.
This ``jump'' from the right-most to the left-most stable fixed point corresponds to a phase transition of the underlying statistical physics system.
Notice that the phase transition may occur at finite SNR, as in  this case, and the phase transition threshold $p^*$ is, in general, strictly larger than
the noiseless perfect reconstruction threshold $q$. 
Finally, for values of $p$ significantly larger than the phase transition threshold  (see the example for 
$p = 0.33$ given in Fig.~\ref{mapf033}) only one solution exists. In this case, the free energy $\Ic_1$ 
has only one extremum point which is its absolute minimum (see Fig.~\ref{freeen033}). For the Gaussian 
iid sensing matrix case it is known (see \cite{reeves-journal1} and references therein) that the iterative algorithm known 
as  AMP-MMSE achieves the right-most fixed point of (\ref{e:fix-pointeq1}) -- (\ref{e:fix-pointeq}). 
This coincides with the  optimal MAP-SBS performance when this is the valid fixed point, corresponding to the minimum of $\Ic_1$. 
Instead, when there are multiple fixed points and the 
left-most fixed point is the valid one, the MAP-SBS estimator is strictly better than 
AMP-MMSE. 
Our results lead us to believe that the same behavior holds for a more general class of sensing matrices, as studied in 
this paper. From the examples above we notice that the right-most fixed point is the valid one for $p$ below the phase transition threshold. Above that threshold, either there is only one fixed point, for sufficiently large $p$, or one has to choose the solution that minimizes the
free energy.

\begin{figure}[ht]
	\centering
	\subfigure[Mapping function]{
	    \includegraphics[width=7cm]{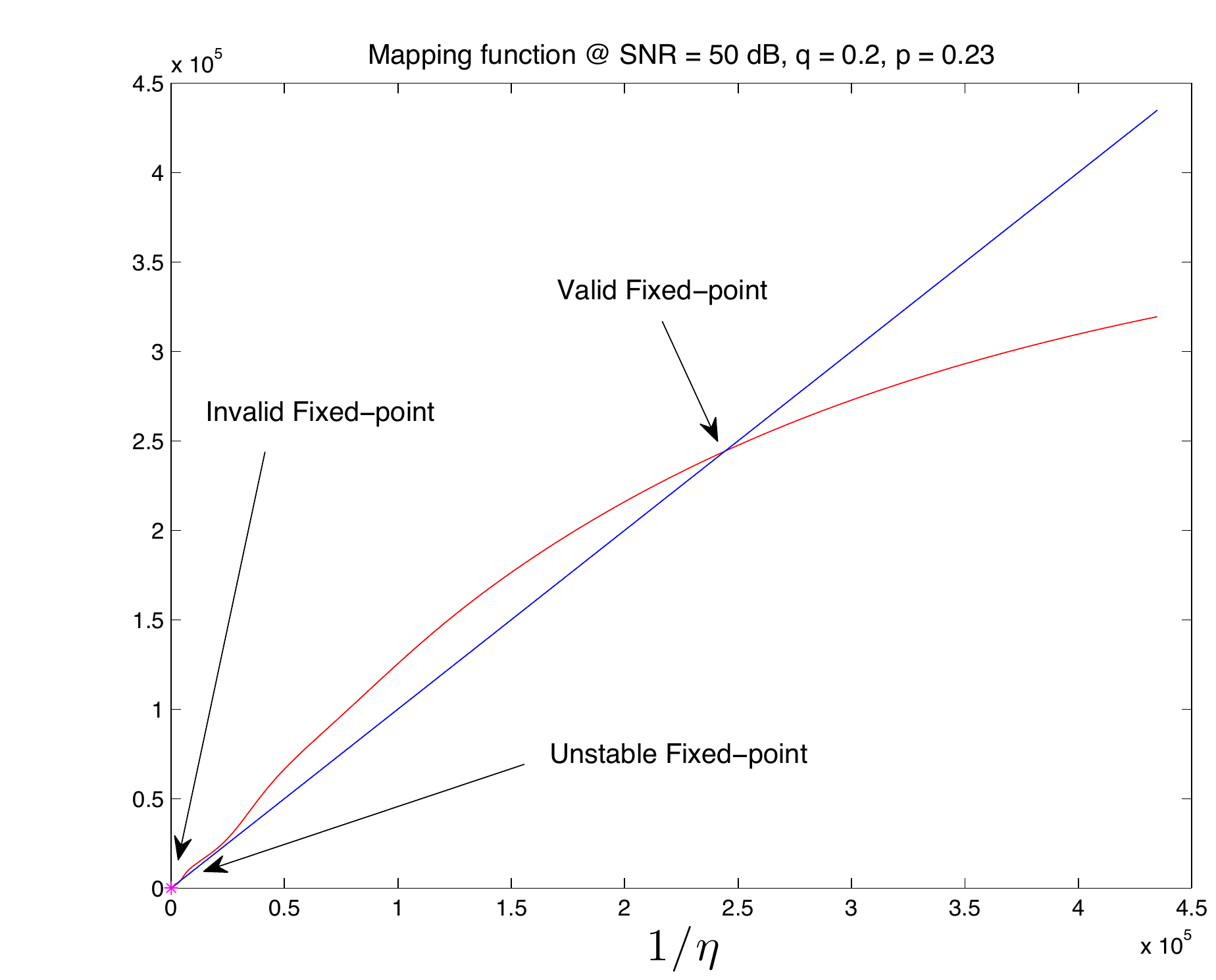}
	    \label{fig:subfig1023}
	}
	\subfigure[Detail near $1/\eta = 0$]{
	    \includegraphics[width=7cm]{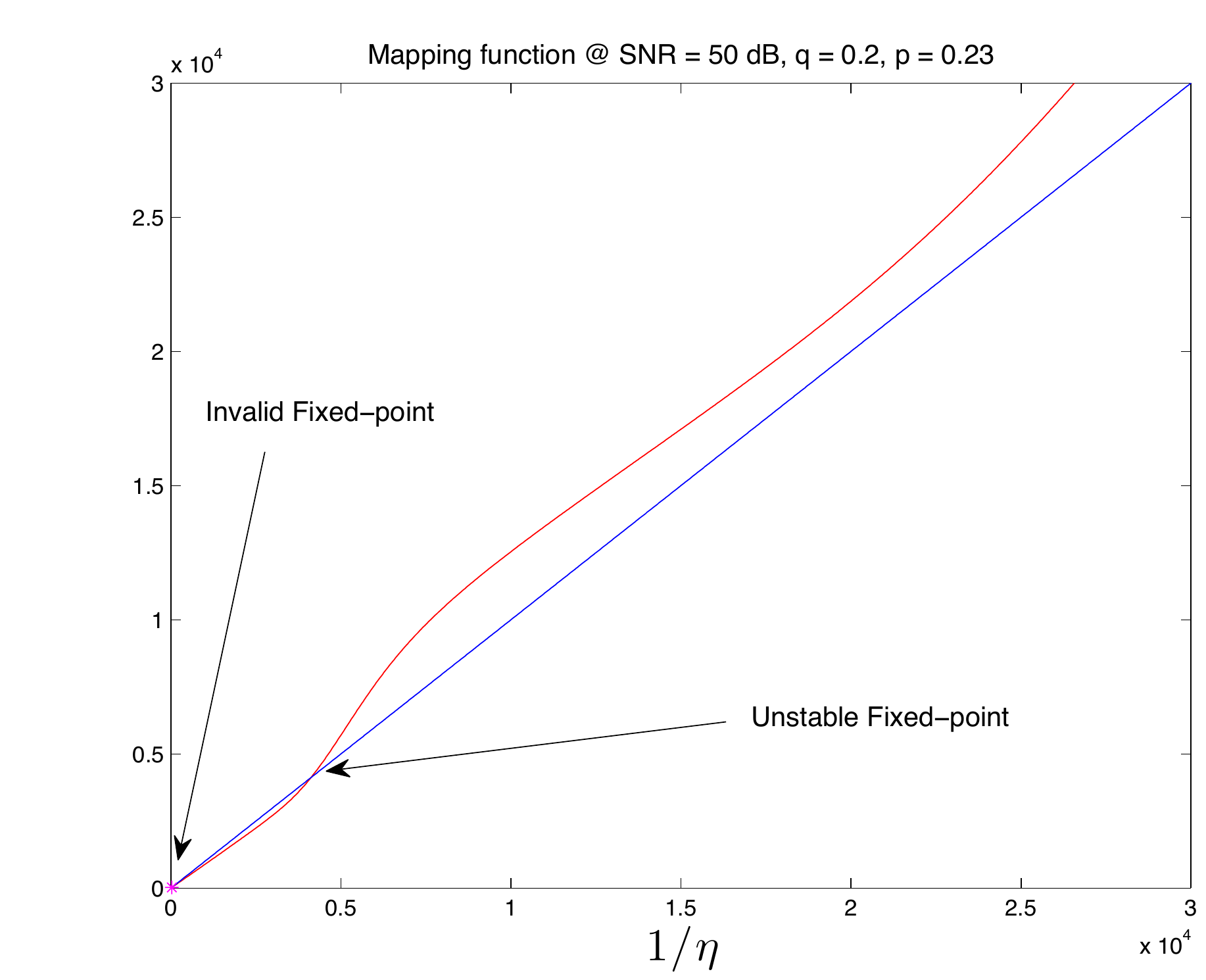}
	    \label{fig:subfig2023zoom}
	}
	\subfigure[Free energy]{
	    \includegraphics[width=7cm]{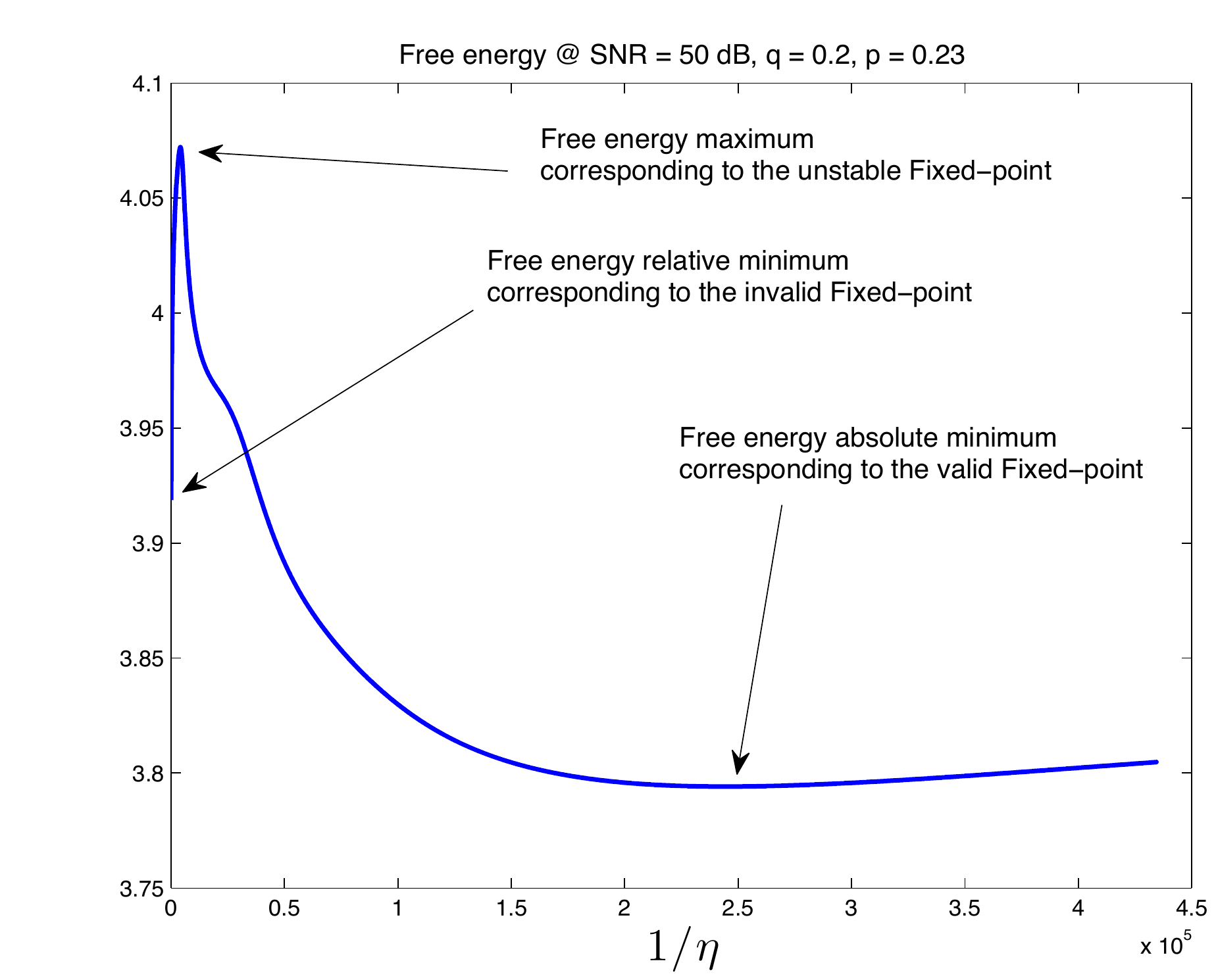}
	    \label{freeen023}
	}	
	\subfigure[Free energy detail near $1/\eta = 0$]{
	    \includegraphics[width=7cm]{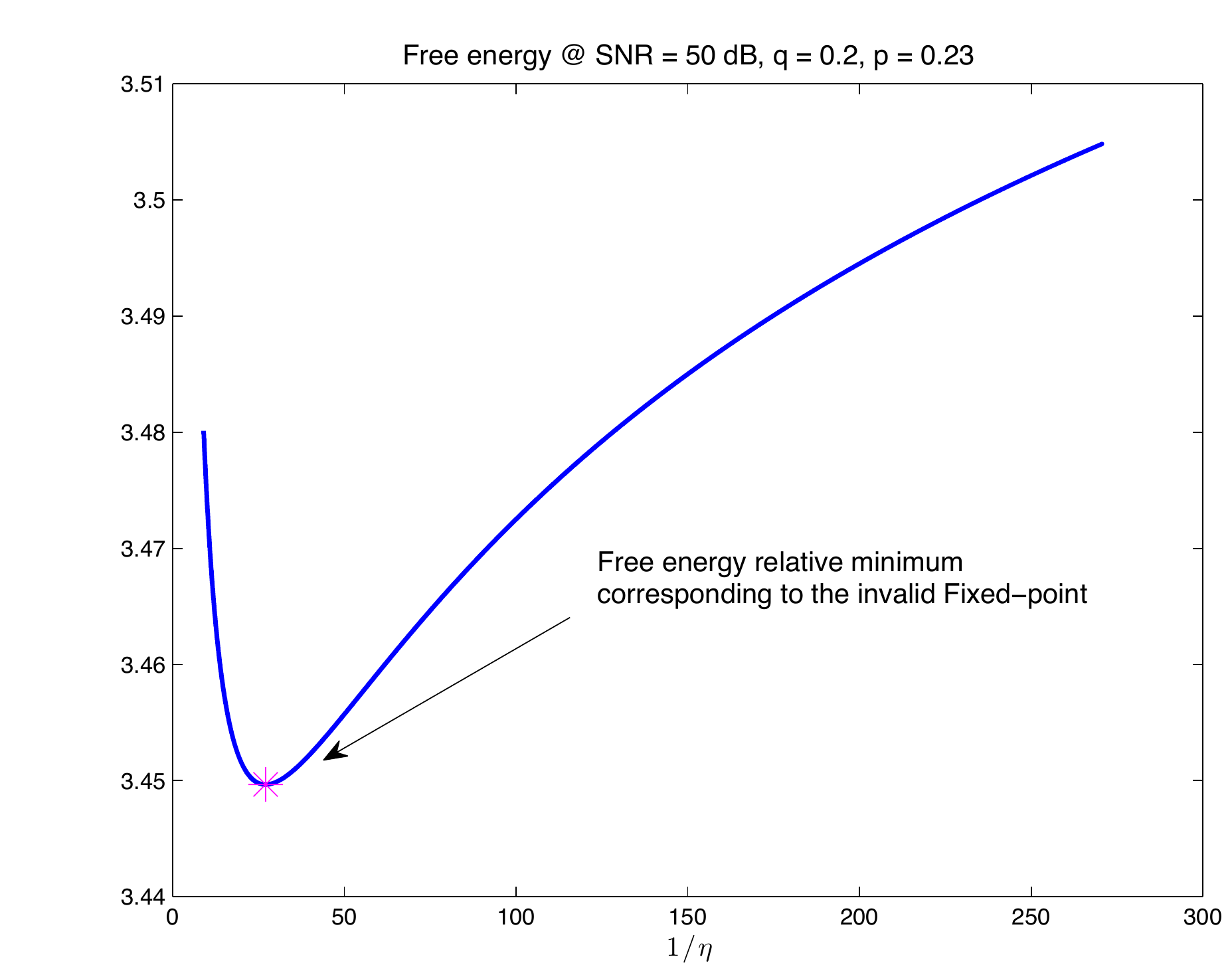}
	    \label{freeen023zoom}
	}	
	\caption{(a) Mapping function for the fixed-point equation (\ref{e:fix-pointeq1}) -- (\ref{e:fix-pointeq}) for $q = 0.2$, $p=0.23$
 and $\SNR = 50$ dB. (b) Detail in order to evidence the unstable fixed point and the left-most fixed point. (c) Corresponding free energy.
 (d) Detail of the free energy for small $1/\eta$ in order to show the minimum corresponding to the left-most fixed point.
	}	   
	\label{mapf023}
	\end{figure}

\begin{figure}[ht]
	\centering
	\subfigure[Mapping function]{
	    \includegraphics[width=7cm]{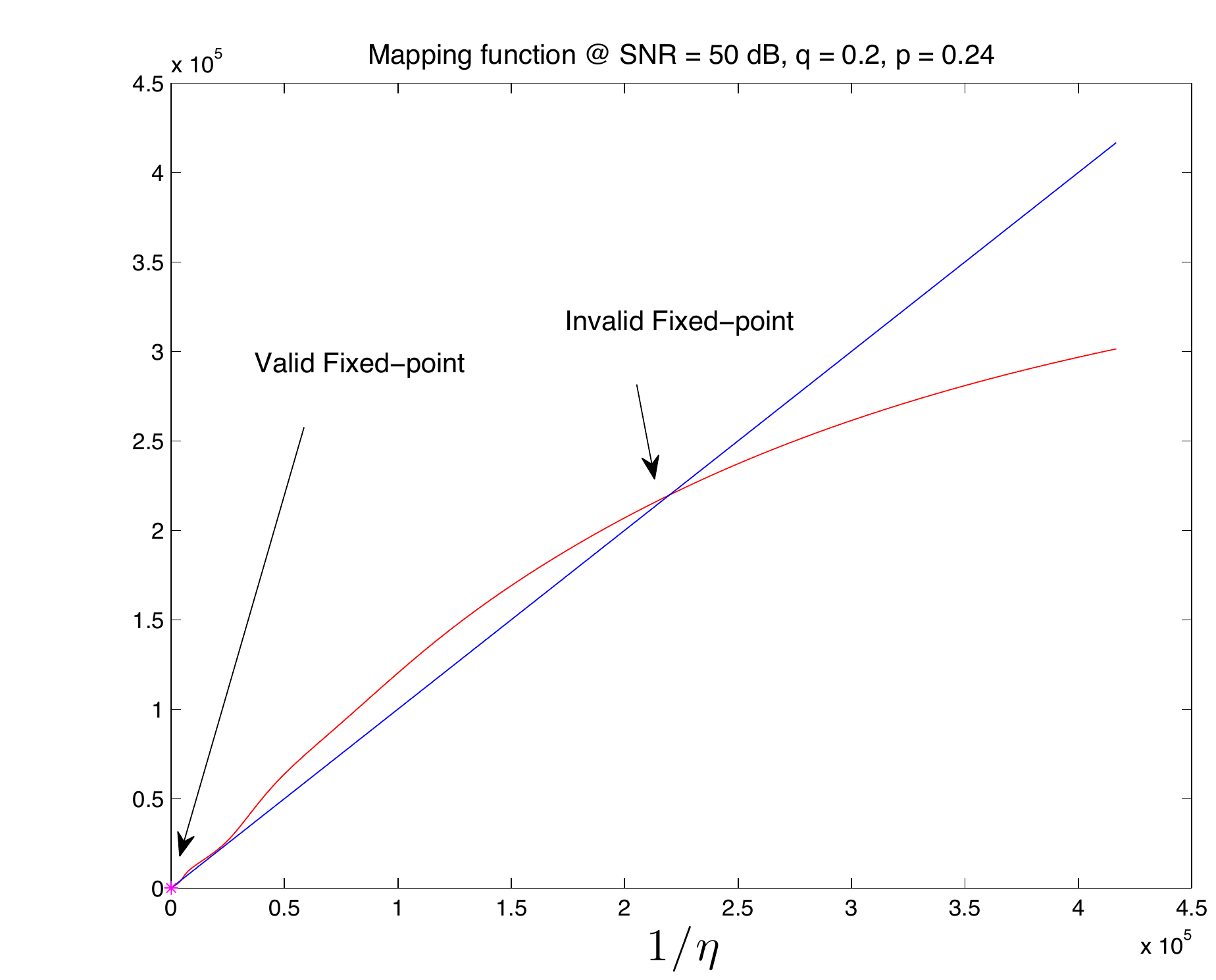}
	    \label{fig:subfig1024}
	}
	\subfigure[Detail near $1/\eta = 0$]{
	    \includegraphics[width=7cm]{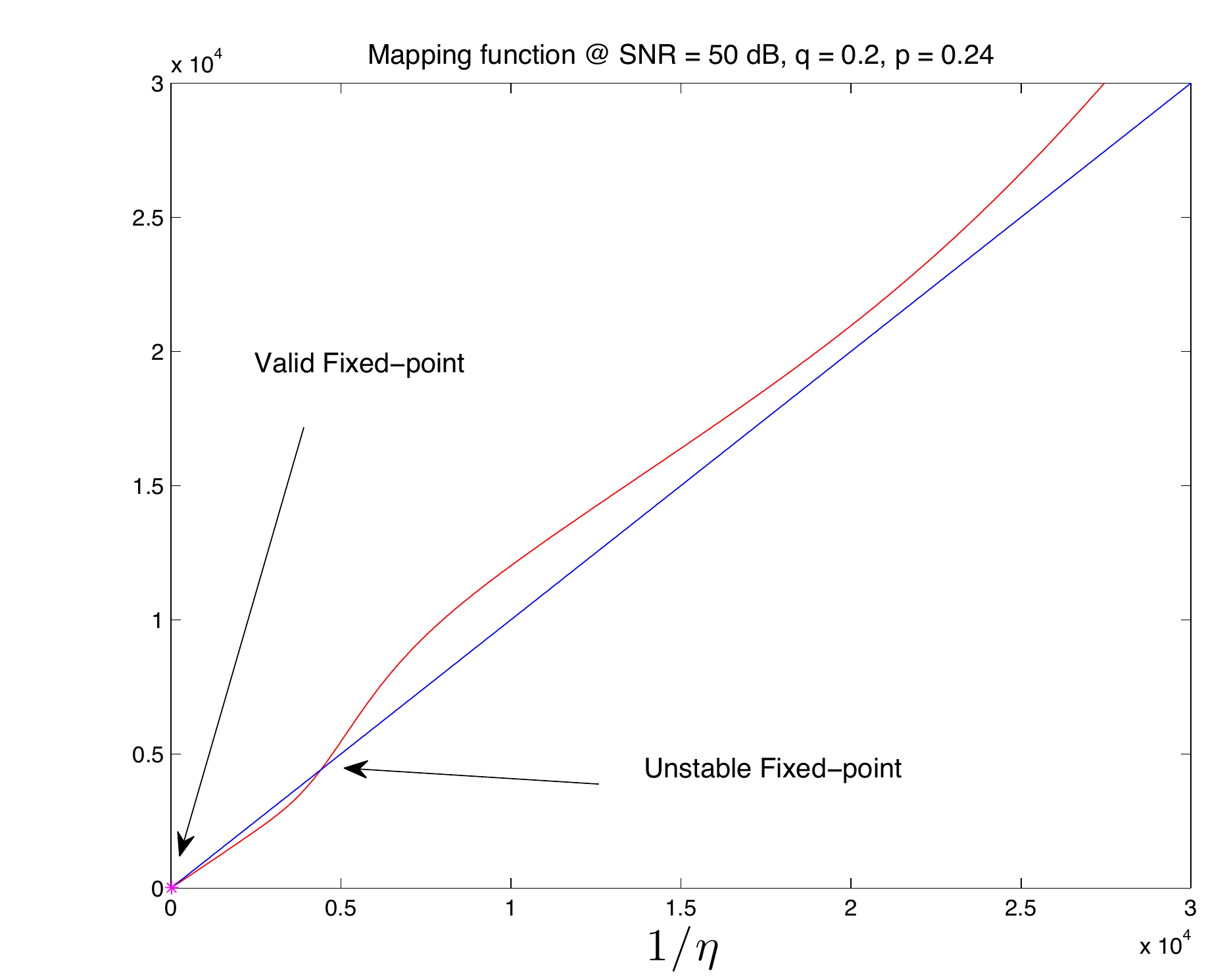}
	    \label{fig:subfig2024zoom}
	}
	\subfigure[Free energy]{
	    \includegraphics[width=7cm]{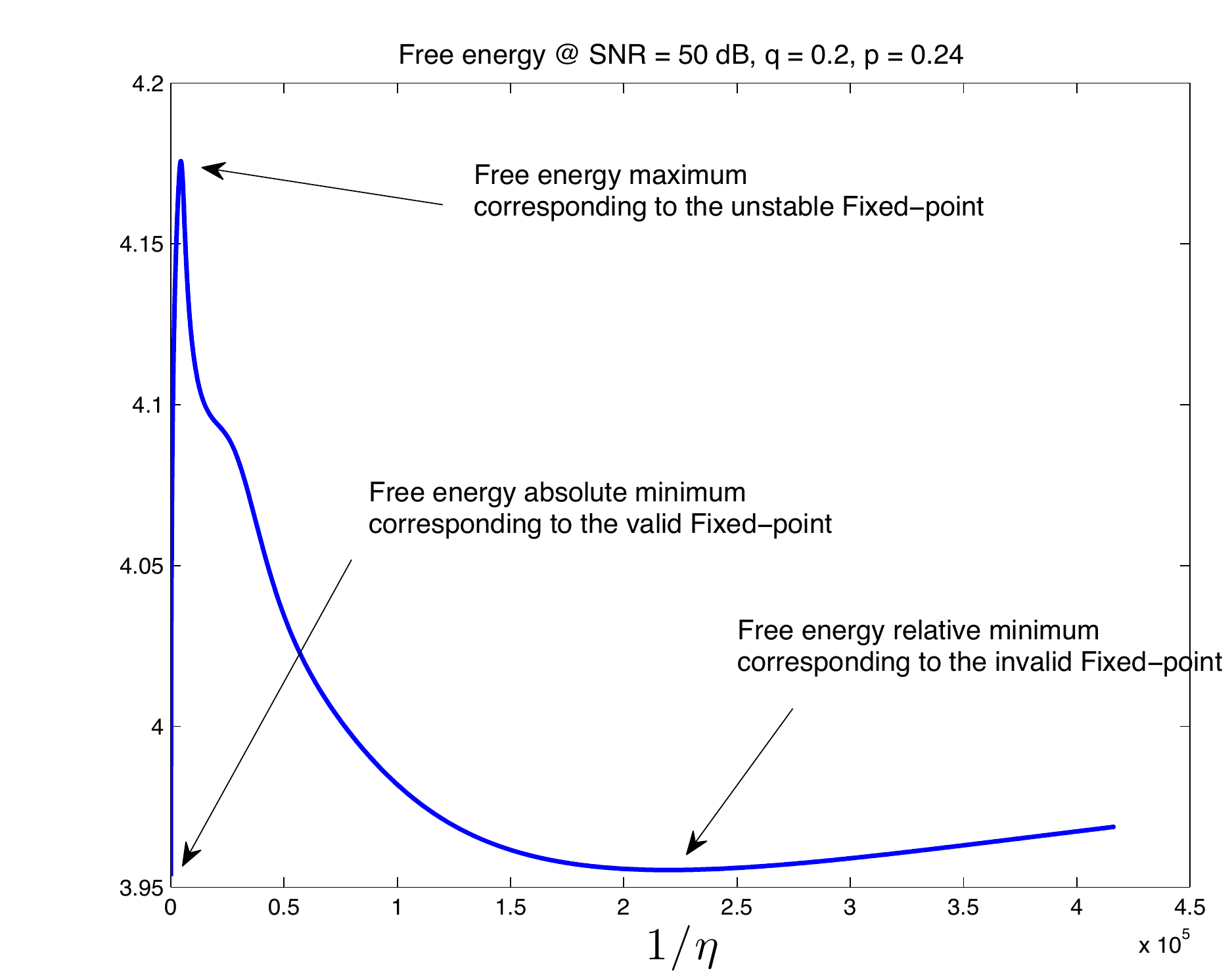}
	    \label{freeen024}
	}	
	\subfigure[Free energy detail near $1/\eta = 0$]{
	    \includegraphics[width=7cm]{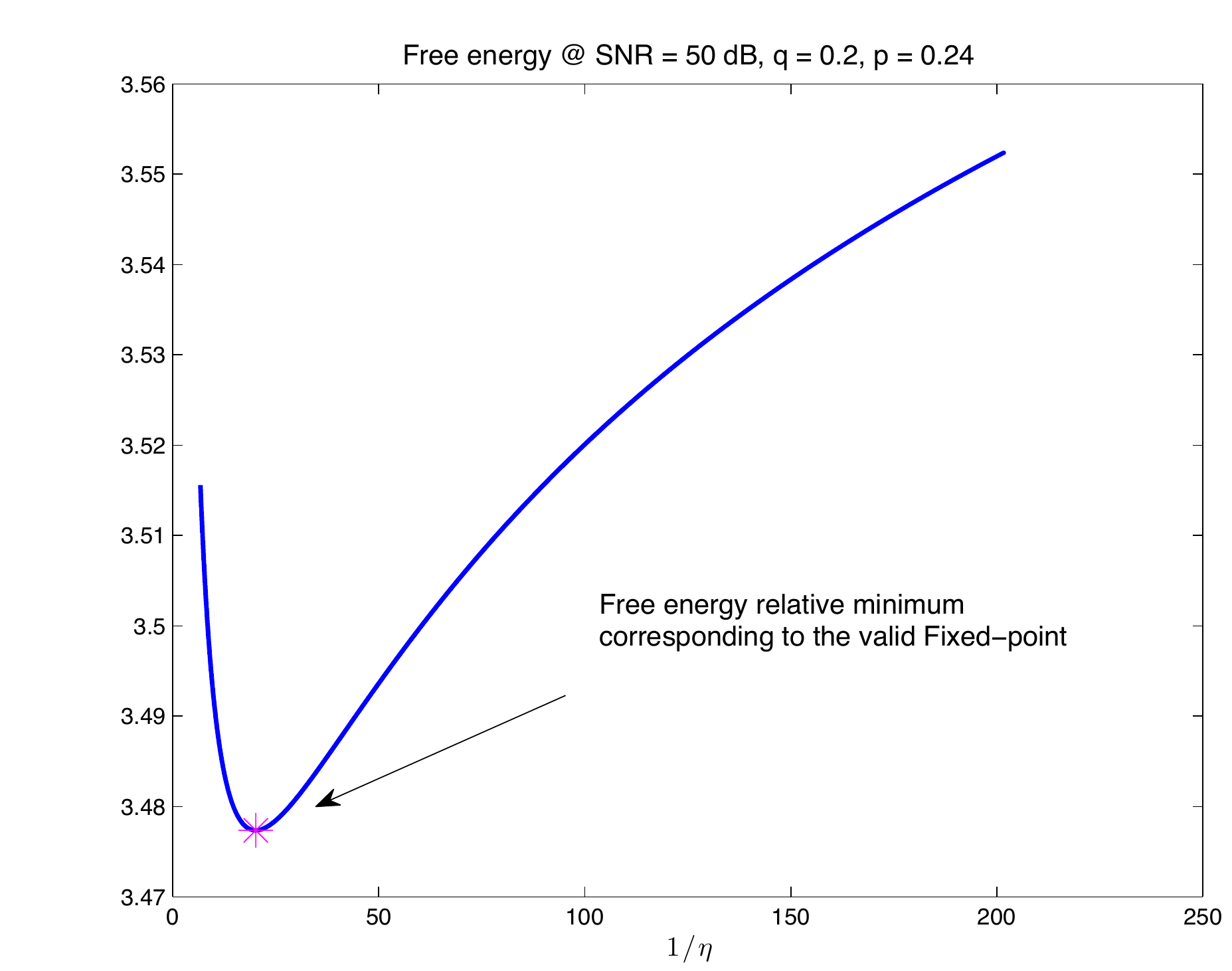}
	    \label{freeen024zoom}
	}	
	\caption{(a) Mapping function for the fixed-point equation (\ref{e:fix-pointeq1}) -- (\ref{e:fix-pointeq}) for $q = 0.2$, $p=0.24$
 and $\SNR = 50$ dB. (b) Detail in order to evidence the unstable fixed point and the left-most fixed point. (c) Corresponding free energy.
 (d) Detail of the free energy for small $1/\eta$ in order to show the minimum corresponding to the left-most fixed point. 
	}	   
	\label{mapf024}
	\end{figure}

\begin{figure}[ht]
	\centering
	\subfigure[Mapping function]{
	    \includegraphics[width=7cm]{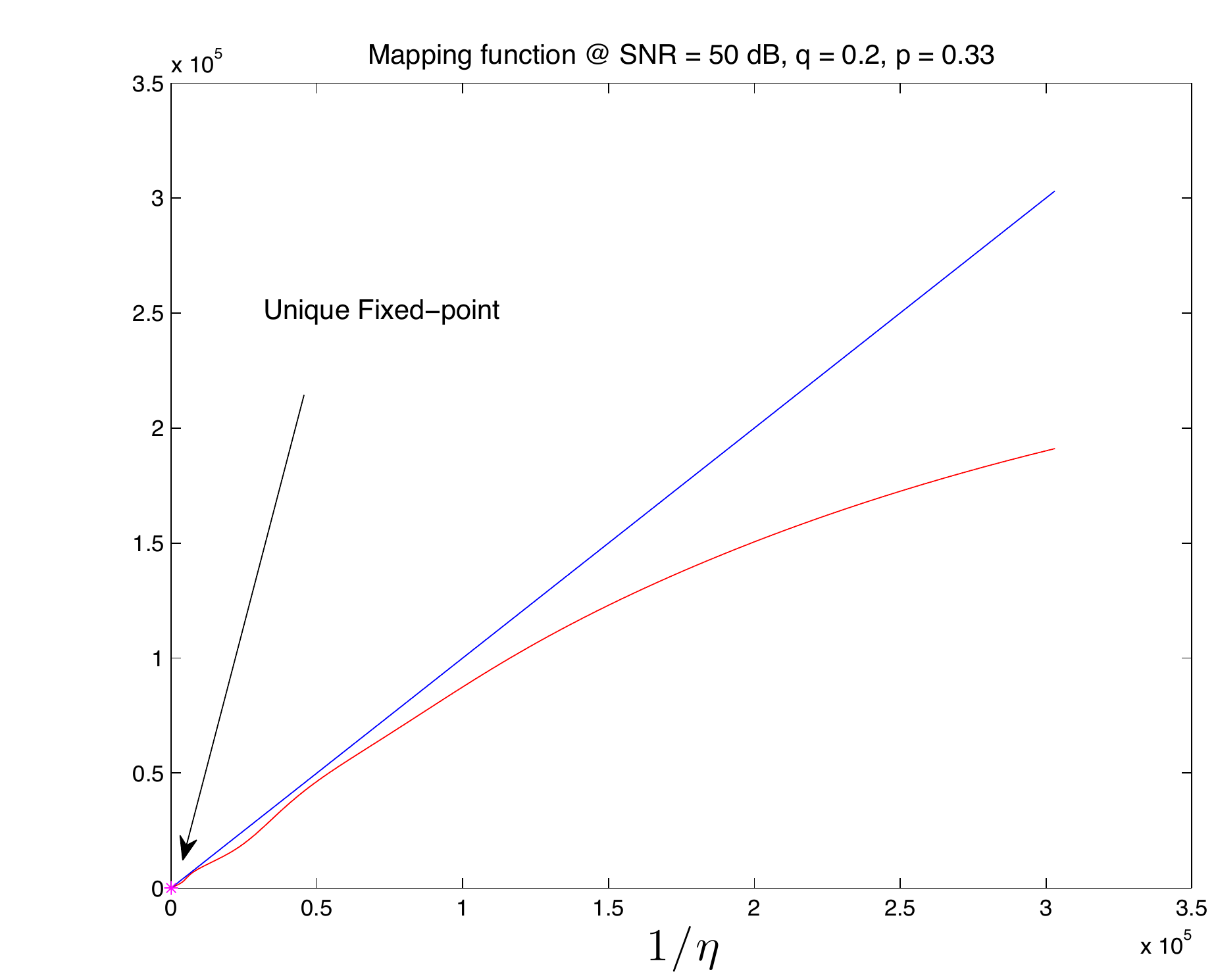}
	    \label{fig:subfig1033}
	}
	\subfigure[Detail near $1/\eta = 0$]{
	    \includegraphics[width=7cm]{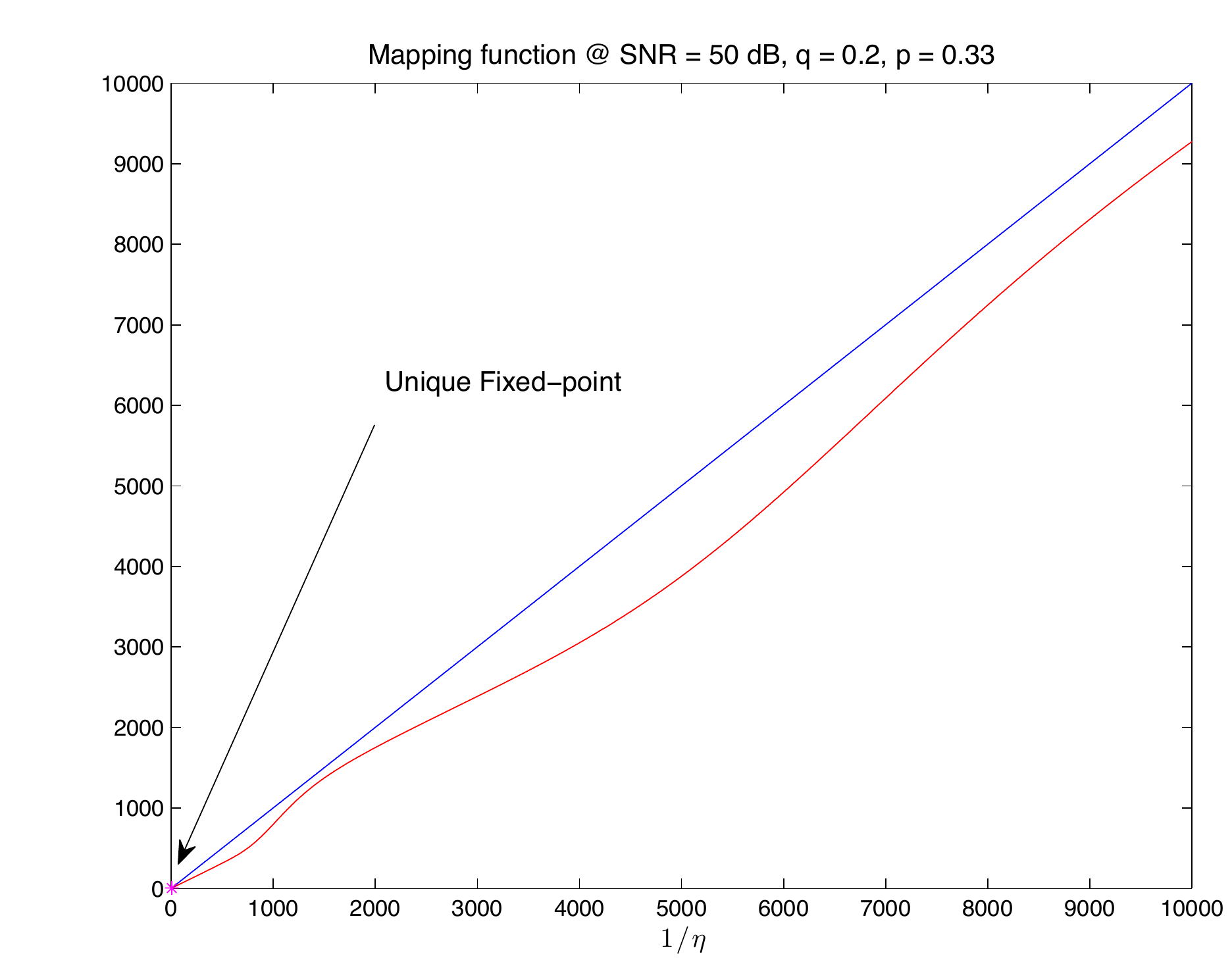}
	    \label{fig:subfig2033zoom}
	}
	\subfigure[Free energy]{
	    \includegraphics[width=7cm]{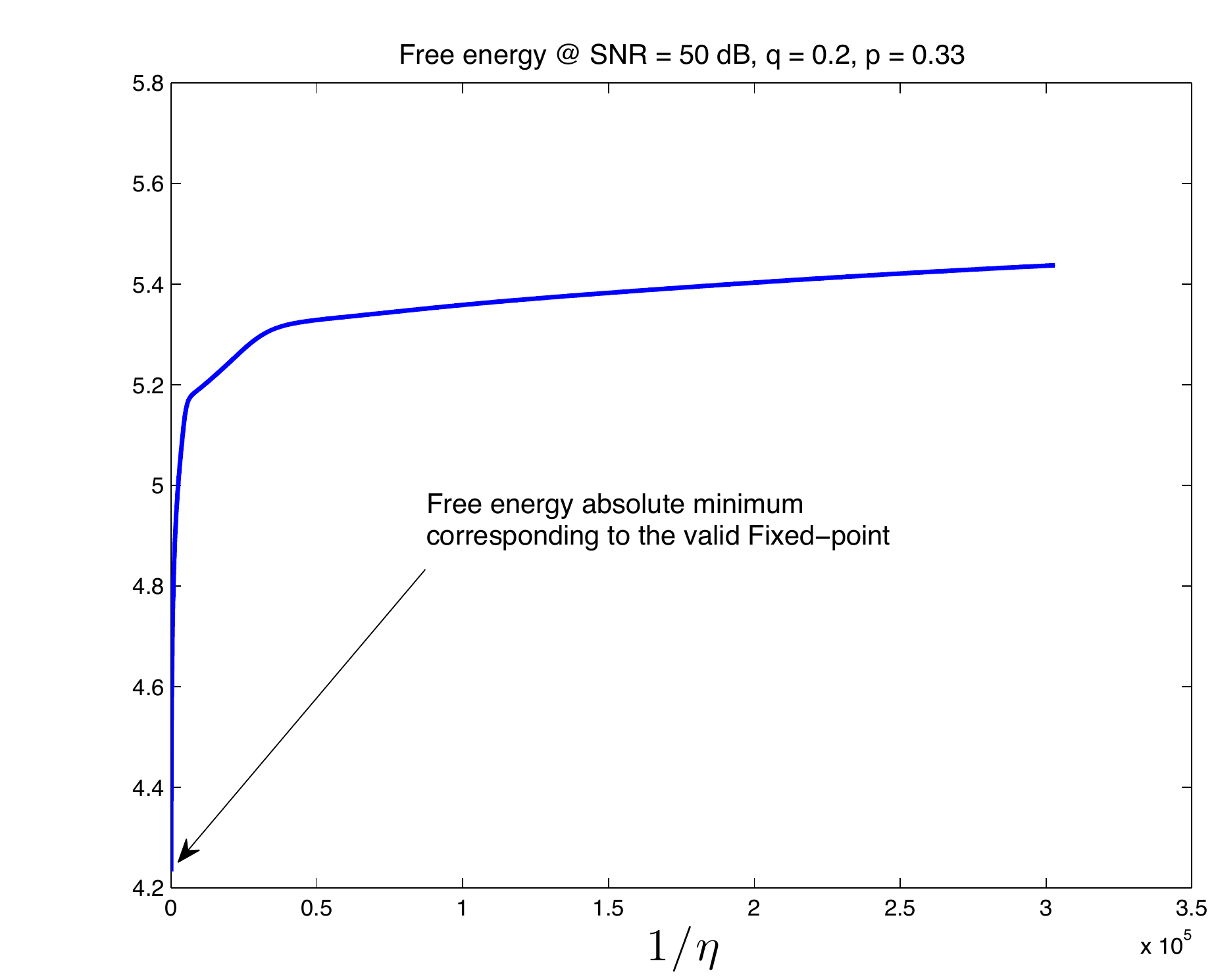}
	    \label{freeen033}
	}	
	\subfigure[Free energy detail near $1/\eta = 0$]{
	    \includegraphics[width=7cm]{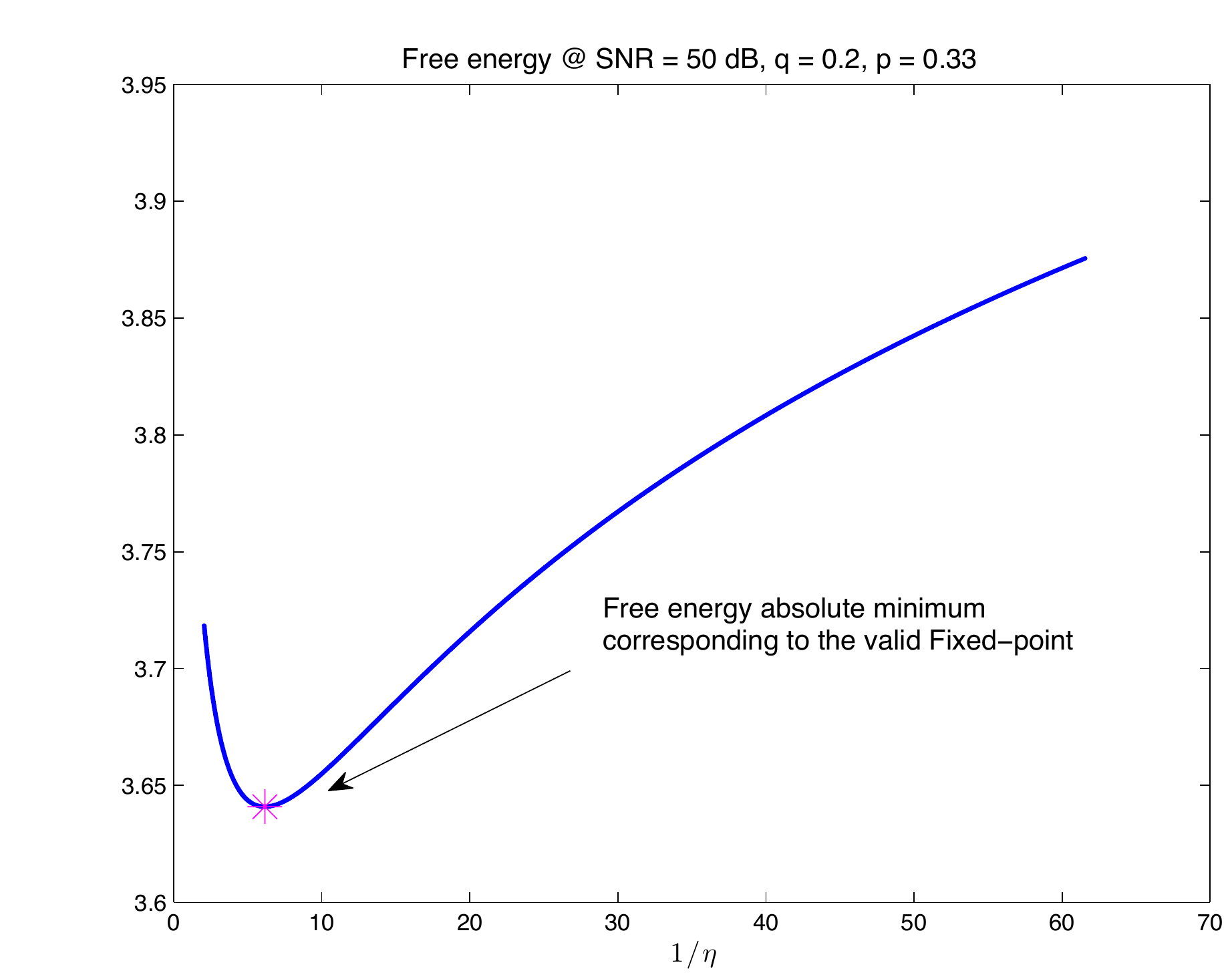}
	    \label{freeen033zoom}
	}	
	\caption{(a) Mapping function for the fixed-point equation (\ref{e:fix-pointeq1}) -- (\ref{e:fix-pointeq}) for $q = 0.2$, $p=0.33$
 and $\SNR = 50$ dB. (b) Corresponding free energy.
 (d) Detail of the free energy for small $1/\eta$ in order to show the minimum corresponding to the left-most fixed point.
	}	   
	\label{mapf033}
	\end{figure}

\section{Analysis of Estimators using the Decoupling Principle}\label{decoupling}

\subsection{Decoupling principle}  \label{decoup-sec}

The {\em decoupling principle} introduced 
by Guo and Verd\'u \cite{guo-verdu} states that 
the marginal joint distribution of each input coordinate and the corresponding estimator coordinate of a class of, 
possibly mismatched,  {\em posterior-mean estimators} (PMEs) converges, 
as the dimension grows,  to a fixed input-output joint distribution that corresponds to 
a ``decoupled'' (i.e., scalar) Gaussian observation model.  The observation model treated by Guo and Verd\'u in 
\cite{guo-verdu} is $\yv = \Sm \Gammam \xv + \zv$,  and the goal is to estimate $\xv$ from $\yv$, while knowing 
$\Sm$ and $\Gammam$, where $\xv$ is an $m \times 1$ iid vector with a given marginal distribution, $\zv$ is the iid Gaussian  noise vector,
$\Sm$ is a random $n \times m$ matrix with iid elements with mean zero and variance $1/n$, and
$\Gammam$ is an $m \times m$ diagonal matrix whose diagonal elements have an empirical distribution converging weakly 
to a given well-behaved distribution. 
Comparing the model of \cite{guo-verdu}  with (\ref{model1}), we notice that as far as the estimation of the 
Bernoulli-Gaussian iid vector $\vv = \Xm \bv$ the two models are similar,  by identifying $\Sm$ with $\Am \Um$,  
$\Gammam$ with $\Id$ and $\xv$ with $\vv$,  with the key difference that we allow a more general class of matrices satisfying 
the freeness condition given at the beginning of Section \ref{sc:setup}. In contrast, as far as the estimation of $\bv$ is concerned, 
our model differs from \cite{guo-verdu} in that in our case the diagonal iid Gaussian matrix $\Xm$ is not known to the estimator. 
 
In this section, we apply the decoupling principle to the estimation of 
$\bv$ for the observation model (\ref{model1}). This allows us to derive the minimum possible support recovery 
error rate for any estimator, achieved by the MAP-SBS estimator. 
The details of the derivations are given in Appendix \ref{proof:decoupling}, and the main results are summarized in the remainder of this section.
We also consider linear MMSE and Lasso \cite{tibshirani}, 
two popular estimators in the  compressed sensing literature. 
These estimators first produce an estimate of $\vv$ and then recover an estimate of the support $\bv$ by component wise thresholding. 
In order to analyze the suboptimal estimators, we resort to the decoupling principle for the estimation of $\vv$, which
can be derived along the same lines as Appendix \ref{proof:decoupling} or, equivalently, 
by extending the  analysis of  \cite{guo-verdu} to the class of sensing matrices considered in this paper. 
In \cite{reeves-journal1}, linear MMSE and Lasso estimators are studied for the case of 
iid sensing matrices as special cases of the Approximated Message Passing (AMP) algorithm \cite{donoho-montanary},  
the performance of which is rigorously characterized for $\Um$ with iid Gaussian entries in the 
large dimensional limit through the solution of a {\em state evolution equation} \cite{montanari-dense}.  
The current AMP rigorous analysis does not go through for the more general class of matrices considered here. 
Therefore, we resort to the {\em replica $+$ large deviation approach} of Rangan, Fletcher and Goyal \cite{rangan-replica} in order to 
obtain the decoupled model corresponding to these estimators. Interestingly, when particularizing our results to the iid case, 
we recover the  same AMP state evolution equations as given in \cite{reeves-journal1}.

For the sake of notation simplicity, we shall assume that all random variables and vectors appearing in the following formulas have
a density (possibly including Dirac distributions), indicated by $p$ with the appropriate subscripts and arguments. 
In order to limit the proliferation of symbols, we use the same symbols to indicate random variables (or vectors) and 
the corresponding dummy arguments in the probability distributions. 

The class of estimators for which the decoupling principle holds are {\em mismatched} PMEs where the mismatch is reflected in 
an assumed channel transition probability and symbol a priori probabilities that may not correspond to the actual ones.
We shall reserve the letter $q$ with the appropriate subscripts and arguments to indicate these assumed distributions.
The true conditional channel transition probability of $\yv$ given $\bv, \Am, \Um, \Xm$ of (\ref{model1}) is given by (\ref{transition-model1}). 
The corresponding assumed channel transition probability is given by 
\begin{eqnarray}  \label{assumed-multivariate-gaussian}
q_{\yv|\bv, \Am, \Um, \Xm}(\yv|\bv, \Am, \Um, \Xm) = \left ( \frac{\gamma}{\pi} \right )^n \exp \left ( - \gamma \left \| \yv - \Am\Um\Xm\bv \right \|^2 \right ),
\end{eqnarray}
where the assumed noise variance is $1/\gamma$ instead of 1. We let also $q_{\bv}(\bv) = \prod_{i=1}^n q_b(b_i)$ denote an assumed 
a-priori distribution for $\bv$, not necessarily Bernoulli-$q$.
The mismatched estimator for $\bv$ given $\yv, \Am, \Um$ is given by 
The corresponding PME takes on the form
\begin{eqnarray} \label{PME}
\widehat{\bv}(\yv, \Am, \Um) = \int \bv \; q_{\bv|\yv, \Am, \Um} (\bv|\yv, \Am, \Um) \; d\bv,
\end{eqnarray}
where
\begin{eqnarray} \label{posterior-b}
q_{\bv|\yv, \Am, \Um} (\bv|\yv, \Am, \Um)
= \frac{\int q_{\yv|\bv, \Am, \Um, \Xm}(\yv|\bv, \Am, \Um, \Xm) q_{\bv}(\bv) p_{\xv}(\xv) d\xv}
{\int q_{\yv|\bv, \Am, \Um, \Xm}(\yv|\bv', \Am, \Um, \Xm) q_{\bv}(\bv') p_{\xv}(\xv) d\xv d\bv'},
\end{eqnarray}
and where $p_{\xv}(\xv) = \frac{1}{(\pi \Pc_x)^n} \exp(-\|\xv\|^2/\Pc_x)$ is the $n$-variate iid Complex Gaussian density with components
$\sim \Cc\Nc(0,\Pc_x)$. 

In the {\em matched} case, for $\gamma = 1$ and $q_{\bv}(\bv) \equiv$ Bernoulli-$q$, 
(\ref{PME}) coincides  with the MMSE estimator.~\footnote{This is the PME for the matched statistics, which effectively
minimizes the MSE.}  
By considering general $\gamma$ and $q_{\bv}(\bv)$, we can study of a whole 
family of mismatched 
PMEs through the same unified framework \cite{Tanaka,guo-verdu}.

For the purpose of analysis,  it is convenient to define a virtual multivariate observation model
involving the random vectors $\bv_0 \sim p_{\bv_0}(\bv_0)$, Bernoulli-$q$, the corresponding observation channel output
$\yv = \Am \Um \Xm \bv_0 + \zv$ as in (\ref{model1}), and an intermediate vector $\bv \sim q_{\bv}(\bv)$, not corresponding to 
any physical quantity present in the original model, such that the conditional joint distribution of 
$\bv_0, \yv, \bv$ given $\Am,\Um$ is given by 
\begin{eqnarray} \label{joint-n-variate}
p_{\bv_0}(\bv_0) \ p_{\yv|\bv_0, \Am,\Um}(\yv|\bv_0,\Am,\Um) \ q_{\bv|\yv, \Am,\Um}(\bv|\yv,\Am,\Um),
\end{eqnarray}
with
\begin{eqnarray}  \label{true-multivariate-gaussian}
p_{\yv|\bv_0, \Am, \Um, \Xm}(\yv|\bv_0, \Am, \Um ,\Xm) = \frac{1}{\pi^n} \exp \left ( - \left \| \yv - \Am\Um\Xm\bv_0 \right \|^2 \right ).
\end{eqnarray}
Then, $\widehat{\bv}(\yv, \Am, \Um)$ can be seen as the ``matched'' PME of $\bv$ given $\yv$ with respect to the 
joint probability distribution (\ref{joint-n-variate}). Notice also that (\ref{joint-n-variate}) satisfies the conditional Markov Chain
$\bv_0 \rightarrow \yv \rightarrow \bv$, for given $\Am,\Um$.

The decoupling principle obtained in this paper and proved in Appendix \ref{proof:decoupling}
can be stated as follows.  Let  $(b_{0i}, b_i, \widehat{b}_i)$ denote the $i$-th components of the random vectors
$\bv_0, \bv, \widehat{\bv}(\yv, \Am, \Um)$, obeying the joint 
conditional distribution (\ref{joint-n-variate}) with $\widehat{\bv}(\yv, \Am, \Um)$ given in (\ref{PME}).
Then, in the limit of $n \rightarrow \infty$, under the assumption that the replica-symmetric 
analysis holds (see Appendix \ref{proof:decoupling}),
the joint distribution of $(b_{0i}, b_i, \widehat{b}_i)$ converges to the joint distribution 
of the triple $(B_0, B, \widehat{B})$ induced by
\begin{eqnarray} \label{joint-decoupled}
p_{B_0}(b_0) \ p_{Y|B_0; \eta}(y|b_0) \ q_{B|Y;\xi}(b|y),
\end{eqnarray}
and by $\widehat{B} = \int b \; q_{B|Y;\xi}(b|y) db$, where we define the decoupled channel
\begin{eqnarray} \label{decoupled-channel} 
Y = V_0 + \eta^{-\frac{1}{2}} Z, 
\end{eqnarray}
with $Z \sim \Cc\Nc(0,1)$ and $V_0 = X_0B_0$, with $B_0 \sim p_{B_0}(b_0)$, Bernoulli-$q$, and with
$X_0 \sim \Cc\Nc(0,\Pc_x)$, and where $X_0, B_0$ and $Z$ are mutually independent. 
Also, we define $V = X B$ with  $X \sim \Cc\Nc(0,\Pc_x)$ and $B \sim q_B(b)$ identically distributed as the 
the marginals of the assumed prior distribution $q_{\bv}(\bv)$. We let $p_X(\cdot)$ denote the common
density of $X_0$ and $X$, and define the following probability densities for the variables
$V_0, Y, V, B_0$ and $B$:
\begin{eqnarray}
p_{Y|V_0; \eta}(y|v_0)  & = & \frac{\eta}{\pi} \exp \left ( - \eta \left | y - v_0 \right |^2 \right ) \label{true-yv} \\
p_{Y|B_0; \eta}(y|b_0)  & = & \int p_{Y|V_0; \eta}(y|x_0 b_0) p_{X}(x_0) dx_0 \label{true-yb0} \\
q_{Y|V; \xi}(y|v)  & = & \frac{\xi}{\pi} \exp \left ( - \xi \left | y - v \right |^2 \right ) \label{fake-yv} \\
q_{Y|B; \xi}(y|b) & = & \int q_{Y|V; \xi}(y|xb) p_{X}(x) dx \label{fake-yb} \\
q_{B|Y;\xi}(b|y) & = & \frac{q_{Y|B; \xi}(y|b) q_B(b)}{\int q_{Y|B; \xi}(y|b') q_B(b') db'},
\end{eqnarray}
where the parameters $\eta$ and $\xi$ are obtained by  solving the system of fixed-point equations\footnote{We use
the dot notation $\dot{f}(x)$ to denote the first derivative  of a single-variate function $f$ with respect to its argument.}
\begin{subequations}
\begin{eqnarray}
\chi & = & \gamma \, \mmse( V | Y) \label{saddle-final-chi} \\
\delta & = & \EE \left [ |V_0 - \EE[V|Y] |^2 \right ] \label{saddle-final-delta} \\
\xi & = & \gamma\,  \Rc_{\Rm}(-\chi) \label{saddle-final-xi} \\
\eta & = & \frac{(\xi/\gamma)^2}{\xi/\gamma  + \dot{\Rc}_{\Rm}(- \chi) (\delta - \chi)}. \label{saddle-final-eta}
\end{eqnarray}
\end{subequations}
The expectations in (\ref{saddle-final-chi}) -- (\ref{saddle-final-eta}) are defined  with respect to the joint distribution of $V_0, Y, V$ given by 
\begin{eqnarray} \label{joint-decoupled-V}
p_{V_0}(v_0) \ p_{Y|V_0; \eta}(y|v_0) \ q_{V|Y;\xi}(v|y),
\end{eqnarray}
where $p_{V_0}(v_0)$ is the Bernoulli-Gaussian distribution of $V_0 = X_0B_0$,
$p_{Y|V_0; \eta}(y|v_0) $ is given in (\ref{true-yv}) and where
\begin{eqnarray}
q_{V|Y;\xi}(v|y) = \frac{q_{Y|V; \xi}(y|v) q_V(v)}{q_{Y;\xi}(y)},
\end{eqnarray}
with $q_{Y|V; \xi}(y|v)$ given in (\ref{fake-yv}), $q_V(v)$ is the distribution of $V = XB$, and
\begin{eqnarray}  \label{qY-ziobono}
q_{Y; \xi}(y)  = \int q_{Y|V;\xi}(y|v) q_V(v) dv.
\end{eqnarray}
In passing, notice also that (\ref{joint-decoupled}) and (\ref{joint-decoupled-V})  satisfy the
Markov Chains $B_0 \rightarrow Y \rightarrow B$ and $V_0 \rightarrow Y \rightarrow V$, 
respectively.

If the solution to (\ref{saddle-final-chi}) -- (\ref{saddle-final-eta})  is not unique, then
we have to select the solution that {\em minimizes} the system ``free energy'' (expressed in nats):
\begin{align} \label{free-energy-babau-final}
\Ec  & = \log\frac{\xi}{\gamma}  - \frac{\xi}{\eta} + \gamma - \xi\chi +
\left (\frac{\xi}{\eta}  -1 \right )  \frac{\xi \chi}{\gamma}  +  \int_0^{\chi} \Rc_{\Rm}(-w)dw  - \EE \left [  \log \left ( q_{Y; \xi}(Y)  \right ) \right ]. 
\end{align}
As expected, by letting $\gamma = 1$ and 
$q_B(b)$ Bernoulli-$q$  we obtain $\xi = \eta$ and $\delta = \chi$ and (\ref{saddle-final-chi}) -- (\ref{saddle-final-eta}) reduce 
to (\ref{e:fix-pointeq1}) -- (\ref{e:fix-pointeq}). It is also immediate to see that in this case we have 
$\Ec = \Ic_1 + \log(\pi e)$ where $\Ic_1$ is given  in (\ref{1l}). 

By particularizing our analysis to the case of $\Um$ with iid elements, 
using (\ref{R-transfromGV}), we obtain the simpler fixed-point equations 
\begin{subequations}
\begin{eqnarray} 
\frac{1}{\eta} & = & \frac{1}{p} \left ( 1 + \EE \left [ |V_0 - \EE[V|Y] |^2 \right ] \right ) \label{saddle-final-eta-iid} \\
\frac{1}{\xi} & = & \frac{1}{p} \left ( \frac{1}{\gamma} + \mmse( V | Y) \right ),  \label{saddle-final-xi-iid}
\end{eqnarray}
\end{subequations}
which recovers the results of \cite{guo-verdu,rangan-replica,guo-baron-shamai} up to a different 
normalization as discussed in the first example of Section \ref{section:specialcases}.

\subsection{Symbol-by-symbol MAP estimator}  \label{sec:map-sbs}

As an application of the decoupling principle, we can determine the minimum achievable $D(p,q,\Pc_x)$
by particularizing the above formulas for the MAP-SBS estimator of $b_{i}$ given 
$\yv, \Am,\Um$, operating according to the optimal decision rule
\begin{eqnarray} \label{MAP}
\widehat{b}_{i}(\yv, \Am, \Um) = \arg \max_{b \in \{0,1\}} \PP [ b_{i} = b | \yv, \Am, \Um ].
\end{eqnarray}
It is well-known that the MAP-SBS minimizes the support recovery error rate over all possible estimators.
A byproduct of the decoupling principle is that, in the matched case, (\ref{joint-decoupled}) yields immediately
that the limiting posterior marginal $\PP [ b_{i} = b | \yv, \Am, \Um ]$ for a randomly chosen $i$-th component of $\bv$ 
is given by $p_{B_0|Y; \eta}(b_0|y)$, the posterior distribution of the decoupled channel (\ref{decoupled-channel}), 
marginalized with respect to $B_0$. In the matched case, (\ref{saddle-final-chi}) -- (\ref{saddle-final-eta})
reduce to (\ref{e:fix-pointeq1}) -- (\ref{e:fix-pointeq}) in Theorem \ref{th:x}, and
$p_{B_0|Y; \eta}(b_0|y)$ is easily obtained by 
noticing that
%
$Y$ given $B_0$ is conditionally distributed as
\begin{eqnarray} \label{Y-given-B_0}
p_{Y|B_0;\eta}(y|b_0) = \frac{1}{\pi(\Pc_x |B_0|^2 + 1/\eta)} \exp\left ( - \frac{|y|^2}{\Pc_x |B_0|^2 + 1/\eta} \right ),
\end{eqnarray}
i.e.,  $Y \sim \Cc\Nc\left(0, \Pc_x +1/\eta\right)$ for $B_0 = 1$ and $Y \sim \Cc\Nc\left(0,1/\eta\right)$ for $B_0 = 0$. Then,
\begin{eqnarray}  \label{PB0=1|Y}
\PP[B_0 = 1| Y = y] = \frac{1}{1+ \frac{1 - q}{q} ( 1 + \eta \Pc_x ) \exp\left(- \eta \Pc_x \mu  |y|^2 \right)},
\end{eqnarray}
(obviously $\PP[B_0 = 0|Y=y] = 1 - \PP[B_0 = 1|Y=y]$)
where $\eta$ is obtained from (\ref{e:fix-pointeq1}) -- (\ref{e:fix-pointeq}) and where we define:
\begin{eqnarray} \label{mudef}
\mu = \frac{\eta}{1 + \Pc_x\eta}.
\end{eqnarray}
The resulting MAP-SBS estimator is
\begin{eqnarray} \label{map-decoupled} 
\widehat{B}(y) = \arg\max_{b_0 \in \{0,1\}} \; \PP[B_0 = b_0 | Y= y], 
\end{eqnarray}
with decision $\widehat{B}(y) = 1$ if
\begin{eqnarray}
\frac{1 - q}{q} ( 1 + \eta \Pc_x ) \exp \left(- \eta \Pc_x \mu |y|^2 \right) < 1
\end{eqnarray}
(with randomization on the boundary). Taking the logarithm of both sides, we find the ``energy detector'' 
(analogous to non-coherent on-off modulation with fading) given by
\begin{eqnarray} \label{energy-detect}
\widehat{B}(y) = \left \{
\begin{array}{ll}
1, & \mbox{for} \; |y|^2 \geq  \tau \\
0, & \mbox{elsewhere}
\end{array}
\right .
\end{eqnarray}
with
\begin{eqnarray}
\tau = \frac{1}{\eta \Pc_x \mu} \log  \frac{(1 - q)(1 + \eta\Pc_x)}{q}.
\end{eqnarray}
We have $\widehat{B}(y) = 1$, regardless of the value of $y \in \CC$,
if  $q > \frac{1 + \eta \Pc_x }{2 + \eta \Pc_x }$, in which case $D(p,q,\Pc_x) = 1 - q$. 
Otherwise,
\begin{eqnarray} \label{distotionexcat}
D(p,q,\Pc_x) =  q \left ( 1 - \exp \left (-\mu \tau \right ) \right ) + (1 - q) \exp
\left ( - \eta \tau \right ),
\end{eqnarray}
obtained from (\ref{energy-detect}) by observing that $|Y|^2$, conditioned on $B_0$, is central chi-square 
with two degrees of freedom with mean $\Pc_x + 1/\eta$ for $B_0 = 1$ 
and with mean $1/\eta$ for $B_0 = 0$.

For $\Um$ with iid elements, we can recover known results. In this case, 
(\ref{e:fix-pointeq1}) -- (\ref{e:fix-pointeq}) reduce to  (\ref{matched-mmse-fixed-point}), which corresponds
to the replica analysis of the MMSE estimator obtained in \cite{guo-verdu} and summarized in
\cite{reeves-journal1} in the context of support recovery in compressed sensing. 
When the iterative solution of the fixed-point equation (\ref{matched-mmse-fixed-point}) is initialized 
by $1/\eta = (1 + q \Pc_x)/p$,  then the iteration converges to the solution of  the so-called ``AMP-MMSE''  state equation given in 
\cite[Th. 6]{reeves-journal1}. In brief, by this initialization the iterative solution converges always to the right-most
fixed point of the mapping function (see Figs. \ref{mapf023} -- \ref{mapf033} and related discussion).  
Instead, if the {\em valid fixed-point} is chosen, i.e., the solution which minimizes the free energy 
$\Ic_1$, then  we obtain the so-called ``replica MMSE solution''  of \cite[Th. 8]{reeves-journal1}. 

Next, we discuss the threshold for perfect support reconstruction in the noiseless case, i.e., in the limit of
$\Pc_x \rightarrow \infty$, and $q > 0$. From Theorem \ref{th:totallysupercool} we already know that
vanishing $D(p,q,\Pc_x)$ cannot be achieved for any $p \leq q$. We now show that $D(p,q,\Pc_x)$ vanishes
for large $\Pc_x$ for all $p > q$.  This has previously been shown for both
optimal nonlinear measurement schemes and for Gaussian iid sensing matrices in \cite{wuverdunoisy}.
Therefore, the conclusion about the asymptotic optimality of Gaussian iid sensing matrices found in \cite{wuverdunoisy}
extends to sparsely sampled free random matrices.
We start by recalling the following general result from \cite{wuverdummse}:

\begin{theorem}
\label{thm:wuverdugral} 
Let $V$ is a discrete-continuous mixed
distribution, i.e. such that its distribution can be represented as
\begin{eqnarray} 
\nu = (1-\rho) \nu_d + \rho \nu_c, 
\end{eqnarray}
where $\nu_d$ is a discrete distribution and $\nu_c$ is an
absolutely continuous distribution, and $0 \leq\rho \leq 1$. 
Then, for $Z \sim \Cc\Nc(0,1)$ we have
\begin{eqnarray}
\mmse( V | \sqrt{\snr}  {V} + Z) = \frac{\rho}{\snr}  + o\left(\frac{1}{\snr}\right).
\end{eqnarray}
\hfill \QED
\end{theorem}

We are interested in the behavior of the SNR of the decoupled channel (\ref{decoupled-channel}) resulting from the 
MAP-SBS estimator, given by $q \, \eta \, \Pc_x$, as $\Pc_x \rightarrow \infty$. In particular, 
for given sparsity $0 < q \leq 1$, we are interested in determining the range of sampling rates $p$ for which
$q \, \eta \, \Pc_x \rightarrow \infty$, implying that $D(p,q,\Pc_x) \rightarrow 0$. 
Let $Z$ and $V_0$ be as defined in Claim \ref{th:x}. 
Then, using Theorem \ref{thm:wuverdugral} we can write 
\begin{eqnarray}
\mmse\left( V_0 | V_0 + \eta^{\frac{1}{2}} Z \right) & = &  \Pc_x \; \mmse\left( V_0/\sqrt{\Pc_x} | \sqrt{\Pc_x \eta} V_0 /\sqrt{\Pc_x} + Z \right) \\
& = & \frac{q}{\eta}  +  o\left( 1 \right), \label{mmseexp}
\end{eqnarray}
where, for the time being, we assume that $\Pc_x \eta$ grows unbounded as $\Pc_x \rightarrow \infty$.  
Using (\ref{mmseexp}) into (\ref{e:fix-pointeq1}) -- (\ref{e:fix-pointeq}), for sufficiently large $\Pc_x$ we have
\begin{eqnarray} \eta & = & \Rc_{\Rm}( - \mmse( V_0 | V_0 + \eta^{\frac{1}{2}}{Z})) \\
& \rightarrow & \Rc_{\Rm} \left (-\frac{q}{\eta} \right).
\label{occhibelli}
\end{eqnarray}
For the case of $\Um$ with iid elements, 
using \eqref{R-transfromGV} we obtain 
\begin{eqnarray}
\Rc_{\Rm}^{-1}(z)= 1-\frac{p}{z}
\end{eqnarray}
and solving (\ref{occhibelli}) with respect to $\eta$, we obtain 
\begin{eqnarray}
\eta = p - q .
\end{eqnarray}
In the case of Haar-distributed $\Um$, using \eqref{cool2}, we obtain
\begin{eqnarray}
\Rc_{\Rm}^{-1}(z) &=& \frac{p -z}{(1 - z)z} \\
\eta &=& \frac{p -q}{1-q}
\end{eqnarray}
For $p > q$, in those two cases the solutions are strictly positive
and, consequently, the support recovery error rate vanishes as the SNR grows without bound.
In fact, as we show next, this conclusion holds for the general class of sparsely sample free random matrices.

The goal is to show that $\lim_{\Pc_x \rightarrow \infty} \eta > 0$  for $p > q$, without relying on a closed-form expression 
for the R-transform. This implies  that 
$D(p,q,\Pc_x)$ vanishes for large $\Pc_x$ for all $p > q$.   
Assuming that (\ref{occhibelli}) holds,  using the definition of the R-transform 
as function of the $\eta$-transform given in \cite[Eq. 2.75 Sec. 2.2.5]{fnt} and the definition of $\eta$-transform 
as given in \cite[Sec. 2.2.2]{fnt},  we can rewrite the asymptotic equality $\eta = \Rc_{\Rm}(-q/\eta)$  as:
\begin{eqnarray}
q & = & 1 -  \EE\left[ \frac{1}{1+ s |{\sf R}|^2} \right] 
\label{amorebello}
\end{eqnarray}
where $s$ satisfies
\begin{eqnarray}
\frac{q}{\eta}   = \EE\left[ \frac{ s}{1+ s |{\sf R}|^2} \right]  
\label{sentimento}
\end{eqnarray}
and 
$|{\sf R}|^2$ denotes a random variable distributed as the limiting spectrum of $\Rm$. 

By eliminating $q$ and solving for $\eta$ in (\ref{amorebello}), (\ref{sentimento}) we obtain
\begin{eqnarray}
\eta = \frac{\EE\left[ \frac{|{\sf R}|^2}{1+ s |{\sf R}|^2} \right]}{\EE\left[ \frac{1}{1+ s |{\sf R}|^2} \right]}.
\label{belin-antonia1}
\end{eqnarray}
It is immediate to see that \eqref{belin-antonia1} is strictly positive for any finite $s$ 
(ranging from the mean to the harmonic mean of $|{\sf R}|^2$). 
In view of Property (\ref{boeta}) of the $\eta$-transform,  
\begin{eqnarray} 
1 - p \leq \EE\left[ \frac{1}{1+ s |{\sf R}|^2} \right] \leq 1, 
\end{eqnarray}
we conclude that (\ref{amorebello}) admits a unique positive and finite solution $s$ if and only if $1 - q \in (1 - p, 1]$, i.e., 
for $p > q$. Hence, (\ref{belin-antonia1}) yields $\eta > 0$ for $\Pc_x \rightarrow \infty$, as we wanted to show. 

We conclude this section by providing expressions for the MMSE 
in the estimation of the Bernoulli-Gaussian signal $\vv$ for high SNR. 
For iid $\Um$, we have
\begin{eqnarray}
\label{trieste} 
\mmse\left( V_0 | V_0 + \eta^{\frac{1}{2}} Z \right) &=& -\Rc_{\Rm}^{-1} ( \eta)  \\
\label{mamma} & = &  \frac{p}{\eta}  - 1, 
\end{eqnarray}
while for Haar-distributed $\Um$, we have 
\begin{eqnarray}
\label{trieste1} \mmse\left( V_0 | V_0 +
\eta^{\frac{1}{2}} Z \right) &=& -\Rc_{\Rm}^{-1} ( \eta)  \\
& = & \frac{p -\eta}{(1 - \eta)\eta}.
\end{eqnarray}
Notice that (\ref{mamma}) coincides with the result derived in \cite{wuverdunoisy} 
and that the high-SNR MMSE diverges for $p = q$. Since deleting samples
cannot improve the performance of the optimal MMSE estimator, it diverges for all $0 \leq p \leq q$.

\subsection{Replica analysis of a class of estimators via the large-deviation limit}

The classical noisy compressed sensing problem seeks the estimation of the sparse 
vector $\vv = \Xm \bv$ from $\yv$ in (\ref{model1}) for known $\Am, \Um$.  
Then, $\bv$ can be estimated by componentwise thresholding the estimate of $\vv$. 

A number of suboptimal low-complexity estimators in the compressed sensing literature take on the form 
\begin{eqnarray} \label{general-form}
\widehat{\vv} = \arg \min_{\vv \in \CC^n}  \left \{ \gamma \left \| \yv - \Am\Um \vv \right \|^2 + \sum_{i=1}^n f(v_i) \right \}, 
\end{eqnarray}
for some weighting parameter $\gamma > 0$ and cost function $f : \CC \rightarrow \RR_+$.  

The replica decoupling principle can be used to study the large-dimensional limit performance of 
such class of estimators by following the large-deviation recipe given in \cite{rangan-replica}.  
Briefly, the approach of \cite{rangan-replica} considers a sequence of  mismatched PMEs indexed by a parameter 
$\kappa \in \RR_+$,  where the assumed a priori density for $\vv$ takes on the form
\begin{eqnarray} \label{prior-kappa-v}
q^{(\kappa)}_{\vv}(\vv) = \frac{\exp\left ( - \kappa \sum_{i=1}^n f(v_i) \right )}{\int \exp\left ( - \kappa \sum_{i=1}^n f(z_i) \right ) d\zv}, 
\end{eqnarray}
(assuming that the integral 
converges for sufficiently large $\kappa$), 
and where the assumed transition density is given by 
\begin{eqnarray}  \label{assumed-multivariate-gaussian-map}
q^{(\kappa)}_{\yv|\vv, \Am, \Um}(\yv|\vv, \Am, \Um) = \left ( \frac{\gamma \kappa}{\pi} \right )^n \exp \left ( - \gamma \kappa \left \| \yv - \Am\Um\vv \right \|^2 \right ). 
\end{eqnarray}
Under a number of mild technical assumptions (see \cite{rangan-replica} for details), 
$\widehat{\vv}$ in (\ref{general-form}) can be obtained as the limit of the PME 
\begin{eqnarray} \label{mismatched-pme-kappa}
\widehat{\vv}^{(\kappa)} = \int \vv q^{(\kappa)}_{\vv | \yv, \Am, \Um} (\vv | \yv, \Am, \Um) d\vv.
\end{eqnarray}
for $\kappa \rightarrow \infty$.  Furthermore, for $n \rightarrow \infty$ and assuming the validity of the 
replica analysis,  a decoupled scalar channel model in the limit of $\kappa \rightarrow \infty$ can be established such that the joint distribution 
of $(v_{0i}, v_i, \widehat{v}_i)$ converges to the joint distribution of $(V_{0}, V, \widehat{V})$, 
where the form of the joint distribution of $V_0, Y, V$ is again given by (\ref{joint-decoupled-V}) and where $\widehat{V}$ is a 
function of $Y$.  The form of the fixed-point equations yielding $\eta$ and $\xi$ and of $\widehat{V}$ as a function of  $Y$ 
depend  on the specific estimator considered, i.e., on the value of $\gamma$ and on the cost function $f(v)$ in (\ref{general-form}). 
In particular, following in the footsteps of  \cite{rangan-replica} with a few minor variations 
in order to adapt to our case,~\footnote{Details are omitted since they can be easily worked out
from \cite{rangan-replica}. 
} it is not difficult to show that 
$\widehat{V} = \widehat{v}(Y;\xi)$, where we define
\begin{eqnarray} \label{scalar-map}
\widehat{v}(y;\xi) = \arg \min_{v \in \CC} \left \{ \xi | y - v |^2 + f(v) \right \}, 
\end{eqnarray}
and that the 
fixed-point equations yielding $\eta$ and $\xi$ in the limit of $\kappa \rightarrow \infty$ are given by 
\begin{subequations}
\begin{eqnarray}
\chi & = & \gamma \EE\left [ \sigma^2(Y;\xi) \right ] \label{saddle-final-chi-map} \\
\delta & = & \EE \left [ |V_0 - \widehat{v}(Y;\xi) |^2 \right ] \label{saddle-final-delta-map} \\
\xi & = & \gamma \Rc_{\Rm}(-\chi) \label{saddle-final-xi-map} \\
\eta & = & \frac{(\xi/\gamma)^2}{\xi/\gamma  + \dot{\Rc}_{\Rm}(- \chi) (\delta - \chi)}, \label{saddle-final-eta-map}
\end{eqnarray}
\end{subequations}
where 
\begin{eqnarray} \label{sigma-rangan}
\sigma^2(y;\xi) = \lim_{v \rightarrow \widehat{v}(y;\xi)} \frac{|v - \widehat{v}(y;\xi)|^2}{\xi | y - v |^2 + f(v)  - \left [ \xi | y -  \widehat{v}(y;\xi)|^2 + 
f(\widehat{v}(y;\xi)) \right ]},
\end{eqnarray}
When $\Um$ has iid elements, from (\ref{saddle-final-eta-iid}) -- (\ref{saddle-final-xi-iid}) we find
\begin{subequations}
\begin{eqnarray} 
\frac{1}{\eta} & = & \frac{1}{p} \left (1 + \EE \left [ |V_0 - \widehat{v}(Y;\xi) |^2 \right ] \right ) \label{rangan-fixed-point-map-eta} \\
\frac{1}{\xi} & = & \frac{1}{p} \left ( \frac{1}{\gamma} +  \EE\left [ \sigma^2(Y;\xi) \right ] \right ), \label{rangan-fixed-point-map-xi}
\end{eqnarray}
\end{subequations}
which coincide with \cite[Eq. (30a) - (30b)]{rangan-replica}, up to a different normalization 
and the fact that we consider complex circularly symmetric instead of real random variables as in \cite{rangan-replica}.

\subsection{Thresholded linear MMSE estimator} \label{section:mmseTL}

A simple suboptimal estimator for $\vv$ is the linear MMSE estimator, given by 
\begin{eqnarray} \label{lmmse-explicit}
\widehat{\vv} =  \left [ \gamma^{-1} \Id + \Rm \right ]^{-1} \Um^\dagger \Am^\dagger \yv.
\end{eqnarray}
with $\gamma = q\Pc_x$ and $\Rm$ defined in (\ref{RR}). 
It is immediate to verify that (\ref{lmmse-explicit}) can be expressed in the form (\ref{general-form})
by letting $f(v) = |v|^2$. 

Although the asymptotic performance and the decoupled channel model of linear MMSE estimation
can be obtained directly from classical results in large random matrix theory both for iid and for 
Haar-distributed $\Um$ (see \cite{fnt} and references therein), it is instructive to apply the replica large-deviation approach
outlined before. In this way, we can recover known results obtained rigorously by other means, thus 
lending support to the validity of the 
replica-based large-deviation approach. 

Particularizing (\ref{scalar-map}) and (\ref{sigma-rangan}) to the case $f(v) = |v|^2$ we obtain
\begin{eqnarray} \label{scalar-map-lmmse}
\widehat{v}(y;\xi) = \frac{\xi}{1 + \xi} y
\end{eqnarray}
and
\begin{eqnarray} \label{sigma-rangan-lmmse}
\sigma^2(y;\xi) = \frac{1}{1 + \xi}, 
\end{eqnarray}
yielding
\begin{eqnarray}  
\EE \left [ |V_0 - \widehat{v}(Y;\xi) |^2 \right ] & = & \EE \left [ \left |V_0 -\frac{\xi}{1 + \xi} Y \right |^2 \right ] \\
& = & \frac{\gamma + \xi^2/\eta}{(1 + \xi)^2},  \label{blablabla}
\end{eqnarray}
where we used the fact that $\EE[|V_0|^2] = q \Pc_x = \gamma$. 
Replacing (\ref{sigma-rangan-lmmse}) and (\ref{blablabla}) into (\ref{saddle-final-chi-map}) -- (\ref{saddle-final-eta-map}), we obtain
the fixed-point equations for the linear MMSE estimator.
In the iid case, using (\ref{rangan-fixed-point-map-eta}) -- (\ref{rangan-fixed-point-map-xi}), we obtain that $\xi = \gamma \eta$ and 
\begin{eqnarray}
\eta = \frac{-(1 + (1-p)\gamma) + \sqrt{(1 + (1-p)\gamma)^2 + 4p\gamma}}{2\gamma},
\end{eqnarray}
which coincides with the well-known expression of the multiuser efficiency of the linear 
MMSE detector for an iid matrix with aspect ratio $pn \times n$ and elements with mean 0 and variance 
$1/n$ (see \cite{fnt} and expression (\ref{eta-iid}) evaluated for $\beta = 1, s = \gamma$). 

In the Haar-distributed case, using (\ref{cool2}), we can solve explicitly for $\xi$ by eliminating $\chi$ in (\ref{saddle-final-chi-map}) 
and (\ref{saddle-final-xi-map}). After some more complicated algebra than in the iid case, 
we arrive at the solution 
\begin{eqnarray}
\xi = \frac{\gamma p}{1 + (1 - p)\gamma}
\end{eqnarray}
We also find that, as in the iid case, $\xi = \gamma\eta$. Hence
$\eta$ is given in closed form as
\begin{eqnarray} 
\eta = \frac{p}{1 + (1-p)\gamma},
\end{eqnarray}
which coincides with the well-known form of the multiuser efficiency of the linear 
MMSE detector for a CDMA system with observation model $\rv = \Am \Um \vv +\zv$, 
where $\Um$ is $n \times n$ Haar-distributed, given by 
the solution of (\ref{eta-haar-fixedp}) in the case $\beta = 1, s = \gamma$ 
(or, equivalently, by the limit of (\ref{eta-haar}) for $\beta \rightarrow 1$). 

In order to calculate the performance of the thresholded linear MMSE estimator, notice that the estimator output 
converges in distribution to $\widehat{V} = \widehat{v}(Y;\xi) = \frac{\xi}{1 + \xi} Y$ where, according to the decoupled channel model,  
$Y = V_0 + \eta^{-\frac{1}{2}} Z$, and $Z \sim \Cc\Nc(0,1)$. Thresholding $\widehat{V}$ or $Y$ is clearly equivalent. 
Hence, the support recovery error rate in this case takes on the same form already derived for the 
MAP-SBS (see (\ref{energy-detect}) -- (\ref{distotionexcat})), for a different value of $\eta$ calculated 
via (\ref{saddle-final-chi-map}) -- (\ref{saddle-final-eta-map}).

\subsection{Thresholded Lasso estimator}  \label{lasso-sec}

We now follow an approach similar to that in Section \ref{section:mmseTL} in order to analyze the Lasso estimator,
which so far has only been analyzed  for iid 
sensing matrices.

The Lasso estimator, widely studied in the compressed sensing literature
\cite{lasso,donoho-stable} comes directly in the form (\ref{general-form}) for $f(v) = |v|$. 
In this case, the parameter $\gamma$ must be optimized depending on the target performance. 
For example, in the classical noisy compressed sensing problem we are interested in the value of $\gamma$ that 
minimizes 
$\EE[\|\vv - \widehat{\vv}\|^2]$. 

Particularizing (\ref{scalar-map}) and (\ref{sigma-rangan}) to the case $f(v) = |v|$ we obtain
\begin{eqnarray} \label{scalar-map-lasso}
\widehat{v}(y;\xi) = \left [ |y| - \frac{1}{2\xi} \right ]_+ \frac{y}{|y|},
\end{eqnarray}
where $[\cdot ]_+$ takes the positive part of its argument, and
\begin{eqnarray} \label{sigma-rangan-lasso}
\sigma^2(y;\xi) = 1\left  \{ |y| - \frac{1}{2\xi} > 0 \right \} \frac{1}{\xi},
\end{eqnarray}
where $1\{ \cdot\}$ is the indicator function of the event inside the brackets.
Notice that (\ref{scalar-map-lasso}) and (\ref{sigma-rangan-lasso})
generalize the expressions found in \cite{rangan-replica} to the complex case.
In this case, we have
\begin{eqnarray}
\EE \left [ |V_0 - \widehat{v}(Y;\xi) |^2 \right ] & = & \EE \left [ \left |V_0 -  \left [ |Y| - \frac{1}{2\xi} \right ]_+ \frac{Y}{|Y|} \right |^2 \right ] \nonumber \\
& = & q \Pc_x + \frac{1-q}{\eta} \left [ e^{-\eta'}  - \sqrt{\pi \eta'} {\rm erfc} \left ( \sqrt{\eta'} \right ) \right ] \nonumber \\
& & + \frac{q}{\mu} \left [ \frac{1 - \Pc_x\eta}{1 + \Pc_x\eta} e^{-\mu'} - \frac{\sqrt{\pi \mu'}}{1 + \Pc_x\eta} {\rm erfc} \left ( \sqrt{\mu'} \right ) \right ], \label{mse-lasso-closed-form} 
\end{eqnarray}
where $\eta' = \eta/(4 \xi^2)$, $\mu' = \mu/(4\xi^2)$, and $\mu$ is defined in (\ref{mudef}). 
The derivation of (\ref{mse-lasso-closed-form}) is not completely straightforward and it is provided in Appendix \ref{app:formulas}. 

From (\ref{sigma-rangan-lasso}) we have 
\begin{eqnarray} 
\EE[\sigma^2(Y;\xi)] & = & \frac{1}{\xi} \PP[|Y| > 1/(2\xi)] \label{sigma-rangan-lasso-closed-form} \\
& = & \frac{1}{\xi} \left (  q e^{-\mu'} + (1 - q) e^{-\eta'} \right )  \label{PPY}
\end{eqnarray}
Replacing (\ref{mse-lasso-closed-form}) and (\ref{sigma-rangan-lasso-closed-form}) into
(\ref{saddle-final-chi-map}) -- (\ref{saddle-final-eta-map}), we obtain the fixed-point equation for calculating the decoupled channel parameters
$\eta, \xi$ for the analysis of the Lasso estimator for given parameter $\gamma$. 
In the iid case, using (\ref{rangan-fixed-point-map-eta}) -- (\ref{rangan-fixed-point-map-xi}), we obtain
the same system of equations given in \cite{rangan-replica}, up to a different normalization and the fact that here we consider 
complex signals. 
Furthermore, it is immediate to recognize that (\ref{rangan-fixed-point-map-eta}) corresponds to the 
state evolution of the AMP with soft-thresholding (AMP-ST) as described in \cite{reeves-journal1}, 
where the scalar soft-thresholding function is given by (\ref{scalar-map-lasso})
for an arbitrary thresholding parameter $\xi > 0$. 
The large-dimensional analysis leading to the state evolution equation (\ref{rangan-fixed-point-map-eta}) 
is rigorously proved in \cite{montanari-dense} for the case where $\Um$ is iid Gaussian. 
Based on this fact, 
it is tempting to conjecture that the analysis is valid for the general iid case (subject to usual mild conditions 
on the matrix element distribution) and that the replica analysis yields correct results also for the more general 
class of matrices considered in this paper. 

In order to obtain an estimate of $\bv$ (support of $\vv$), a natural approach consists of 
selecting the non-zero components of $\widehat{\vv}$. However, this method yields rather poor results 
in the Bernoulli-Gaussian case and in other cases where the magnitudes of the 
non-zero components of $\vv$ are not bounded away from zero.
Instead, in an iterative implementation of the Lasso solver (e.g., using the method in \cite{giannakis}, or the AMP-ST), 
it is possible to generate a ``noisy'' version of the Lasso estimate $\widehat{\vv}$ before the soft-thresholding step
(see Section \ref{sec:dist-results} and \cite{reeves-journal1}). This noisy Lasso estimate corresponds to the 
decoupled channel model with marginal distribution $Y = V_0 + \eta^{-\frac{1}{2}} Z$, 
with $\eta$ given by the fixed-point equation in the Lasso case.  Hence, the support recovery error rate 
takes on the same form already derived for the  MAP-SBS (see (\ref{energy-detect}) -- (\ref{distotionexcat})), for a different 
value of $\eta$, calculated via (\ref{saddle-final-chi-map}) -- (\ref{saddle-final-eta-map}) for the Lasso case as explained above.

\subsection{Support recovery error rate examples}  \label{sec:dist-results}

In order to illustrate the above results and compare the behavior of different support estimators, we show some
numerical examples and compare the theoretical asymptotic results with finite-dimensional simulations.
Figs.~\ref{dist_all20db} and \ref{dist_all50db} show the support recovery error rate $D(p,q,\Pc_x)$ 
versus the sampling rate $p$ for a Haar-distributed sensing matrix $\Um$ and a Gaussian-Bernoulli source signal $\vv$ with
$q = 0.2$ and $\SNR = q\Pc_x$ equal to $20$ and 50 dB, respectively.

\begin{figure}[ht]
\centerline{\includegraphics[width=12cm]{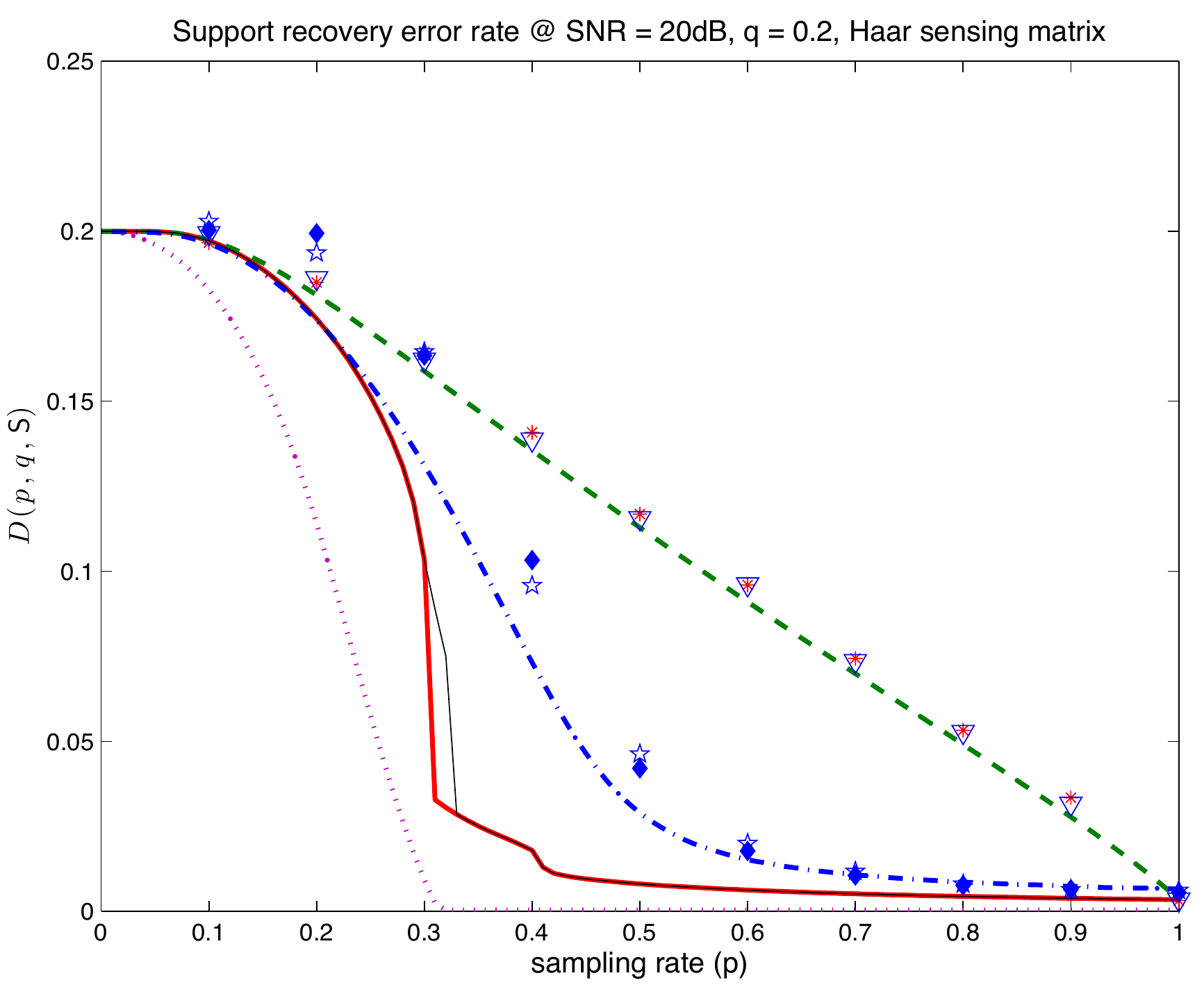}}
\caption{Support recovery error rate $D(p,q,\Pc_x)$ versus $p$, for $q = 0.2$ and $\SNR = q\Pc_x = 20$ dB for different estimators, 
asymptotic results and and finite-dimensional simulations.
Solid thick line: MAP-SBS, asymptotic; Dotted line: Information theoretic lower bound; Dot-dash line: Thresholded Lasso, asymptotic; 
Dashed line: Thresholded linear MMSE, asymptotic; Thin solid line: Conjectured AMP-MMSE, corresponding to the right-most fixed point
of (\ref{e:fix-pointeq1}) -- (\ref{e:fix-pointeq}). Some finite-dimensional simulations are shown 
for dimension $n = 100$  for the thresholded linear MMSE
estimator (asterisk: Haar sensing matrix; triangle: DFT sensing matrix) and for the thresholded Lasso
(lozenge: Haar sensing matrix; star: DFT sensing matrix).}
\label{dist_all20db}
\end{figure}

\begin{figure}[ht]
\centerline{\includegraphics[width=12cm]{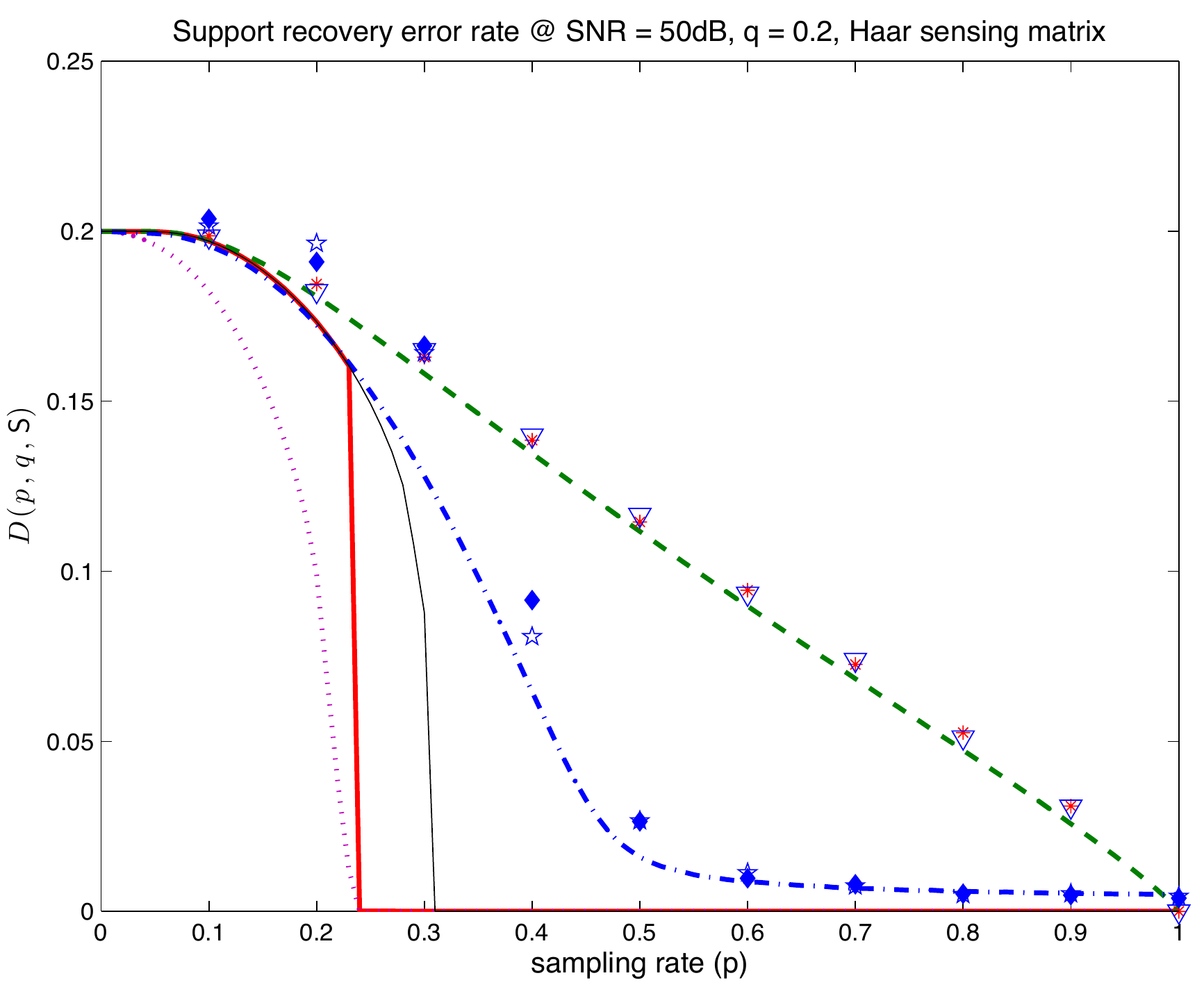}}
\caption{Support recovery error rate $D(p,q,\Pc_x)$ versus $p$, for $q = 0.2$ and $\SNR = q\Pc_x = 50$ dB for different estimators, 
asymptotic results and and finite-dimensional simulations.
Solid thick line: MAP-SBS, asymptotic; Dotted line: Information theoretic lower bound; Dot-dash line: thresholded Lasso, asymptotic; 
Dashed line: Thresholded linear MMSE, asymptotic; Thin solid line: Conjectured AMP-MMSE, corresponding to the right-most fixed point
of (\ref{e:fix-pointeq1}) -- (\ref{e:fix-pointeq}). Some finite-dimensional simulations are shown for dimension $n = 100$ for 
the thresholded linear MMSE estimator (asterisk: Haar sensing matrix; triangle: DFT sensing matrix) and for the thresholded Lasso
(lozenge: Haar sensing matrix; star: DFT sensing matrix).}
\label{dist_all50db}
\end{figure}

A few remarks are in order: 
\begin{itemize}
\item The MAP-SBS asymptotic distortion is obtained by choosing the fixed-point solution of (\ref{e:fix-pointeq1}) -- (\ref{e:fix-pointeq}) that minimizes
the free energy $\Ic_1$, as discussed in Section \ref{sec:info-results}. Instead, if we choose only the right-most fixed point, we obtain 
the solution of the conjectured state evolution equation corresponding to the AMP-MMSE applied to Haar-distributed sensing matrices.
As previously remarked, it is known that such state evolution equation is exact in the case of iid sensing matrices.
\item The information theoretic lower bound is obtained by taking the minimum of all the upper bounds on $\Ic$ developed in 
Theorems \ref{minfo-bound2ex} and \ref{th:bounds}, and using it in (\ref{love}). 
\item We run finite-dimensional simulations for dimension $n = 100$ 
for the thresholded linear MMSE and thresholded Lasso estimators. We considered both random unitary $\Um$ (Haar distributed) and
the case of a fixed deterministic $\Um = \Fm$, where $\Fm$ is the $n$-dimensional unitary DFT matrix with elements 
$[\Fm]_{m,k} = \frac{e^{j\frac{2\pi}{n}(m-1)(k-1)}}{\sqrt{n}}$. Interestingly, the simulations show that random unitary and deterministic DFT
yields essentially the same performance (up to Monte Carlo simulation fluctuations). 
This corroborates our conjecture that the asymptotic analysis 
of Haar-distributed $\Um$ carries over to the case of a DFT matrix. The case of DFT matrices is particularly relevant for applications, since in many communication and signal processing problems signals are sparse in the time (resp., frequency) domain and are randomly sampled in the dual domain, so that
a random selection of the rows of a DFT matrix arises as a sensing matrix naturally matched to the problem.
\item As already noticed in several works, the gap between the optimal MAP-SBS estimator and the suboptimal low-complexity estimators
grows for high SNR (compare Fig. \ref{dist_all20db} and Fig. \ref{dist_all50db}). In contrast, the thresholded linear MMSE estimator
yields poor performance for all $p < 1$, and this is quite insensitive to SNR. 
\item In order to solve the complex Lasso, we used the iterative method of \cite{giannakis}. This scheme has slightly 
lower complexity than AMP-ST, and provably converges to the Lasso solution. By comparing the component-wise thresholding
step in \cite{giannakis} and the symbol-by-symbol estimator $\hat{v}(Y;\xi)$ for the decoupled channel model
given in (\ref{scalar-map-lasso}), it is natural to identify the noisy Lasso solution with the vector
\begin{eqnarray}  \label{lasso-noisy}
\widetilde{\vv} = \widehat{\vv}^{(\infty)} + \Dm \Gm^\dagger \left ( \yv - \Gm \widehat{\vv}^{(\infty)} \right ),
\end{eqnarray}
where $\widehat{\vv}^{(\infty)}$ is the solution of the iterative algorithm of \cite{giannakis} after convergence, 
$\Gm$ is the matrix obtained by taking the non-zero rows of $\Am\Um$, and 
$\Dm = \diag(1/\|\gv_1\|^2, \ldots, 1/\|\gv_n\|^2)$ where $\gv_\ell$ is the $\ell$-th column of $\Gm$. 
The support recovery error rate shown in Figs.~\ref{dist_all20db} and \ref{dist_all50db} for the finite-dimensional simulation of
the thresholded Lasso is obtained by applying the threshold detector given in (\ref{energy-detect}), for 
$\eta$ calculated via the asymptotic fixed-point equations (\ref{saddle-final-chi-map}) -- (\ref{saddle-final-eta-map}), to the components
of $\widetilde{\vv}$ given in (\ref{lasso-noisy}).
The asymptotic analysis and the finite-dimensional simulation were computed for the same value of the 
parameter $\gamma$, which must be chosen for each combination of system parameters $p, q$ and $\Pc_x$. 
Several heuristic methods for the choice of $\gamma$ are proposed in the literature. Following \cite{lasso20}, 
we used $\gamma = (1/20)\|\Gm^\dagger \yv\|_\infty$ (the optimization of $\gamma$ for the asymptotic case is an interesting topic for further investigation.)
\end{itemize}

\section{Conclusion}  

%
%

In the standard compressed sensing model,  
the sensing matrix $\Am \Um$ is such that $\Am$ is diagonal with independent $\{0, 1\}$ components
and $\Um$ has iid coefficients. 
In addition to this model, 
we allow the square matrix $\Um$ to be Haar-distributed (uniformly distributed among all unitary matrices)
or, more generally, to be free from any Hermitian deterministic matrix.

Motivated by  applications, in this paper we have carried out a large-size analysis of:
\begin{enumerate}
\item
the mutual information
between the noisy observations and the Bernoulli-Gaussian input (conditioned on the sensing matrix), 
\item
the mutual information
between the noisy observations and the Gaussian input prior to being subject to random ``hole-punching".
\end{enumerate}
We have obtained asymptotic formulas using fundamentally different approaches for both mutual informations:
the first following a replica-method analysis whose scope we enlarge to encompass the desired class of random matrices, while the second invokes results from freeness and the asymptotic spectral distribution
of random matrices. 

Depending on the case, the mutual informations are expressed either 
through the mutual information between a scalar Bernoulli-Gaussian random variable and its Gaussian-contaminated version,
or
explicitly, through the solution of coupled nonlinear equations.
We have also studied how to choose among the solutions of those equations.

Our upper and lower bounds on the mutual informations do not rely on the replica method. Yet, 
they turn out to give excellent agreement with the replica analysis.  Through the analysis of the bounds we also provide 
a simple converse which shows that the asymptotic distortion is bounded away from zero regardless of signal-to-noise ratio 
for $p \leq q$.
For $p >q $, Wu and Verd\'{u} \cite{wuverdunoisy} showed that 
Gaussian iid sensing matrices are asymptotically as effective for compressed sensing as the 
best nonlinear measurement (or encoder). Here, we have been able to extend that conclusion to the
class of sparsely sampled free random matrices.

We have analyzed several decision rules
such as the optimum symbol-by-symbol rule, the Lasso, and the linear MMSE estimator, followed by thresholding
for support recovery. 
Those analyses follow the decoupling principle, originally introduced in \cite{guo-verdu} for iid matrices.
Specializing these new results we recover the iid formulas found in \cite{guo-verdu,rangan-replica,reeves-journal1},
with the exception of the ML detector analyzed in \cite{reeves-journal1}, which is tailored to the case when the number
of nonzero coefficients is known at the estimator, while in our analysis that number is binomially distributed.

The important case where $\Um$ is a deterministic DFT matrix remains open.
However, we have provided intuition and simulation evidence to buttress the conjecture
that its solution in fact coincides with the case where $\Um$ is Haar distributed.

\centerline{{\bf Acknowledgements}}

 G. Caire, S. Shamai, and S. Verd\'u wish to acknowledge the Binational Science Foundation Grant N. 2008269.
The work of G. Caire was partially supported by NSF Grant CCF-0729162.
The work of S. Verd\'u was partially supported by the Center for Science of Information (CSoI),
an NSF Science and Technology Center, under grant agreement CCF-0939370.

\appendices

\section{Proof of Theorem \ref{th:totallyelementary}}\label{proof:totallyelementary}

Let $(X,Y) \sim P_X P_{Y|X}$.
For any $\widehat{X} \in \widehat{\Xc}$, such that $X \leftrightarrow Y \leftrightarrow \widehat{X}$, and function
$\mathsf{d} \colon \Xc \times \widehat{\Xc} \rightarrow [0, \infty)$, (\ref{poju1}) and the data processing inequality  yield
\begin{eqnarray}
R ( \mathbb{E} [\mathsf{d} ( X, \widehat{X} ) ]) & \leq & I (X; \widehat{X}) \label{jopu1}\\
&\leq& I (X;Y) \label{jopu2}
\end{eqnarray}
Supremizing over $\widehat{X}$ and in view of the fact that
$R ( \cdot )$ is a monotonically non-increasing function, the result follows.

It is worth emphasizing the totally elementary nature of the proof of Theorem \ref{th:totallyelementary},
and in particular the fact that it does not involve any type of operational characterization of information theoretic fundamental coding limits. A different approach based on those limits and Fano's inequality is taken in \cite{reeves-journal1} to show
Lemma 5 therein.

\section{Proof of Claim \ref{th:x}}\label{proof:x}

We let $\vv_0 = \Xm \bv_0$ with $\Xm \diag(\xv)$ and $\xv$ an iid Gaussian vector with $p_\xv(\xv) = \frac{1}{(\pi \Pc_x)^n} \exp\left (-\|\xv\|^2/\Pc_x\right )$ and $\bv_0$ Bernoulli-$q$, with probability mass function $p_{\bv_0}(\bv_0)$. 
Notations are as in Section \ref{decoup-sec}). In particular, $\yv = \Am \Um \Xm \bv_0 + \zv$, as in (\ref{model1}). 
Consider the assumed conditional probability density 
\begin{eqnarray}  \label{postulated-y-v}
q_{\yv|\vv, \Am, \Um}(\yv|\vv, \Am, \Um) = \left ( \frac{\gamma}{\pi} \right )^n \exp \left ( - \gamma \left \| \yv - \Am\Um\vv \right \|^2 \right )
\end{eqnarray}
for some $\gamma > 0$. We also consider an assumed iid prior density on $\vv$, denoted by $g(\vv)$ for simplicity of notation, 
and let $g_0(\vv_0)$ denote the Bernoulli-Gaussian density of $\vv_0$. 
Removing the conditioning with respect to $\vv$, we obtain
\begin{eqnarray} \label{qyAU}
q_{\yv|\Am,\Um}(\yv|\Am,\Um) = \left ( \frac{\gamma}{\pi} \right )^n \int g(\vv)  \exp \left ( - \gamma \left \| \yv - \Am\Um\vv \right \|^2 \right ) d\vv.
\end{eqnarray}
We wish to calculate the mutual information rate $\Ic_1$ defined in (\ref{inforate1}), which can be expressed as
\begin{eqnarray}  
\Ic_1 & = & \lim_{n \rightarrow \infty} \frac{1}{n} I(\vv_0; \yv|\Am, \Um)   \\
& = & \lim_{n \rightarrow \infty} \frac{1}{n} \left [ - \EE [ \log p_{\yv|\Am,\Um} (\yv|\Am,\Um) ]   +  \EE[ \log p_{\yv|\vv_0, \Am\Um}(\yv|\vv_0, \Am,\Um) ] \right ]   \\
& = & - \lim_{n \rightarrow \infty} \frac{1}{n} \left . \EE [ \log Z(\yv,\Am,\Um) ]  \right |_{\stackrel{g(\cdot) = g_0(\cdot)}{\gamma = 1}} - \log(\pi e)  \label{info1}
\end{eqnarray}
where we define $Z(\yv, \Am, \Um) = q_{\yv|\Am,\Um}(\yv|\Am,\Um)$, 
and recognize that (\ref{qyAU}) can be interpreted as the partition function (from which the notation ``$Z$'')
of a statistical mechanical  system with ``quenched disorder parameters'' $\yv, \Am, \Um$, ``state'' $\vv \sim g(\cdot)$ and 
unnormalized Boltzman distribution $q_{\yv|\vv, \Am, \Um}(\yv|\vv, \Am, \Um) g(\vv)$.~\footnote{In this case,
$\Hc(\vv|\yv,\Am,\Um) =  \left \| \yv - \Am\Um\vv \right \|^2 - \frac{1}{\gamma} \log g(\vv)$ plays the role of the system's Hamiltonian, and $\gamma$ is the inverse temperature \cite{guo-verdu}.} 

The condition $g(\cdot) = g_0(\cdot0$ and $\gamma=1$ correspond to the case where the assumed prior and noise variance in the observation model 
are ``matched'', i.e., they coincide with the true priors and noise variance. However, it is useful to consider the derivation for
general $\gamma$ and $g(\cdot)$, since this same derivation will apply to the general class  of mismatched PMEs defined in Section 
\ref{decoup-sec}).  The quantity
\begin{eqnarray} \label{free-energy-def}
\Ec = - \lim_{n \rightarrow \infty} \frac{1}{n}  \EE [ \log Z(\yv,\Am,\Um) ]
\end{eqnarray}
is the system per-component {\em free-energy} of the underlying physical system. In the following, we
shall compute $\Ec$ using the Replica Method of statistical physics,
under the so-called Replica Symmetry (RS) assumption
\cite{Nishimori,Mezard,Tanaka,guo-verdu,Guo1}. Summarizing, the method
comprises the following steps: since computing the expectation of
the log in (\ref{free-energy-def}) is usually complicated, we use the identity
\begin{eqnarray} \label{replica-trick1}
\EE [ \log Z(\yv,\Am,\Um) ]  = \lim_{u \rightarrow 0} \frac{\partial}{\partial u} \log \left ( \EE \left [ Z^u(\yv, \Am, \Um) \right ] \right )
\end{eqnarray}
for $u \in \RR_+$. Then, exchanging limits, we can write
\begin{eqnarray} \label{replica-trick2}
\Ec =  - \lim_{u \rightarrow 0} \frac{\partial}{\partial u}    \lim_{n \rightarrow \infty} \frac{1}{n} \log \left ( \EE \left [ Z^u(\yv, \Am, \Um) \right ] \right )
\end{eqnarray}
Finally, we evaluate the quantity
\begin{eqnarray} \lim_{n \rightarrow \infty} \frac{1}{n} \log \left ( \EE \left [ Z^u(\yv, \Am, \Um) \right ] \right ) \end{eqnarray}
for $u$ positive integer, such that $Z^u(\yv, \Am, \Um)$ can be seen as the partition function of a 
$u$-fold Cartesian product system (i.e., $u$ parallel ``replicas'' of the original system), with state vectors 
$\vv_1,\ldots, \vv_u$, and the {\em same} quenched parameters
$\yv, \Am, \Um$. In particular, we can write
\begin{eqnarray} \label{free-energy-step1}
Z^u(\yv,\Am,\Um)
& = &
\left (\frac{\gamma}{\pi}\right )^{un}  \int  d\vv_1 \cdots d\vv_u \left ( \prod_{a=1}^u g(\vv_a)  \right )
\exp \left ( - \gamma \sum_{a=1}^u \left \| \yv - \Am\Um\vv_a \right \|^2 \right ).
\end{eqnarray}
The next step consists of calculating
\begin{align} \label{free-energy-step2}
& \EE[Z^u(\yv,\Am,\Um) | \Am,\Um]  =   \\
& =  \left (\frac{\gamma}{\pi}\right )^{un}  \int  d\vv_0 d\vv_1 \cdots d\vv_u \left ( g_0(\vv_0) \prod_{a=1}^u g(\vv_a)  \right ) \; \int \exp \left ( - \gamma \sum_{a=1}^u \left \| \zv - \Am\Um(\vv_a - \vv_0) \right \|^2 \right ) \; \frac{1}{\pi^n} e^{-\|\zv\|^2} \; d\zv
\end{align}
Standard Gaussian integration (by completing the squares) yields
\begin{eqnarray} \label{free-energy-step3}
\int \exp \left ( - \gamma \sum_{a=1}^u \left \| \zv - \Am\Um(\vv_a - \vv_0) \right \|^2 \right ) \; \frac{1}{\pi^n} e^{-\|\zv\|^2} \; d\zv & = &
(1 + u\gamma)^{-n} \exp \left ( -n \trace \left (\Rm \Lm \right ) \right )
\end{eqnarray}
where $\Rm = \Um^\dagger \Am^\dagger \Am \Um$,  as defined in (\ref{RR}),
and where $\Lm$ is a rank-$u$ matrix defined as follows: let  $\sv_a = \vv_a - \vv_0$ for $a = 1,\ldots, u$, and let
$\Sm = [\sv_1, \ldots, \sv_u]$. Then,
\begin{eqnarray} \label{LL}
\Lm = \frac{\gamma}{n} \Sm \left ( \Id - \frac{1}{\gamma^{-1} + u} \onev\onev^\transp \right ) \Sm^\dagger
\end{eqnarray}
where $\onev$ denotes an all-ones column vector of appropriate
dimension. Next, we need to average with respect to $\Am, \Um$, 
i.e.,  with respect to $\Rm$. To this purpose, we apply the generalized
Harish-Chandra-Itzykson-Zuber integral \cite{Harish,Itzykson} as
follows:
\begin{align} \label{free-energy-step4}
& \lim_{n \rightarrow \infty} \frac{1}{n} \log \EE[Z^u(\yv,\Am,\Um)] =    \nonumber \\
& = u \log \gamma - u \log \pi - \log (1 + u\gamma) \nonumber \\
& + \lim_{n \rightarrow \infty} \frac{1}{n} \log \left ( \int  d\vv_0  \cdots d\vv_u \left ( g_0(\vv_0) \prod_{a=1}^u g(\vv_a)  \right )
\EE \Big [ \exp \left ( -n \trace \left (\Rm \Lm \right )  \right ) \Big | \Lm \Big ] \right )   \\
& = u \log \gamma - u \log \pi - \log (1 + u\gamma) \nonumber \\
& + \lim_{n \rightarrow \infty} \frac{1}{n} \log \left ( \int  d\vv_0 \cdots d\vv_u \left ( g_0(\vv_0) \prod_{a=1}^u g(\vv_a)  \right ) 
 \exp \left ( -n \sum_{i=1}^u \int_0^{\lambda_i(\Lm)} \Rc_{\Rm}(-w) dw \right ) \right )
\end{align}
where $\Rc_{\Rm}(w)$ denotes the R-transform of $\Rm$ and $\lambda_(\Lm)$ denotes the $i$-th eigenvalue of $\Lm$.

Our goal now is to evaluate the limit
\begin{eqnarray} \label{free-energy-step5}
\lim_{n \rightarrow \infty} \frac{1}{n} \log \left ( \int  d\vv_0 \cdots d\vv_u \left ( g_0(\vv_0) \prod_{a=1}^u g(\vv_a)  \right )  \;
\exp \left ( -n \sum_{i=1}^u \int_0^{\lambda_i(\Lm)} \Rc_{\Rm}(-w) dw \right ) \right ).
\end{eqnarray}
In order to proceed, we make the common RS assumption and define the empirical correlations
\begin{eqnarray} \label{empirical-corr}
Q_{a,a'} = \frac{1}{n} \sum_{k=1}^n v_{ak} v^*_{a'k}
\end{eqnarray}
of vectors $\vv_a, \vv_{a'}$ for $0 \leq a,a' \leq u$. 
Noticing that the limit (\ref{free-energy-step5}) is given as the limit of a normalized log-sum,
Varadhan's lemma yields that this limit is given by the ``dominant configuration'' of the vectors $\vv_0,\ldots, \vv_u$, defined in terms of their empirical correlation matrix $\Qm = [Q_{a,a'}]$. The RS assumption ``postulates'' that  this dominant configuration satisfies the following 
symmetric form:
\begin{eqnarray} \label{QRS}
\Qm = \left [ \begin{array}{cc}
\epsilon_0 & \vartheta \onev^\transp  \\
\vartheta^* \onev & (\epsilon_1 - \omega)\Id + \omega \onev\onev^\transp \end{array} \right ].
\end{eqnarray}
In Appendix \ref{eigenvalues-L} we show that, for $\Qm$ in the form (\ref{QRS}), the eigenvalues of $\Lm$ are given by
\begin{eqnarray} \label{eigenvalues}
\lambda_1 \eqdef \lambda_1(\Lm) & = & \frac{\epsilon_1 - \omega + u(\epsilon_0 -  2\Re\{\vartheta\} + \omega)}{\gamma^{-1} + u}   \\
\lambda_2 \eqdef \lambda_i(\Lm) & = & \gamma(\epsilon_1 - \omega), \;\;\;\; i = 2,\ldots, u   \\
\lambda_i(\Lm) & = & 0, \;\;\;\; i = u+1, \ldots, n.
\end{eqnarray}
Therefore, we define
\begin{eqnarray} \label{chandra-exponent}
\Gc^{(u)}(\Qm) & \eqdef & \sum_{i=1}^u \int_0^{\lambda_i(\Lm)} \Rc_{\Rm}(-w) dw   \\
& = & \int_0^{\lambda_1} \Rc_{\Rm}(-w) dw + (u - 1) \int_0^{\lambda_2}  \Rc_{\Rm}(-w) dw.
\end{eqnarray}
The argument of the logarithm in (\ref{free-energy-step5}) can be
interpreted as an expectation with respect to $\vv_0, \ldots,
\vv_u$, with joint pdf $g_0(\vv_0) \prod_{a=1}^u g(\vv_a)$. By the
law of large numbers, this measure satisfies a concentration
property with respect to the empirical correlations
(\ref{empirical-corr}). Hence, we can invoke
Cram\'er's large
deviation theorem \cite{cramer}
as follows. Since $\Qm$ is a function of $\vv_0, \ldots, \vv_u$,  the conditional pdf of $\Qm$ given $\vv_0, \ldots, \vv_u$ is just a multi-dimensional delta
function (i.e., a product of delta functions), hence,  we can write
\begin{align}
& \int  d\vv_0 \cdots d\vv_u \left ( g_0(\vv_0) \prod_{a=1}^u g(\vv_a)  \right )  \times   \exp \left ( -n \Gc^{(u)} (\Qm) \right )   \\
& \; = \EE \left [ \int \exp \left ( -n \Gc^{(u)} (\Qm) \right )  \mu_n^{(u)}(d\Qm| \vv_0, \ldots, \vv_u) \right ]   \\
& \; =  \int \exp \left ( -n \Gc^{(u)} (\Qm) \right )  \mu_n^{(u)}(d\Qm) \label{thisquantity} \\
& \; \approx \int   \exp \left ( -n \left ( \Gc^{(u)} (\Qm)+ I^{(u)}(\Qm) \right ) \right ) d\Qm  \label{free-energy-step6}
\end{align}
where (\ref{free-energy-step6}) holds in the sense that, when we consider the quantity (\ref{thisquantity}) inside
the logarithm in the limit (\ref{free-energy-step5}), it can be replaced by (\ref{free-energy-step6}).

The rate function $I^{(u)}(\Qm)$  of the measure $\mu^{(u)}_n(d\Qm)$ defined as
\begin{eqnarray} \label{measure-mun}
\mu_n^{(u)}(d\Qm) & = &  \int  d\vv_0 \cdots d\vv_u \left ( g_0(\vv_0) \prod_{a=1}^u g(\vv_a)  \right )
\prod_{a\leq a'}^u \delta \left ( \sum_{k=1}^n v_{ak}v^*_{a'k} - n Q_{a,a'} \right ) \; d\Qm    \\
& = & \EE \left [ \prod_{a\leq a'}^u \delta \left ( \sum_{k=1}^n v_{ak}v^*_{a'k} - n Q_{a,a'} \right ) \right ] \; d\Qm,
\end{eqnarray}
is given by the Legendre-Fenchel  transform of the log-Moment Generating Function (log-MGF) of the random vector
$\underline{\Vm} = (V_0, V_1, \ldots, V_u)^\transp$, where
$V_0 = X_0 B_0$ and $V_a = X_a B_a$,  $X_0, X_1, \ldots, X_u$ are iid Gaussian RVs $\sim p_X(x) = \frac{1}{\pi \Pc_x} \exp(-|x|^2/\Pc_x)$,
and $B_0, B_1, \ldots, B_u$ are independent variables
with $B_0 \sim p_{B_0}$ and $B_a \sim q_B$. 
The MGF of $\underline{\Vm}$ is given by
\begin{eqnarray} \label{MGFV}
M^{(u)}(\tilde{\Qm}) = \EE \left [ \exp\left ( \underline{\Vm}^\dagger \tilde{\Qm} \underline{\Vm} \right ) \right ]
\end{eqnarray}
and the rate function is given by
\begin{eqnarray}
I^{(u)}(\Qm) = \sup_{\tilde{\Qm}} \left \{ \trace(\tilde{\Qm} \Qm ) - \log M^{(u)}(\tilde{\Qm}) \right \}
\end{eqnarray}
Eventually, using this into (\ref{free-energy-step6}) and the resulting expression in the limit (\ref{free-energy-step5})
and applying Varadhan's lemma, we arrive at the saddle-point condition
\begin{eqnarray} \label{free-energy-step7}
\lim_{n \rightarrow \infty} \frac{1}{n} \log \left ( \int  \mu^{(u)}_n(d\Qm)  \exp \left ( -n \Gc^{(u)}(\Qm) \right ) \right ) & = &
-\inf_{\Qm} \left \{ \Gc^{(u)}(\Qm)  + \sup_{\tilde{\Qm}} \left \{ \trace(\tilde{\Qm} \Qm ) - \log M^{(u)}(\tilde{\Qm}) \right \} \right \}   \nonumber  \\
& &  \\
& = & -\inf_{\Qm} \sup_{\tilde{\Qm}} \left \{ \Gc^{(u)}(\Qm)  + \trace(\tilde{\Qm} \Qm ) - \log M^{(u)}(\tilde{\Qm})  \right \} \nonumber \\
& &
\end{eqnarray}
Now we focus on the calculation of the MGF.
Under the RS assumption, the supremum in (\ref{free-energy-step7}) is achieved for $\tilde{\Qm}$ in the form
\begin{eqnarray} \label{QtildeRS}
\tilde{\Qm} = \left [ \begin{array}{cc}
c & d \onev^\transp  \\
d^* \onev & (g - f)\Id + f \onev\onev^\transp \end{array} \right ]
\end{eqnarray}
where $c,d,g,f$ are parameters.
Using the RS form for  $\tilde{\Qm}$ we obtain
\begin{eqnarray} \label{MGFV1}
M^{(u)}(\tilde{\Qm}) = \EE \left [ \exp \left ( \left | \frac{d}{\sqrt{f}} V_0 + \sqrt{f} \sum_{a=1}^u V_a \right |^2 + \left (c - \frac{|d|^2}{f}\right ) |V_0|^2 + (g-f) \sum_{a=1}^u |V_a|^2 \right ) \right ]
\end{eqnarray}
We use the complex circularly-symmetric version of the scalar Hubbard-Stratonovich transform
\cite{Stratonovich,Hubbard}:
\begin{eqnarray} e^{|x|^2} = \frac{\eta}{\pi} \int \exp \left ( - \eta |z|^2 + 2 \sqrt{\eta} \Re\{ x^*z\} \right ) dz \end{eqnarray}
for $x,z \in \CC$ and $\eta \in \RR_+$.  Choosing $\eta = |d|^2/f$, we obtain
\begin{eqnarray} 
\exp\left (  \left | \frac{d}{\sqrt{f}} V_0 + \sqrt{f} \sum_{a=1}^u V_a \right |^2 \right ) & = & 
\frac{|d|^2}{\pi f} \int \exp \left ( - \frac{|d|^2}{f} |z|^2 + 2 \frac{|d|}{\sqrt{f}}  \Re\left \{  \left (\frac{d}{\sqrt{f}} V_0 + \sqrt{f} \sum_{a=1}^u V_a \right )^*  z \right \} \right ) dz \nonumber \\
& & 
\end{eqnarray}
Using this into (\ref{MGFV1}), after some straightforward algebra, we find
\begin{eqnarray} \label{MGFV2}
M^{(u)}(\tilde{\Qm}) & = &  \EE \left [ \frac{|d|^2}{\pi f} \int \exp \left ( - \frac{|d|^2}{f} \left | z -  (d/|d|) V_0\right |^2 + c|V_0|^2 \right )
\exp \left ( \sum_{a=1}^u \left ( 2 |d| \Re\{ V_a^* z\} + (g - f) |V_a|^2 \right ) \right )  dz \right ]  \nonumber \\
& & 
\end{eqnarray}
Notice that $V_0$ has a circularly symmetric distribution, therefore $(d/|d|) V_0$ and $V_0$ are identically distributed. Hence, we can write
\begin{eqnarray} \label{MGFV3}
M^{(u)}(\tilde{\Qm}) & =  & \EE \left [ \frac{|d|^2}{\pi f} \int \exp \left ( - \frac{|d|^2}{f} \left | z -  V_0\right |^2 + c|V_0|^2 \right )
\exp \left ( \sum_{a=1}^u \left ( 2 |d| \Re\{ V_a^* z\} + (g - f) |V_a|^2 \right ) \right )  dz \right ] \nonumber \\
& & 
\end{eqnarray}
Since (\ref{MGFV3}) depends only on $|d|$, without loss of generality we re-define the parameter $d$ to be in $\RR_+$.
Also, notice from (\ref{MGFV3}) that $\lim_{u \rightarrow 0} M^{(u)}(\tilde{\Qm}) = 1$.

Following the replica derivation steps outlined at the beginning of
this section, we have to determine the saddle-point $\Qm^\star(u)$
and $\tilde{\Qm}^\star(u)$ achieving the extremal condition in
(\ref{free-energy-step7}), for general $u$, and finally replace the
result in (\ref{free-energy-step4}), differentiate with respect to
$u$ and let $u \rightarrow 0$. Since the function in (\ref{free-energy-step7}) is differentiable
and admits a minimum and a maximum, following the result of Appendix \ref{antonia-general-proof} we have
that determining the saddle-point $(\Qm^\star(u), \tilde{\Qm}^\star(u))$,
replacing it in (\ref{free-energy-step4}), differentiating the resulting expression with respect to $u$ and letting $u \rightarrow 0$
yields the same result of replacing in (\ref{free-energy-step7}) the saddle-point for $u = 0$,
denoted by $\Qm^\star(0) = \Qm^\star$ and $\tilde{\Qm}^\star(0) = \tilde{\Qm}^\star$,
differentiating the result with respect to $u$ and letting $u \rightarrow 0$, where now
$\Qm^\star, \tilde{\Qm}^\star$ are constants independent of $u$.

Differentiating (\ref{free-energy-step7}) with respect to $\tilde{\Qm}$, we obtain the equation
\begin{eqnarray} \label{saddle1}
\Qm = \frac{\EE \left [ \underline{\Vm} \underline{\Vm}^\dagger \exp \left (\underline{\Vm}^\dagger \tilde{\Qm} \underline{\Vm}\right ) \right ]}
{\EE \left [ \exp \left (\underline{\Vm}^\dagger \tilde{\Qm} \underline{\Vm}\right ) \right ]}
\end{eqnarray}
Since we evaluate the saddle-point conditions at $u \rightarrow 0$, and since the denominator in
(\ref{saddle1}) is $M^{(u)}(\tilde{\Qm})$, which is equal to 1 at $u \downarrow 0$, we can just disregard the denominator and focus on the
numerator in the following.  Using the expression (\ref{chandra-exponent}) for $\Gc^{(u)}(\Qm)$, with eigenvalues $\lambda_1$ and $\lambda_2$ given by
(\ref{eigenvalues}), and noticing that the RS conditions (\ref{QRS}) and (\ref{QtildeRS}) yield
\begin{eqnarray} \label{traceQQ}
\trace(\tilde{\Qm} \Qm ) = \epsilon_0 c + u \epsilon_1 g + 2 \Re\{\vartheta\} d u + u (u - 1) \omega f
\end{eqnarray}
we have that the whole exponent depends only on the real part of $\vartheta$. Therefore, we re-define $\vartheta$ to be a real parameter and
differentiate with respect to $\vartheta$, $\omega$, $\epsilon_1$ and $\epsilon_0$, and impose that
the partial derivatives are equal to zero. We find the conditions
\begin{eqnarray}
d & = & \frac{1}{\gamma^{-1} + u} \Rc_{\Rm}(-\lambda_1) \label{dstar} \\
f & = & \frac{1}{u} \left ( \gamma \Rc_{\Rm}(-\lambda_2)  - \frac{1}{\gamma^{-1} + u} \Rc_{\Rm}(-\lambda_1)   \right ) \label{fstar} \\
g - f & = & -\gamma \Rc_{\Rm}(-\lambda_2) \label{gstar} \\
c & = & - \frac{u}{\gamma^{-1} + u} \Rc_{\Rm}(-\lambda_1) \label{cstar}
\end{eqnarray}
Evaluating these conditions for $u \downarrow 0$ and noticing that, as $u$ vanishes,
$\lambda_1 \rightarrow \lambda_2$, we find:
\begin{eqnarray}
d^\star & = & \gamma \Rc_{\Rm}(-\lambda^\star_2) \label{dstar1} \\
f^\star & = & \lim_{u \rightarrow 0} \frac{1}{u} \left ( \gamma \Rc_{\Rm}(-\lambda^\star_2)  - \frac{1}{\gamma^{-1} + u} \Rc_{\Rm}(-\lambda^\star_1)   \right )   \\
& = &   \left .  - \frac{\partial}{\partial u}  \left [ \frac{1}{\gamma^{-1} + u} \Rc_{\Rm}(-\lambda^\star_1)  \right ] \right |_{u=0}   \\
& = &  \gamma^2 \Rc_{\Rm}(-\lambda^\star_2) + \gamma^2 \dot{\Rc}_{\Rm}(-\lambda^\star_2) \left ( \epsilon^\star_0 - 2\vartheta^\star + \omega^\star - \gamma(\epsilon^\star_1 - \omega^\star) \right )
\label{fstar1} \\
g^\star - f^\star & = & - d^\star \label{gstar1} \\
c^\star & = & 0 \label{cstar1}
\end{eqnarray}
where $\lambda_2^\star = \gamma(\epsilon_1^\star - \omega^\star)$ and where $\dot{\Rc}_{\Rm}(\cdot)$ 
denotes the first derivative of $\Rc_{\Rm}(\cdot)$.

The conditions for $\epsilon_0^\star, \epsilon_1^\star, \vartheta^\star, \omega^\star$ in terms of $d^\star, g^\star$ and $f^\star$ are obtained
from (\ref{saddle1}), recalling that, by definition, $\epsilon_0 = Q_{00}, \epsilon_1 = Q_{11}, \vartheta = Q_{01}$ and $\omega = Q_{12}$.
In order to obtain more useful expressions for these parameters,  we use (\ref{cstar1}) and (\ref{gstar1}) in (\ref{MGFV3}) and write
\begin{eqnarray} \label{MGFV4}
M^{(u)}(\tilde{\Qm}^\star) & =  &
\EE \left [ \frac{(d^\star)^2}{\pi f^\star} \int \exp \left ( - \frac{(d^\star)^2}{f^\star} \left | z -  V_0\right |^2 \right )
\exp \left ( \sum_{a=1}^u \left ( 2 d^\star \Re\{ V_a^* z\}  - d^\star |V_a|^2 \right ) \right )  dz \right ]    \nonumber \\
& &   \\
& = & \EE \left [ \frac{(d^\star)^2}{\pi f^\star} \int \exp \left ( - \frac{(d^\star)^2}{f^\star} \left | z -  V_0\right |^2 \right )
\exp \left ( - d^\star \sum_{a=1}^u \left | z - V_a \right |^2 \right )  e^{ d^\star u |z|^2} dz \right ]    \\
& = & \EE \left [ \frac{\eta}{\pi} \int \exp \left ( - \eta \left | z -  V_0\right |^2 \right )
\exp \left ( - \xi \sum_{a=1}^u \left | z - V_a \right |^2 \right )  e^{ \xi |z|^2 u} dz \right ]    \\
\end{eqnarray}
where we define $\eta = (d^\star)^2/f^\star$ and $\xi = d^\star$.

Focusing on the numerator in (\ref{saddle1}) and following steps similar to the derivation of (\ref{MGFV4}), we obtain the following expressions for the correlation
coefficients $\epsilon_0, \vartheta, \epsilon_1, \omega$:
\begin{enumerate}
\item For $\epsilon_0 = Q_{00}$ we have
\begin{eqnarray}
\epsilon_{0}^\star & = &  \EE \Big [ \int \frac{\eta}{\pi}  |V_0|^2 \exp\left ( - \eta |z - V_0|^2  \right ) \cdot
\exp \left ( - \xi \sum_{a=1}^u |z  - V_a |^2 \right ) \; e^{\xi |z|^2 u}   dz \Big ] \Big |_{u \downarrow 0}   \\
& = & \EE[|V_0|^2]
\end{eqnarray}

\item For $\vartheta = Q_{01}$, we introduce the RV $V \sim g(\cdot)$ (same distribution as any of the $V_a$'s) and independent of $V_0, V_1$. Then, we can write
\begin{align} \label{zio1}
& \vartheta^\star =   \\
& = \EE \Big [ \int \frac{\eta}{\pi}  V_0  \exp\left ( - \eta |z - V_0|^2  \right ) \cdot V^*_1 \exp \left ( - \xi  |z  - V_1|^2 \right ) \cdot
\exp \left ( - \xi \sum_{a=2}^u |z  - V_a|^2 \right )  \; e^{\xi |z|^2 u}   dz \Big ] \Big |_{u \downarrow 0}   \\
& = \EE \Big [ \int \frac{\eta}{\pi}  V_0  \exp\left ( - \eta |z - V_0|^2  \right ) \cdot V^*_1 \exp \left ( - \xi  |z  - V_1 |^2 \right ) \cdot \EE \left [ \exp \left ( - \xi |z  - V|^2 \right ) \right ]^{u-1}  \; e^{\xi |z|^2 u}   dz \Big ] \Big |_{u \downarrow 0}   \\
& = \EE \Big [ \int \frac{\eta}{\pi}  V_0  \exp\left ( - \eta |z - V_0|^2  \right ) \cdot
\EE \left [ \frac{V^* \exp \left ( - \xi  |z  - V|^2 \right )}{ \EE \left [ \exp \left ( - \xi |z  - V|^2 \right ) \right ]} \right ] dz \Big ]
\end{align}

\item For $\epsilon_1 = Q_{11}$ we have:
\begin{align} \label{zio2}
& \epsilon_1^\star =   \\
& = \EE \Big [ \int \frac{\eta}{\pi}  \exp\left ( - \eta |z - V_0|^2  \right ) \cdot |V_1|^2  \exp \left ( - \xi  |z  - V_1 |^2 \right ) \cdot
\exp \left ( - \xi \sum_{a=2}^u |z  - V_a |^2 \right )  \; e^{\xi |z|^2 u}   dz \Big ] \Big |_{u \downarrow 0}   \\
& = \EE \Big [ \int \frac{\eta}{\pi}  \exp\left ( - \eta |z - V_0|^2  \right ) \cdot |V_1|^2 \exp \left ( - \xi  |z  - V_1 |^2 \right ) \cdot
\EE \left [ \exp \left ( - \xi |z  - V|^2 \right )  \right ]^{u-1}  \; e^{\xi |z|^2 u}   dz \Big ] \Big |_{u \downarrow 0}   \\
& = \EE \Big [ \int \frac{\eta}{\pi}  \exp\left ( - \eta |z - V_0|^2  \right ) \cdot
\EE \left [  \frac{|V|^2 \exp \left ( - \xi  (z  - V )^2 \right )}{ \EE \left [  \exp \left ( - \xi |z  - V|^2 \right ) \right ]}  \right ] dz \Big ]
\end{align}

\item For $\omega = Q_{12}$ we have:
\begin{align} \label{zio3}
& \omega^\star =    \\
& = \EE \Big [ \int \frac{\eta}{\pi}
\exp\left ( - \eta |z - V_0|^2  \right )  V_1 V^*_2 \exp \left ( - \xi  |z  - V_1 |^2 \right ) \exp \left ( - \xi  |z  - V_2 |^2 \right )  \cdot   \\
& \;\;\;\;\;\;\;\;\;\;\;\;\;\;\;\;\;\;\;\;\;\;\;\;\;\;\;\;\;\;\;\;\;\;\;\;\;\;\;\;\;\;\;\;\;\;\;\;\;\;\;\;\;\;\;\;\;\;\;\;\;\;\; \cdot \exp \left ( - \xi \sum_{a=3}^u |z  - V_a |^2 \right )  \; e^{\xi |z|^2 u}   dz \Big ] \Big |_{u \downarrow 0}   \\
& = \EE \Big [ \int \frac{\eta}{\pi}  \exp\left ( - \eta |z - V_0|^2  \right ) V_1 V^*_2 \exp \left ( - \xi  |z  - V_1 |^2 \right ) \exp \left ( - \xi |z  - V_2 |^2 \right )  \cdot   \\
& \;\;\;\;\;\;\;\;\;\;\;\;\;\;\;\;\;\;\;\;\;\;\;\;\;\;\;\;\;\;\;\;\;\;\;\;\;\;\;\;\;\;\;\;\;\;\;\;\;\;\;\;\;\;\;\;\;\;\;\;\;\;\; \cdot
\EE \left [ \exp \left ( - \xi |z  - V|^2 \right )  \right ]^{u-2}  \; e^{\xi |z|^2 u}   dz \Big ] \Big |_{u \downarrow 0}   \\
& = \EE \Big [ \int \frac{\eta}{\pi}  \exp\left ( - \eta |z - V_0|^2  \right ) \cdot
\left | \EE \left [ \frac{V \exp \left ( - \xi  |z  - V |^2 \right )}{ \EE \left [ \exp \left ( - \xi |z  - V|^2 \right ) \right ]}  \right ] \right |^2  dz \Big ]
\end{align}
\end{enumerate}
Finally, we define a single-letter joint probability distribution and restate the expectations appearing in 
(\ref{zio1}), (\ref{zio2}), (\ref{zio3})
in terms of this new single-letter model.  
Let $p_{V_0}(v_0)$ denote the Bernoulli-Gaussian density of $V_0$, 
induced by $p_X(\cdot)$ and by $p_{B_0}(\cdot)$, and let
\begin{eqnarray} \label{decoupled-p}
p_{Y|V_0; \eta}(y|v_0; \eta) = \frac{\eta}{\pi} \exp \left ( - \eta |y - v_0|^2 \right )
\end{eqnarray}
denote the transition probability density of the complex (scalar) circularly symmetric AWGN channel
\begin{eqnarray}
Y = V_0 + \eta^{-\frac12} Z
\end{eqnarray}
with $Z \sim \Cc\Nc(0,1)$.
Also, define the conditional complex circularly symmetric Gaussian pdf
\begin{eqnarray} \label{decoupled-q}
q_{Y|V; \xi}(y|v; \xi) = \frac{\xi}{\pi} \exp \left ( - \xi |y - v|^2 \right )
\end{eqnarray}
and, using Bayes rule,  consider the a-posteriori probability distribution
\begin{eqnarray} \label{decoupled-posterior-q}
q_{V|Y; \xi} (v | y ; \xi) & = &
\frac{q_{Y|V; \xi}(y|v; \xi) g(v)}{\int q_{Y|V; \xi}(y|v; \xi) g(v) dv}   \\
& = & \frac{\exp\left ( - \xi |y - v|^2\right )  g(v)}{\EE\left [ \exp\left ( - \xi |y - V|^2\right )  \right ]}.
\end{eqnarray}
The joint single-letter probability distribution of interest for the variables $V_0, Y$ and $V$ is given by
\begin{eqnarray} \label{decoupled-joint}
p_{V_0}(v_0) p_{Y|V_0; \eta}(y|v_0; \eta) q_{V|Y; \xi}(v|y; \xi).
\end{eqnarray}
This explains the decoupled channel single-letter probability model  (\ref{joint-decoupled-V}).

Now, we can define the conditional mean of $V$ given $Y$ as
\begin{eqnarray} \label{decoupled-posterior-mean}
\EE[V|Y=y] & = & \int v \; q_{V|Y; \xi} (v | y ; \xi) dv   \\
& = & \EE \left [ \frac{V \exp\left ( - \xi |y - V|^2\right )}{\EE\left [ \exp\left ( - \xi |y - V|^2\right )  \right ]} \right ].
\end{eqnarray}
The corresponding conditional second moment is given by
\begin{eqnarray} \label{decoupled-posterior-second-moment}
{\sf sm}_V(y; \xi) & = & \int |v|^2 \; q_{V|Y; \xi} (v | y ; \xi) dv   \\
& = & \EE \left [ \frac{|V|^2 \exp\left ( - \xi |y - V|^2\right )}{\EE\left [ \exp\left ( - \xi |y - V|^2\right )  \right ]} \right ].
\end{eqnarray}
At this point, it is easy to identify the terms and write the expressions (\ref{zio1}), (\ref{zio2}), (\ref{zio3}) in terms of expectations with respect to
the single-letter joint probability measure defined in (\ref{decoupled-joint}). We have
\begin{eqnarray} \label{decoupled-parameters}
\epsilon_0^\star & = & \EE[|V_0|^2] \label{puppa0} \\
\vartheta^\star & = & \EE \left [ V_0  (\EE[V|Y])^*  \right ] \label{puppa1} \\
\epsilon_1^\star & = & \EE \left [  {\sf sm}_V(Y; \xi)   \right ] \label{puppa2} \\
\omega^\star & = & \EE \left [  \left | \EE[V|Y] \right |^2 \right ]. \label{puppa3}
\end{eqnarray}
In order to obtain the desired fixed-points equations for the saddle-point that defines the result in (\ref{free-energy-step7}), we notice that
\begin{eqnarray} \label{decoupled-mse}
\epsilon_0^\star - 2\vartheta^\star + \omega^\star = \EE \left [ \left |V_0 - \EE[V|Y] \right |^2 \right ]
\end{eqnarray}
and that
\begin{eqnarray} \label{decoupled-varq}
\epsilon_1^\star - \omega^\star & = & \EE \left [\left  |V - \EE[V|Y] \right |^2 \right ]   \\
& = & {\sf mmse}(V|Y).
\end{eqnarray}
Using (\ref{dstar}), (\ref{fstar}), the equality $\lambda_2^\star = \gamma(\epsilon_1^\star - \omega^\star)$, and recalling 
that $\xi = d^\star$ and $\eta = (d^\star)^2/f^\star$, we arrive at the system of fixed-point equations
(\ref{saddle-final-chi}) -- (\ref{saddle-final-eta}).
In the matched case, where $q_B(\cdot) =  p_B(\cdot)$ and $\gamma = 1$ we immediately obtain 
that $\delta = \chi$ and therefore $\xi = \eta$, and the fixed-point equations reduce to  
(\ref{e:fix-pointeq1}) -- (\ref{e:fix-pointeq}) in Claim \ref{th:x}.

Using the values solution of (\ref{saddle-final-chi}) -- (\ref{saddle-final-eta})
into (\ref{free-energy-step7}), using the trace expression (\ref{traceQQ}) and finally putting everything together 
into (\ref{free-energy-step4}) and taking the derivative w.r.t. $u$
evaluated at $u \downarrow 0$, we eventually obtain the free energy $\Ec$ in (\ref{free-energy-def}) as given by
\begin{eqnarray} \label{free-energy-babau}
\Ec
& = &
\log(\pi/\gamma) + \gamma + \frac{\partial}{\partial u} \left . \left \{  \int_0^{\lambda_1^\star} \Rc_{\Rm}(-w) dw + (u-1) \int_0^{\lambda_2^\star} \Rc_{\Rm}(-w)dw \right \} \right |_{u \downarrow 0}   \\
& & + \left . \frac{\partial}{\partial u} \Big \{ \epsilon^\star_0 c^\star + u \epsilon^\star_1 g^\star + 2 \vartheta^\star d^\star u + u (u - 1) \omega^\star f^\star \Big \}
\right |_{u \downarrow 0}   \\
& & - \frac{\partial}{\partial u} \left .   \log M^{(u)}(\tilde{\Qm}^\star) \right |_{u \downarrow 0}
\end{eqnarray}
Examining each term separately, we have:
\begin{align}  \label{antonia}
& \displaystyle{\frac{\partial}{\partial u} \left . \left \{  \int_0^{\lambda_1^\star} \Rc_{\Rm}(-w) dw + (u-1) \int_0^{\lambda_2^\star} \Rc_{\Rm}(-w)dw \right \} \right |_{u \downarrow 0}}   \\
& = \displaystyle{ \left . \Rc_{\Rm}(-\lambda_1^\star) \frac{(\epsilon_0^\star - 2\vartheta^\star + \omega^\star)(\gamma^{-1} + u) -
(\epsilon_1^\star - \omega^\star + u(\epsilon_0^\star - 2\vartheta^\star + \omega^\star))}{(\gamma^{-1} + u)^2}  \right |_{u \downarrow 0} + \int_0^{\lambda_2^\star} \Rc_{\Rm}(-w)dw}   \\
& = \displaystyle{ \gamma \Rc_{\Rm}(-\lambda_2^\star) \left (\epsilon_0^\star - 2\vartheta^\star + \omega^\star - \gamma (\epsilon_1^\star - \omega^\star)\right ) + \int_0^{\lambda_2^\star} \Rc_{\Rm}(-w)dw}   \\
& = \displaystyle{ \gamma \Rc_{\Rm}(-\chi) (\delta -  \chi ) + \int_0^{\chi} \Rc_{\Rm}(-w)dw}
\end{align}
where we have used  the definition of $\chi$ and $\delta$ in (\ref{saddle-final-chi}) and (\ref{saddle-final-delta}), respectively, and the relations
(\ref{decoupled-mse}) and (\ref{decoupled-varq}).
For the trace term, recalling that $c^\star = 0$, we have
\begin{eqnarray} \label{antonia1}
\left . \frac{\partial}{\partial u}  \Big \{ \epsilon^\star_0 c^\star + u \epsilon^\star_1 g^\star + 2 \vartheta^\star d^\star u + u (u - 1) \omega^\star f^\star \Big \}
\right |_{u \downarrow 0} & = & \epsilon^\star_1 g^\star + 2 \vartheta^\star d^\star - \omega^\star f^\star
\end{eqnarray}
Finally, for the log-MGF term we use (\ref{MGFV4}) and performing the expectation with respect to $V_1, \ldots, V_u$ (independent and
identically distributed as $V$) first, we obtain
\begin{eqnarray} \label{MGFV5}
M^{(u)}(\tilde{\Qm}^\star) & =  &
\EE \left [ \frac{\eta}{\pi} \int \exp \left ( - \eta \left | z -  V_0\right |^2 \right )   \left (
\EE \left [  \exp \left ( - \xi  \left | z - V \right |^2 \right ) \right ]  e^{ \xi |z|^2} \right )^u  dz \right ]
\end{eqnarray}
Hence,
\begin{eqnarray} \label{antonia2}
- \frac{\partial}{\partial u} \left .   \log M^{(u)}(\tilde{\Qm}^\star) \right |_{u \downarrow 0}  & = &
- \EE \left [ \frac{\eta}{\pi} \int \exp \left ( - \eta \left | z -  V_0\right |^2 \right )  \log \left ( \EE \left [  \exp \left ( - \xi  \left | z - V \right |^2 \right ) \right ] \right ) dz \right ]   \\
& & - \xi \EE \left [ \frac{\eta}{\pi} \int |z|^2 \exp \left ( - \eta \left | z -  V_0\right |^2 \right ) dz \right ]   \\
& = &  - \EE \left [  \log \left ( q_{Y; \xi}(Y)  \right ) \right ] + \log \frac{\xi}{\pi} - \frac{\xi}{\eta} - \xi \EE[|V_0|^2]
\end{eqnarray}
where in the last line we use (\ref{decoupled-q}) and define
\begin{eqnarray} \label{decoupled-q1}
q_{Y; \xi} (y) & = & \int  q_{Y|V;\xi} (y|v) g(v) dv   \\
& = &  \frac{\xi}{\pi} \EE \left [ \exp \left ( - \xi |y - V|^2 \right ) \right ].
\end{eqnarray}
It is understood that if (\ref{saddle-final-chi}) -- (\ref{saddle-final-eta})
have multiple solutions, then the solution that minimizes the free energy
should be chosen.

We conclude by showing that (\ref{free-energy-babau}) can be written in the form (\ref{free-energy-babau-final}).
Putting together (\ref{antonia1})  and the last term of (\ref{antonia2}) and recalling that $\EE [|V_0|^2 ] = \epsilon^\star_0$
and that $\xi= d^\star$  we have:
\begin{eqnarray} \epsilon^\star_1 g^\star + 2 \vartheta^\star d^\star - \omega f^\star - \epsilon^\star_0 \xi  =
\epsilon^\star_1 g^\star + 2 \vartheta^\star d^\star - \omega^\star f^\star -  \epsilon_0^\star d^\star. \end{eqnarray}
Adding and subtracting $\omega d^\star$ and using (\ref{decoupled-mse}) and the definition of $\delta$ (see (\ref{saddle-final-delta})) we have
\begin{eqnarray}
-  \epsilon_0^\star d^\star  + 2 \vartheta^\star d^\star  - \omega^\star d^\star + \omega^\star d^\star - \omega^\star f^\star  + \epsilon^\star_1 g^\star
& =  & (-  \epsilon_0^\star  + 2 \vartheta^\star  - \omega^\star) d^\star + \omega^\star d^\star - \omega^\star f^\star  + \epsilon^\star_1 g^\star    \nonumber \\
& & \\
& = & - \delta  d^\star  -  \omega^\star  (f^\star - d^\star)   + \epsilon^\star_1 g^\star.
\end{eqnarray}
Recalling that $g^\star= f^\star-d^\star$, we obtain
\begin{eqnarray} - \delta  d^\star  +  (\epsilon^\star_1  - \omega^\star  ) g^\star \end{eqnarray}
Recalling that  $\epsilon^\star_1  - \omega^\star = \lambda^\star_2/\gamma  = \chi/\gamma$, we get
\begin{eqnarray} -  \delta  d^\star  +  g^\star \chi/\gamma \end{eqnarray}
Finally, using $d^\star = \xi$ and $\eta = (d^\star)^2/f^\star$ we arrive at
\begin{eqnarray} \label{antonia3}
-  \delta \xi  +  d^\star ( f^\star/ d^\star -1) \chi/\gamma
=   -   \delta \xi  +  d^\star ( (f^\star d^\star)/  (d^\star)^2 -1) \chi/\gamma =  - \delta \xi  +  \xi  (\xi / \eta  -1)  \chi/\gamma
\end{eqnarray}
Next, we use (\ref{antonia3}), the remaining terms of (\ref{antonia2}), (\ref{antonia}) and
the first terms in (\ref{free-energy-babau}), together with the saddle-point equations  
(\ref{saddle-final-chi}) -- (\ref{saddle-final-eta}),
to eventually obtain $\Ec$ in the form (\ref{free-energy-babau-final}).

For the case $q_B(\cdot) = p_B(\cdot)$ and $\gamma = 1$, noticing that 
$\delta = \chi$ and $\xi = \eta$, with $\chi$ and $\eta$ 
given by (\ref{e:fix-pointeq1}) -- (\ref{e:fix-pointeq}), the free energy takes on the form
\begin{eqnarray}  \label{free-energy-babau-final-matched}
\Ec  =  I\left (V_0; V_0 +  \eta^{-\frac{1}{2}} Z \right )  + \int_0^\chi \left (\Rc_{\Rm}(-w) - \eta \right ) dw + \log (e\pi),
\end{eqnarray}
where $Z \sim \Cc\Nc(0,1)$ and where  we used the fact that when $\xi = \eta$ and $V \sim V_0$, 
then $q_{Y; \xi}(y) = p_{Y; \eta}(y) = \frac{\pi}{\eta}  \EE \left [ \exp(-\eta(y- V_0)^2) \right ]$, so that
\begin{eqnarray} 
- \EE \left [ \log q_{Y;\xi}(Y) \right ] - \log \frac{e\pi}{\xi} = h\left (V_0 + \eta^{-\frac{1}{2}} Z \right ) - h\left (\eta^{-\frac{1}{2}} Z \right ) = I\left (V_0; V_0 +  \eta^{-\frac{1}{2}}   Z\right ).  
\end{eqnarray}
Using (\ref{free-energy-babau-final-matched}) in the mutual information expression (\ref{info1}) 
we obtain (\ref{1l}) in Claim \ref{th:x}.

\section{Proof of Theorem \ref{th:y}}\label{proof:y}

We start by recalling some transforms in random matrix theory and some related results from \cite{fnt}.

\begin{definition}
\label{def:eta} The $\eta$-transform of a nonnegative random
variable $X$ is
 \begin{eqnarray}
 \eta_{X} ( s ) = \E \left[ \frac{1}{1 + s X } \right]
 \label{etadef}
 \end{eqnarray}
with $s \geq 0$.  \hfill $\lozenge$
\end{definition}
Note that
\begin{eqnarray}\label{boeta}
\PP(X = 0) <\eta_X ( s ) \leq 1
\end{eqnarray}
with the lower bound asymptotically tight as $s \rightarrow
\infty$.

\begin{definition}
\label{def:shannon} The Shannon transform of a nonnegative random
variable $X$ is defined as
\begin{eqnarray}
{\mathcal V}_{X} ( s ) = \EE [ \log ( 1 + s X ) ]
\label{shtr}
\end{eqnarray}
with $s \geq 0$. \hfill $\lozenge$
\end{definition}

Assuming that the logarithm in (\ref{shtr}) is natural, the $\eta$
and Shannon transforms are related through
\begin{eqnarray}\label{link}
\frac{d}{ds} \Vc_X(s) = \frac{1 - \eta_X(s)}{s}
\end{eqnarray}
Also, it is useful to recall here the definition of the S-transform
of free probability (see \cite{fnt} and references therein),
which is used in some of the proofs that follow.

\begin{definition} \label{def:S}
The S-transform of a nonnegative random variable $X$ is defined as
\begin{eqnarray}\label{defS}
\Sigma_{X} (z) &=& - \frac{z+1}{z} \eta_{X}^{-1} ( z+1 )
\end{eqnarray}
where $\eta_{X}^{-1}(\cdot)$ denotes the inverse function of the
$\eta$-transform. \hfill $\lozenge$
\end{definition}

It is common to denote the $\eta$-transform, the Shannon transform
and the S-transform of the  spectral distribution of a sequence of
nonnegative-definite $n \times n$ random matrices $\Bm$, for $n
\rightarrow \infty$,  by $\eta_\Bm(\cdot)$, $\Vc_\Bm(\cdot)$ and
$\Sigma_{\Bm}(\cdot)$, respectively. In this case, the lower bound
in \eqref{boeta} corresponds to the limiting fraction of zero
eigenvalues of $\Bm$.

\begin{theorem}
Let ${\Am}$ and ${\Bm}$ be  nonnegative asymptotically free random
matrices, then for $0 < s < 1 $,
\begin{eqnarray}
\label{etainp} \eta_{\bf A B}^{-1} ( s ) = \frac{s}{1 -
s} \, \eta_{\bf A  }^{-1} ( s ) \, \eta_{\bf   B}^{-1} (
s )
\end{eqnarray}
\hfill \QED
\end{theorem}
In addition, the following implicit relation is also useful:
\begin{eqnarray}
\label{newark} \eta_{\Am \Bm}(s)= \eta_{\Am} \left (
\frac{s}{\Sigma_{\Bm}(\eta_{\Am \Bm}(s) -1 )}  \right)
\end{eqnarray}

The next two results are instrumental to the proof of Theorem \ref{th:y}. While
they might have appeared elsewhere, a simple and self-contained proof is given 
here for the sake of completeness.

\begin{theorem}\label{thm:finalapp}
Let ${\Am}$ and ${\Bm}$ be  nonnegative asymptotically free random
matrices. For $s \geq 0$, let  $(\eta,\alpha,\nu)$ be the
solution of the system of equations:
\begin{eqnarray}
\eta &=& \eta_{\bf A} \left ( \alpha\, s \right)
\label{eta-equivalent}
\\
\eta&=& \eta_{\bf B}\left ( \nu s\right) \label{putomadrid}
\\
\eta &=& \frac{1}{1+\alpha\,  \nu s} \label{putomilan}
\end{eqnarray}
Then, the $\eta$-transform of ${\Am} \Bm$ is given by
\begin{eqnarray}
\label{gesu} \eta_{{\Am} \Bm }( s ) = \eta.
\end{eqnarray}
\end{theorem}

\begin{IEEEproof} Letting $\eta_{{\Am} \Bm }( s ) = \eta(s)$ for simplicity of notation and using \eqref{newark},  we have:
\begin{eqnarray} \label{atroops}
\eta(s)= \eta_{\Am }( \alpha\,  s)
\end{eqnarray}
where
\begin{eqnarray} \label{atroops1}
\alpha=\frac{1}{\Sigma_{\bf B}(\eta(s) -1 )}
\end{eqnarray}
which is equivalent, using Definition \ref{def:S}, to:
\begin{eqnarray} \label{atroops2}
\eta_{\bf B}\left(\frac{1}{\alpha}\left( \frac{1}{\eta(s)}-1 \right)
\right) =\eta(s)
\end{eqnarray}
Letting $$ \nu= \frac{1}{ \alpha\,  s}\left( \frac{1}{\eta(s)}-1
\right),
 $$
from  \eqref{atroops} and  \eqref{atroops2}, Theorem
\ref{thm:finalapp} follows immediately. \end{IEEEproof}

As a consequence of Theorem \ref{thm:finalapp},  we have:
\begin{theorem}\label{thm:finalapp1}
Let ${\Am}$ and ${\Bm}$ be  nonnegative asymptotically free random
matrices. The Shannon-transform of ${\Am}{\Bm}$  is given by
\begin{eqnarray}
{\mathcal V}_{{\Am}{\Bm}}(  s )  =  {\mathcal V}_{{\Am}}\left(
\alpha\,  s  \right) + {\mathcal V}_{{\Bm}}\left(\nu s
\right)  - \log ( 1 + \alpha\,  \nu s ) \label{eurocopa}
\end{eqnarray}
where $\alpha$ and $\nu$ are the solutions of the system of equations
(\ref{eta-equivalent}) - (\ref{putomilan}), which depend on $s$.
\end{theorem}

\begin{IEEEproof}
The proof  follows an idea originated in
\cite{verdu-shamai-random-cdma-fading} to write the Shannon
transform when the $\eta$-transform is given as the solution of a
fixed-point equation: for any differentiable function $f$, the
definition of the Shannon transform of an arbitrary nonnegative
random variable $X$ leads to
\begin{eqnarray} \label{play}
\frac{\intd}{\intd x} \mathcal{V}_X ( s f(s) ) = \EE \left[ 
\frac{(s \dot{f} (s)  + f(s)) X}{1 + s f(s) X}\right],
\end{eqnarray}
where the ``dot'' here denotes differentiation with respect to the variable $s$.
Since both sides of (\ref{eurocopa}) are equal to zero at
$s =0$, it is sufficient to show that the derivatives with respect to $s$ of both sides of (\ref{eurocopa})
 coincide. Letting $\Asf$ and $\Bsf$ denote random variables distributed according to the spectral distribution of
 $\Am$ and $\Bm$, respectively, differentiating w.r.t. $s$ the difference of the right side minus the left side of (\ref{eurocopa}) yields
 \begin{eqnarray}  
&& \EE\left [  \frac{( \dot{\alpha}  s  + \alpha ) \Asf }{ 1
+ \alpha\,  s  \Asf }   \right ] + \EE \left[  \frac{(
\dot{\nu} s + \nu)\Bsf }{ 1 + \nu  s
\Bsf }  \right] - \frac{\alpha\,  \nu + \alpha\,  \dot{\nu} s +
\dot{\alpha} \nu s}{1 + \alpha\,  \nu s} - \frac{1 -
\eta_{{\Am}{\Bm}} (s)}{s}
\label{fabregas}\\
&=& \frac{\dot{\alpha}  s  + \alpha}{\alpha\,  s}\left( 1 -
\eta_{\Am} (\alpha\,  s)\right) + \frac{ \dot{\nu}
s + \nu}{\nu  s}\left( 1 - \eta_{\Bm} (\nu s )
\right) - \frac{\alpha\,  \nu + \alpha\,  \dot{\nu} s + \dot{\alpha}
\nu s}{1 + \alpha\,  \nu s} - \frac{1 - \eta_{{\Am}{\Bm}}
(s)}{s}
  \\
& & \label{torres}\\
&=& \frac{\dot{\alpha}  s  + \alpha}{\alpha\,  s}\left( 1 -
\eta\right) + \frac{ \dot{\nu} s + \nu}{\nu  s}\left( 1 -
\eta\right) - \frac{\alpha\,  \nu + \alpha\,  \dot{\nu} s +
\dot{\alpha} \nu s}{1 + \alpha\,  \nu s} - \frac{1 -
\eta}{s}
\label{villa}\\
&=& \left( \alpha\,  \nu + \alpha\,  \dot{\nu} s + \dot{\alpha} \nu
s \right) \left( \frac{1 - \eta}{\alpha\,  \nu s} - \eta
\right)
\\
\label{iker}&=& 0
\end{eqnarray}
where used (\ref{link})
to write the left side of (\ref{fabregas});
the right side of (\ref{fabregas}) follows from the definition of the
$\eta$-transform; (\ref{villa}) follows from Theorem
\ref{thm:finalapp} for $(\eta, \alpha,\nu)$ solutions of (\ref{eta-equivalent}) - (\ref{putomilan});
and \eqref{iker} follows again from the equality in (\ref{putomilan}).
\end{IEEEproof}

Theorem \ref{th:y} now follows as an application of Theorem \ref{thm:finalapp1} by identifying the terms.
We write
\begin{eqnarray}
\Ic_2 & = & \lim_{n\rightarrow\infty} \frac{1}{n} I(\xv;\yv|\Am,\Um,\bv) \\
& = & \lim_{n\rightarrow\infty} \frac{1}{n} \EE \left [ \log \det \left ( \Id + \Pc_x \Am \Um \Bm \Bm^\dagger \Um^\dagger \Am^\dagger \right ) \right ] \\
& = & \lim_{n\rightarrow\infty} \frac{1}{n} \EE \left [ \log \det \left ( \Id + \Pc_x \Um^\dagger \Am^\dagger \Am \Um \Bm \Bm^\dagger  \right ) \right ] \\
& = & \lim_{n\rightarrow\infty} \frac{1}{n} \EE \left [ \log \det \left ( \Id + \Pc_x \Rm  \Bm \Bm^\dagger   \right ) \right ] \\
& = & \Vc_{\Rm \Bm \Bm^\dagger}(\Pc_x) \\
& = & \Vc_{\Rm}(\alpha\,  \Pc_x) + \Vc_{\Bm \Bm^\dagger}(\nu\Pc_x) - \log (1 + \alpha\,  \nu \Pc_x)
\end{eqnarray}
where $(\alpha,\nu)$ are solutions of (\ref{eta-equivalent}) - (\ref{putomilan}) after replacing $\Am$ by $\Rm$,
$\Bm$ by $\Bm \Bm^\dagger$ and $\gamma$ by $\Pc_x$.
The final expressions (\ref{2I}) and  (\ref{e:fix-pointeq-alpha-nu}) follow by noticing that the spectral distribution of $\Bm$ has only two
mass points at zero and at one, with probabilities $1 - q$ and $q$, respectively.

\section{Proof of Theorems \ref{th:bounds} and \ref{LB-theorem}}\label{proof:bounds}

Notations are as in Section \ref{sc:setup}, following the observation model (\ref{model1}). In particular,
we let $\Xm = \diag(\xv)$ and $\Bm = \diag(\bv)$, and $\vv = \Xm \bv = \Bm \xv$.

{\bf Proof of bound (\ref{minfo-bound1}):}
We have
\begin{eqnarray} \label{minfo-bound-proof1}
I(\vv;\yv | \Am,\Um)
& \leq & \EE \left [ \log \det \left ( \Id   +    q \Pc_x \Am \Um \Um^\dagger \Am^\dagger\right ) \right ]
\end{eqnarray}
where the inequality follows by the fact that, conditionally on $\Am, \Um$, the differential entropy of $\yv = \Am \Um \vv + \zv$ for assigned covariance
\begin{eqnarray} \label{sAm}
 \EE[ \yv \yv^\dagger |\Am,\Um] = \Id + q \Pc_x \Am \Um \Um^\dagger \Am^\dagger
\end{eqnarray}
is maximized by a Gaussian complex circularly symmetric distribution $\yv \sim \Cc\Nc(\zerov,  \Id + q \Pc_x \Am \Um \Um^\dagger \Am^\dagger)$.
Recalling the definition of $\Rm$ in (\ref{RR}), we have
\begin{eqnarray} \label{shannon-transf}
\lim_{n \rightarrow \infty} \frac{1}{n}  \EE \left [ \log \det \left ( \Id   +    q \Pc_x \Am \Um \Um^\dagger \Am^\dagger\right ) \right ] & = &
\EE[ \log (1 + q\Pc_x |{\sf R}|^2 )]   \\
& = & \Vc_{\Rm}(q\Pc_x),
\end{eqnarray}
from the definition of Shannon transform.  Hence, (\ref{minfo-bound1}) follows.

{\bf Proof of bound (\ref{minfo-bound3}):}
This bound can be regarded as a ``matched filter bound'' on the vector channel with input $\vv = \Bm\xv$ and output $\yv$.  We can write
\begin{eqnarray}
I(\vv; \yv|\Am,\Um) & = & I(\vv; \Am\Um \vv + \zv |\Am,\Um)  \\
& = & \sum_{i=1}^n I \left (V_i ; \Am \Um \vv + \zv | \Am, \Um, V_1^{i-1} \right ) \\
& \leq & \sum_{i=1}^n I \left (V_i ; \Am \Um \vv + \zv, V_{i+1}^n  | \Am, \Um, V_1^{i-1} \right ) \\
& \stackrel{(a)}{=}  & \sum_{i=1}^n I \left (V_i ; \Am \Um \vv + \zv | \Am, V_1^{i-1}, V_{i+1}^n \right ) \\
& \stackrel{(b)}{=}  & \sum_{i=1}^n I \left (V_i ; \Am \uv_i V_i  + \zv | \Am,\Um \right ) \\
& \stackrel{(c)}{=}  & \sum_{i=1}^n I \left (V_i ; \uv_i^\dagger \Am^\dagger \Am \uv_i V_i  + W_i | \Am,\Um \right )
\end{eqnarray}
where (a) follows from the fact that $\vv$ is iid, in (b)
we define $\uv_i$ to be the $i$-th columns of $\Um$ and in (c)
we define $\uv_i^\dagger \Am^\dagger \zv = W_i \sim \Cc\Nc(0,\uv_i^\dagger \Am^\dagger \Am^\dagger \uv_i )$, conditionally on $\Am, \Um$.
Dividing both sides by $n$, letting defining the iid variables $Z_i \sim \Cc\Nc(0,1)$ and taking the limit, we obtain
\begin{eqnarray}  \label{mfb}
\Ic_1 & \leq  & \lim_{n\rightarrow \infty} \frac{1}{n} \sum_{i=1}^n I \left (\left . V_i ; \sqrt{\uv_i^\dagger \Am^\dagger \Am \uv_i}  V_i  + Z_i \right | \Am,\Um \right )   \\
& \stackrel{(a)}{\leq} &  \lim_{n\rightarrow \infty}  I \left (\left . V_i ; \sqrt{\frac{1}{n} \sum_{i=1}^n \uv_i^\dagger \Am^\dagger \Am \uv_i}  V_i  + Z_i \right | \Am,\Um \right )   \\
& = & I\left (V_0; \sqrt{\EE[|{\sf R}|^2]} V_0 + Z \right )
\end{eqnarray}
where $Z \sim Z_i$, where by definition $\frac{1}{n} \sum_{i=1}^n \uv_i^\dagger \Am^\dagger \Am \uv_i = \frac{1}{n} \trace (\Rm) \rightarrow \EE[|{\sf R}|^2]$ and
where in (a) we used Jensen's inequality and the fact that the mutual information $I(V; \sqrt{s} V + Z)$, for any distribution of $V$ with bounded second moment, is concave in $s$ \cite{guo-verdu-shamai}.

{\bf Proof of bound (\ref{I1LB}):}
Let $\Um = [\uv_1, \ldots, \uv_n]$ where $\uv_j$ denotes the $j$-th column of $\Um$. Then,   we have
\begin{eqnarray}
I(\vv; \yv | \Am,\Um)
& = & \sum_{i=1}^n I(V_i ; \yv |\Am,\Um, V_{i+1}^n) \\
& = & \sum_{i=1}^n I\left (\left . V_i ; \yv  -  \Am \sum_{j=i+1}^{n} \uv_j V_j \right |\Am,\Um, V_{i+1}^n \right ) \label{lbstep1} \\
& \geq & \sum_{i=1}^n I\left (\left . V_i ; \gv_i^\dagger \left ( \yv  -  \Am \sum_{j=i+1}^{n} \uv_j V_j \right ) \right |\Am,\Um, V_{i+1}^n \right ) \label{lbstep2} \\
& = & \sum_{i=1}^n I\left (\left . V_i ; \gv_i^\dagger \left ( \Am \sum_{j=1}^i \uv_j V_j + \zv \right ) \right |\Am,\Um \right ) \label{lbstep3}
\end{eqnarray}
where (\ref{lbstep1}) follows by the chain rule and by subtracting the conditioning term, preserving the mutual information,
(\ref{lbstep2}) holds for any linear projection defined by the vector $\gv_i$, function of $\Am, \Um$,
(\ref{lbstep3}) follows by noticing that the arguments of the mutual information do not depend any longer on $V_{i+1}^n$.

Next, we choose $\widetilde{\gv}_i$ to be the linear MMSE (LMMSE) receiver for ``user'' $i$, of the formally equivalent  CDMA system
\begin{eqnarray}  \label{merouane-channel}
\rv_i = \Am \uv_i V_i + \Am \underbrace{\left [ \uv_1, \uv_2, \ldots, \uv_{i-1} \right ]}_{\Um_{i-1}}  \left [ \begin{array}{c} V_1 \\ V_{2} \\ \vdots \\ V_{i-1} \end{array} \right ] + \zv,
\end{eqnarray}
In particular, using $\EE[|V_i|^2] = q \Pc_x$, we obtain
\begin{eqnarray} \label{lmmse-bingo}
\gv_i = \left [ \Id  + q\Pc_x  \Am \Um_{i-1} \Um_{i-1}^\dagger \Am^\dagger \right ]^{-1} \Am \uv_i
\end{eqnarray}
We indicate by
\begin{eqnarray}
\eta^{(n)}(q\Pc_x; i/n)  = \uv_i^\dagger \Am^\dagger \left [ \Id  + q \Pc_x \Am \Um_{i-1}\Um_{i-1}^\dagger \Am^\dagger \right ]^{-1} \Am \uv_i
\end{eqnarray}
the corresponding {\em multiuser efficiency} of the LMMSE detector for ``user'' $i$. Noticing that, in the limit of $n \rightarrow \infty$,
the residual noise plus interference at the output of the LMMSE detector is marginally Gaussian (we omit the explicit proof of this well-known fact,
which holds under the assumptions of our model), \cite{fnt} letting $\beta = i/n$, and denoting by $\eta(q \Pc_x; \beta)$ the limiting multiuser efficiency for
$n \rightarrow \infty$,  from (\ref{lbstep3}) we arrive at (\ref{I1LB}) by dividing by $n$ and
taking the limit.

\section{Proof of the Decoupling Principle}\label{proof:decoupling}

Notations and definitions are as in Section \ref{decoupling} and Appendix \ref{proof:x}.
We let $(b_{0\kappa}, b_\kappa, \widehat{b}_\kappa)$ denote the $\kappa$-th components of the random vectors
$\bv_0, \bv, \widehat{\bv}$, obeying the joint $n$-variate conditional  distribution (\ref{joint-n-variate}) for given $\Am,\Um$,
with $\widehat{\bv} = \widehat{\bv}(\yv, \Am, \Um)$ given by (\ref{PME}).
We are interested in showing that the asymptotic joint marginal distribution of $(b_{0\kappa}, b_\kappa, \widehat{b}_\kappa)$, 
for some generic index $\kappa$, converges to the joint distribution of the triple $(B_0, B, \widehat{B})$ 
given by (\ref{joint-decoupled}) in Section \ref{decoupling}, independent of $\kappa$.

To this purpose, we follow in the footsteps of \cite{guo-verdu} and consider the calculation of the joint moments
$\EE[b_{0\kappa}^i b^j_\kappa]$ for arbitrary integers $i,j \geq 0$.
Since the moments are uniformly bounded, the $\kappa$-th joint marginal distribution
is thus uniquely determined  due to Carleman's Theorem \cite[p. 227]{feller}. 
The desired result will follow upon showing that the moments converge to limits independent of $\kappa$.
Furthermore, as we will see, the form of the asymptotic moments yields explicitly the joint distribution of $(B_0, B, \widehat{B})$ 
given in (\ref{joint-decoupled}).

In order to proceed, we define the replicated model given by the distribution
of $\bv_0, \yv, \bv_1,\ldots, \bv_u$, for given $\Am, \Um$, as:
\begin{eqnarray} \label{TCVS-joint-replica}
p_{\bv_0}(\bv_0)  p_{\yv|\bv_0,\Am,\Um}(\yv|\bv_0,\Am,\Um)  \prod_{a=1}^u q_{\bv|\yv,\Am,\Um}(\bv_a|\yv,\Am,\Um).
\end{eqnarray}
All expectations in the following derivations are with respect to the joint measure (\ref{TCVS-joint-replica}).
For a function $f(\bv_0, \bv_1, \ldots, \bv_u)$, we define
\begin{eqnarray} \label{TCVS-Zuh}
Z^{(u)}(\yv, \Am, \Um, \bv_0; h) & = & \sum_{\bv_1, \ldots, \bv_u} e^{h f(\bv_0, \bv_1, \ldots, \bv_u)} \prod_{a=1}^u q_{\bv}(\bv_a) q_{\yv|\bv,\Am,\Um}(\yv|\bv_a,\Am,\Um).
\end{eqnarray}
By \cite[Lemma 1]{guo-verdu}, if $\EE[f(\bv_0, \bv_1, \ldots, \bv_u) | \yv, \Am, \Um, \bv_0]$ is $O(n)$ and does not depend on $u$, then
\begin{eqnarray} \label{lemma-gv1}
\lim_{n \rightarrow \infty} \frac{1}{n}  \EE[f(\bv_0, \bv_1, \ldots, \bv_u) ] = \lim_{n \rightarrow \infty} \lim_{u \rightarrow 0} \left . \frac{\partial}{\partial h} \frac{1}{n} \log \EE[ Z^{(u)}(\yv, \Am, \Um, \bv_0; h) ] \right |_{h = 0}.
\end{eqnarray}
In our case, we let $f(\bv_0, \bv_1, \ldots, \bv_u) = \sum_{k=1}^{n} b_{0k}^i b_{mk}^j$ for given $i,j \geq 0$, and some replica index
$m \in \{1,\ldots, u\}$.  By the symmetry with respect to the replica index and the indices of the vector components, for any $\kappa$ we can write
\begin{eqnarray} \label{moments-step0}
\EE[b_{0\kappa}^i b_\kappa^j] & = & \lim_{n \rightarrow \infty} \frac{1}{n} \EE \left [ \sum_{k=1}^{n}  b_{0k}^i b_{mk}^j \right ]   \\
& = & \lim_{n \rightarrow \infty} \frac{1}{n} \EE[ f(\bv_0, \bv_1, \ldots, \bv_u) ].
\end{eqnarray}
Using the procedure outlined before, we need to calculate
\begin{eqnarray} \label{TCVS-ziofa}
\lim_{n \rightarrow \infty} \lim_{u \rightarrow 0} \left . \frac{\partial}{\partial h} \frac{1}{n} \log \EE[ Z^{(u)}(\yv, \Am, \Um, \bv_0; h) ] \right |_{h = 0}.
\end{eqnarray}
As usual in replica derivations, we switch limits and calculate first
\begin{eqnarray} \label{TCVS-ziofafa}
\lim_{n \rightarrow \infty} \frac{1}{n} \log \EE[ Z^{(u)}(\yv, \Am, \Um, \bv_0; h) ].
\end{eqnarray}
In passing, we notice that $Z^{(u)}(\yv, \Am, \Um, \bv_0; h=0) =  Z^u(\yv, \Am, \Um)$, so that the calculation of
(\ref{TCVS-ziofa}) is closely related to the calculation of the free energy by the replica method in Appendix \ref{proof:x},
i.e., to  the evaluation of the limit
\begin{eqnarray} \label{TCVS-free-energy1}
\lim_{n \rightarrow \infty} \frac{1}{n} \log \EE[ Z^u(\yv, \Am, \Um) ].
\end{eqnarray}

Operating along the same steps leading to (\ref{free-energy-step5}) in the derivation of the free energy (see Appendix \ref{proof:x}),  we arrive at:
\begin{align}
&  \lim_{n \rightarrow \infty} \frac{1}{n} \log \EE[ Z^{(u)}(\yv, \Am, \Um, \bv_0; h)]   \\
& = \; u \log \gamma - u \log \pi - \log(1 + u\gamma) +   \\
&  \; + \lim_{n \rightarrow \infty} \frac{1}{n} \log \left ( \sum_{\bv_0, \ldots, \bv_u} \int  d\Xm_0 \cdots d\Xm_u p_{\bv_0}(\bv_0) \prod_{a=1}^u q_{\bv}(\bv_a)  \prod_{a=0}^u p_{\xv}(\Xm_a)    \right . \label{sucalimit} \\
& \;  \left . \cdot \exp \left ( h \sum_{k=1}^n b_{0k}^i b_{mk}^j\right )  \cdot  \exp \left ( -n \sum_{i=1}^u \int_0^{\lambda_i(\Lm)} \Rc_{\Rm}(-w) dw \right ) \right )  \label{gfunct}
\end{align}
We notice that the second exponential term in (\ref{gfunct}) is identical to what appears in the computation of (\ref{TCVS-free-energy1})
and, following the steps in Appendix \ref{proof:x}, yields an exponential term $\exp (- n \Gc^{(u)}(\Qm))$ given in (\ref{chandra-exponent}),
function of the empirical correlations  of the vectors $\vv_a = \Xm_a \bv_a$ as defined in (\ref{empirical-corr}), and collected in the empirical correlation
matrix $\Qm$ whose form, under the RS assumption,  is given in  (\ref{QRS}).

Invoking the large deviation theorem,  we can write
\begin{align} 
& \displaystyle{\sum_{\bv_0, \ldots, \bv_u} \int  d\Xm_0 \cdots d\Xm_u p_{\bv_0}(\bv_0) 
\prod_{a=1}^u q_{\bv}(\bv_a)  \prod_{a=0}^u p_{\xv}(\Xm_a) \cdot \exp \left ( h \sum_{k=1}^n b_{0k}^i b_{mk}^j\right )
\cdot   \exp \left ( -n \Gc^{(u)} (\Qm) \right )}  \\
& \; = \EE \left [ \int \exp \left ( -n \Gc^{(u)} (\Qm) \right ) \mu^{(u)}_n(d\Qm;h |\bv_0,\ldots, \bv_u, \Xm_0, \ldots, \Xm_u) \right ] \\
& \; = \int \exp \left ( -n \Gc^{(u)} (\Qm) \right ) \mu^{(u)}_n(d\Qm;h) \label{exact-large-dev} \\
& \; \approx \int   \exp \left ( -n \left ( \Gc^{(u)} (\Qm)+ I^{(u)}(\Qm;h) \right ) \right ) d\Qm \label{moments-step1}
\end{align}
where the approximation step holds in the sense that when $n$ gets large 
we can replace the argument of the logarithm in 
(\ref{sucalimit}) with the quantity in (\ref{moments-step1}).

Using Cram\'er's theorem, we have that the rate function $I^{(u)}(\Qm;h)$  for the measure
\begin{eqnarray} 
\mu^{(u)}_n(d\Qm;h) & = &
\sum_{\bv_0, \ldots, \bv_u} \int  d\Xm_0 \cdots d\Xm_u p_{\bv_0}(\bv_0) \prod_{a=1}^u q_{\bv}(\bv_a)  \prod_{a=1}^u p_{\xv}(\Xm_a)   \\
& & \cdot \exp \left ( h \sum_{k=1}^n b_{0k}^i b_{mk}^j\right ) \cdot \prod_{a\leq a'}^u \delta \left ( \sum_{k=1}^n x_{ka}b_{ka}x^*_{ka'}b^*_{ka'} - n Q_{a,a'}\right ) d\Qm  \\
& = & \EE \left [ \exp \left ( h \sum_{k=1}^n b_{0k}^i b_{mk}^j\right ) \prod_{a\leq a'}^u \delta \left ( \sum_{k=1}^n x_{ka}b_{ka}x^*_{ka'}b^*_{ka'} - n Q_{a,a'}\right ) d\Qm \right ] \label{measure-fuckedup}
\end{eqnarray}
is given by the Legendre-Fenchel  transform
\begin{eqnarray}
I^{(u)}(\Qm;h) = \sup_{\tilde{\Qm}} \left \{ \trace(\tilde{\Qm} \Qm ) - \log M^{(u)}(\tilde{\Qm};h) \right \}
\end{eqnarray}
where the relevant MGF for the measure (\ref{measure-fuckedup}) is
\begin{eqnarray} \label{MGFV-fuckedup}
M^{(u)}(\tilde{\Qm};h) = \EE \left [ \exp\left (h B_0^i B_m^j +  \underline{\bv}^\herm \underline{\Xm}^\herm \tilde{\Qm} \underline{\Xm} \underline{\bv} \right ) \right ],
\end{eqnarray}
where we define $\underline{\bv} = (B_0, B_1, \ldots, B_u)^\transp$ and $\underline{\Xm} = (X_0, X_1, \ldots, X_u)^\transp$, with
$B_0 \sim p_{B_0}$ (Bernoulli-$q$), 
$B_a \sim q_B$ (marginal of the assumed prior distribution $q_{\bv}(\cdot)$) 
for $1 \leq a \leq u$, and $X_a \sim p_X(x) = \frac{1}{\pi\Pc_x} e^{-|x|^2/\Pc_x}$ 
for all $0 \leq a \leq u$.

Plugging (\ref{MGFV-fuckedup}) into (\ref{moments-step1}) and the resulting expression in the limit of $\frac{1}{n} \log (\cdot)$ appearing in (\ref{sucalimit})
and applying Varadhan's lemma, we arrive at the saddle-point condition
\begin{eqnarray} \label{moments-step2}
\lim_{n \rightarrow \infty} \frac{1}{n} \log \left (  \int \mu^{(u)}_n(d\Qm)  \exp \left ( -n \Gc^{(u)}(\Qm) \right ) \right )  & = &
 -\inf_{\Qm} \sup_{\tilde{\Qm}} \left \{ \Gc^{(u)}(\Qm)  + \trace(\tilde{\Qm} \Qm ) - \log M^{(u)}(\tilde{\Qm};h)  \right \}   \nonumber \\
 & &
\end{eqnarray}
Following the replica derivation steps outlined at the beginning of this Appendix, we have to determine
the saddle-point $\Qm^\star(h)$ and $\tilde{\Qm}^\star(h)$ achieving the extremal condition in
(\ref{moments-step2}), for general $h$, and finally replace the result in (\ref{TCVS-ziofafa}),
differentiate with respect to $h$ and evaluate the result for $h = 0$.
Using again the result of Appendix \ref{antonia-general-proof},
since the function in (\ref{TCVS-ziofafa}) is differentiable and admits a minimum and a maximum,
we can replace the saddle-point of (\ref{moments-step2}) for $h = 0$, denoted by $\Qm^\star$ and $\tilde{\Qm}^\star$, and then differentiate
the result with respect to to $h$ and let $h \rightarrow 0$.
Noticing that for $h = 0$ the saddle-point condition (\ref{moments-step2}) coincides with the saddle-point condition (\ref{free-energy-step7}),
we have that $\Qm^\star$ and $\tilde{\Qm}^\star$ coincide with what derived in Appendix \ref{proof:x} for the
free energy (\ref{TCVS-free-energy1}). In particular, under the RS assumption, these parameters are given by the fixed-point
equations (\ref{saddle-final-chi}) -- (\ref{saddle-final-eta}). Furthermore, since (\ref{moments-step2}) and therefore the whole limit
(\ref{TCVS-ziofafa}) depends on $h$ only through the log-MGF term, using (\ref{lemma-gv1}) and (\ref{moments-step0})
we arrive at
\begin{eqnarray} \label{moments-step3}
\EE[ b_{0\kappa}^i b_{m\kappa}^j] & = & \lim_{u \rightarrow 0} \left . \frac{\partial}{\partial h} \log M^{(u)}(\tilde{\Qm}^\star;h) \right |_{h=0}   \\
& = & \lim_{u \rightarrow 0} \frac{\EE\left [ B_0^i B_m^j \exp\left ( \underline{\bv}^\herm \underline{\Xm}^\herm \tilde{\Qm}^\star \underline{\Xm} \underline{\bv} \right ) \right ]}
{\EE\left [\exp\left ( \underline{\bv}^\herm \underline{\Xm}^\herm \tilde{\Qm}^\star \underline{\Xm} \underline{\bv} \right ) \right ] }.
\end{eqnarray}
The denominator of (\ref{moments-step3}) is identical to the MGF $M^{(u)}(\tilde{\Qm}^\star)$ defined in (\ref{MGFV}) and,
as shown in Appendix \ref{proof:x} we have that  $\lim_{u \rightarrow 0} M^{(u)}(\tilde{\Qm}^\star) = 1$.
As for the numerator, we follow steps similar to the derivation of (\ref{MGFV4}) and obtain
\begin{eqnarray} 
\EE[ b_{0\kappa}^i b_{m\kappa}^j]
& = & \EE \Big [ \int \frac{\eta}{\pi}  B_0^i   \exp\left ( - \eta |z - X_0B_0|^2  \right ) \cdot  B_m^j \exp \left ( - \xi  |z  - X_mB_m|^2 \right ) \cdot   \\
& & \exp \left ( - \xi \sum_{a \neq m }^u |z  - X_aB_a|^2 \right )  \; e^{\xi |z|^2 u}   dz \Big ] \Big |_{u \downarrow 0} \\
& = & \EE \Big [ \int \frac{\eta}{\pi}  B_0^i  \exp\left ( - \eta |z - X_0B_0|^2  \right ) \cdot B_m^j \exp \left ( - \xi  |z  - X_mB_m |^2 \right ) \cdot   \\
& & \EE \left [ \exp \left ( - \xi |z  - XB|^2 \right ) \right ]^{u-1}  \; e^{\xi |z|^2 u}   dz \Big ] \Big |_{u \downarrow 0}  \\
& = & \EE \left  [ \int \frac{\eta}{\pi}  B^i_0  \exp\left ( - \eta |z - X_0B_0|^2  \right ) \cdot
\EE \left [ \frac{B^j \exp \left ( - \xi  |z  - XB|^2 \right )}{ \EE \left [ \exp \left ( - \xi |z  - XB|^2 \right ) \right ]} \right ] dz \right ],    \\
& &  \label{moments-step4}
\end{eqnarray}
where we define $X \sim X_a$ for $a = 0, \ldots, u$, $B \sim B_a$ for $a = 1,\ldots, u$, and
where the parameters $\eta$ and $\xi$ are given by the fixed-point equations (\ref{saddle-final-chi}) -- (\ref{saddle-final-eta}).

Finally, we define a single-letter joint probability distribution and restate the expectations
appearing in (\ref{moments-step4}) in terms of this new single-letter model.
We let
\begin{eqnarray} \label{decoupled-p-moments}
p_{Y|B_0; \eta}(y|b_0) = \frac{\eta}{\pi} \int \exp \left ( - \eta |y - x b_0|^2 \right ) p_X(x) dx
\end{eqnarray}
denote the transition probability density of the complex circularly symmetric AWGN channel with Gaussian circularly symmetric fading
not known at the receiver,
\begin{eqnarray} \label{decoupled-channel1}
Y = X_0 B_0 + \eta^{-\frac12} Z,
\end{eqnarray}
with $Z \sim \Cc\Nc(0,1)$ and $X_0 \sim p_X(x)$. Also, we define the conditional pdf
\begin{eqnarray} \label{decoupled-q2}
q_{Y|B; \xi}(y|b) = \frac{\xi}{\pi} \int \exp \left ( - \xi |y - xb|^2 \right ) p_X(x) dx
\end{eqnarray}
and, using Bayes rule,  consider the a-posteriori probability distribution
\begin{eqnarray} \label{decoupled-posterior-q1}
q_{B|Y; \xi} (b | y) & = &
\frac{q_{Y|b; \xi}(y|b) q_B(b)}{\sum_{b'} q_{Y|b'; \xi}(y|b) q_B(b')}   \\
& = & \frac{\int \exp\left ( - \xi |y - xb|^2\right )  q_B(b) p_X(x) dx}{\EE\left [ \exp\left ( - \xi |y - XB|^2\right )  \right ]}.
\end{eqnarray}
The joint single-letter probability distribution of interest for the variables $B_0, Y$ and $B$ is given by
\begin{eqnarray} \label{decoupled-joint-moments}
p_{B_0,Y,B; \eta,\xi} (b_0, y, b; \eta,\xi) = p_{B_0}(b_0) \ p_{Y|B_0; \eta}(y|b_0) \ q_{B|Y; \xi}(b|y).
\end{eqnarray}
With these definitions,  it is immediate to identify the moment expression (\ref{moments-step4}) as the joint moment of the single-letter
probability distribution (\ref{decoupled-joint-moments}), by writing:
\begin{eqnarray} 
\EE[ B_{0\kappa}^i B_{m\kappa}^j]
& = & \EE \Big [ \int \frac{\eta}{\pi}  B^i_0  \exp\left ( - \eta |z - X_0B_0|^2  \right ) \cdot
\EE \left [ \frac{B^j \exp \left ( - \xi  |z  - XB|^2 \right )}{ \EE \left [ \exp \left ( - \xi |z  - XB|^2 \right ) \right ]} \right ] dz \Big ]   \\
& = & \int  \EE \left [ B^i_0  \int \frac{\eta}{\pi} \exp\left ( - \eta |z - x_0 B_0|^2  \right ) p_X(x_0) dx_0 \right ] \cdot   \\
& &  \EE \left [ \frac{B^j \int \exp \left ( - \xi  |z  - xB|^2 \right ) p_X(x) dx }{ \EE \left [ \exp \left ( - \xi |z  - XB|^2 \right ) \right ]} \right ] dz   \\
& = & \int  \left ( \sum_{b_0} b^i_0  p_{Y|B_0;\eta}(z|b_0) p_{B_0}(b_0)  \right ) \cdot  \left ( \sum_{b} b^j q_{B|Y; \xi}(b|z) \right ) dz  \\
& = &   \sum_{b_0} \sum_b \int  b^i_0  b^j  p_{B_0}(b_0)  p_{Y|B_0;\eta}(z|b_0)  q_{B|Y; \xi}(b|z)  dz  \\
& = &   \EE [ B_0^i B^j]   \label{moments-step5} 
\end{eqnarray}
where the expectation in (\ref{moments-step5})  the last line is with respect to the probability distribution (\ref{decoupled-joint-moments}).

Summarizing, we have that as far as the joint probability distribution of each component of $\bv$ in (\ref{model1}) and the
corresponding component of the PME $\widehat{\bv}$ (matched or
mismatched), the system decouples asymptotically into a bank of
``parallel'' AWGN channels of the form  (\ref{decoupled-channel1}), with symbol-by-symbol PME given by
\begin{eqnarray} 
\widehat{B} = \EE[B | Y] = \sum_b  \; b \; q_{B|Y; \xi}(b|Y), 
\end{eqnarray}
for $B_0, Y, B$ distributed as in (\ref{decoupled-joint-moments}),
where the parameters $\eta$ and $\xi$ are given by (\ref{saddle-final-chi}) -- (\ref{saddle-final-eta}).

\section{Eigenvalues of the matrix $\Lm$} \label{eigenvalues-L}

The eigenvalues of $\Lm$ are readily computed
from (\ref{LL}). Notice that the matrix $\left ( \Id   -
\frac{1}{\gamma^{-1}+u} \onev \onev^\dagger  \right )$ has eigenvalues
\begin{eqnarray} \nu_1 = \frac{1}{1+u\gamma} \end{eqnarray}
corresponding to the (normalized) eigenvector
$\frac{1}{\sqrt{u}}\onev$, and $\nu_2 = \cdots = \nu_u = 1$,
corresponding to eigenvectors $\ev_2, \ldots, \ev_u$ forming an
orthonormal basis of the orthogonal complement of Span$\{\onev\}$ in
$\CC^u$. It follows that
\begin{eqnarray} \Lm = \frac{\gamma}{ n  } \Sm \Em \;  \diag\left (\frac{1}{1 + u \gamma} , 1, \ldots, 1 \right ) \; \Em^\dagger \Sm^\dagger \end{eqnarray}
where
\begin{eqnarray} \Em = \left [ \frac{1}{\sqrt{u}} \onev, \ev_2, \ldots, \ev_u \right ] \end{eqnarray}
The non-zero eigenvalues of $\Lm$ are the same as those of the ``flipped'' matrix
\begin{eqnarray} \diag\left (\sqrt{\frac{1}{1 + u \gamma}} , 1, \ldots, 1 \right ) \;
\Em^\dagger \left ( \frac{\gamma}{n} \Sm^\dagger \Sm \right ) \Em \;
\diag\left (\sqrt{\frac{1}{1+u \gamma}} , 1, \ldots, 1 \right )
\end{eqnarray}
Under the RS assumption, the empirical
correlation matrix of the vectors $\sv_1, \ldots, \sv_u$ takes on the form
\begin{eqnarray} \label{replica-symmetry}
\frac{1}{n} \Sm^\dagger \Sm \rightarrow (\alpha - \beta) \Id + \beta
\onev \onev^\dagger
\end{eqnarray}
Using the orthonormality properties of the columns of $\Em$, we have
\begin{eqnarray} \Em^\dagger \left ( (\alpha - \beta) \Id + \beta \onev \onev^\dagger \right ) \Em = \diag \left (
\alpha + (u-1)\beta, \alpha - \beta, \ldots, \alpha - \beta \right ) \end{eqnarray}
Finally, we have that under the RS assumption and
in the limit of large $n$ the eigenvalues of $\Lm$ are given by
\begin{eqnarray} \label{eigv-L}
\lambda_1 & = & \frac{\alpha + (u-1)\beta}{\gamma^{-1}+u}   \\
\lambda_a & = & \gamma(\alpha - \beta), \;\;\;\; \mbox{for} \;\; a =
2,\ldots, u
\end{eqnarray}
Using the fact that $\sv_a = \Xm_a\bv_a - \Xm_0 \bv_0$, we have that
\begin{eqnarray} \alpha = \epsilon_1 +  \epsilon_0 - 2 \Re \{ \vartheta \}, \;\;\;\;
\beta = \omega + \epsilon_0 - 2 \Re \{ \vartheta \} \end{eqnarray}
Therefore, the eigenvalues (\ref{eigv-L}) can be expressed in terms
of the correlations $\epsilon_0, \epsilon_1, \vartheta, \omega$ in the form (\ref{eigenvalues}).

\section{A property of stationary points of multivariate functions} \label{antonia-general-proof}

Let $f(\tv, \vv, \theta)$ be a differentiable multivariate function
with $\tv \in \mathbb{C}^N$, $\vv \in \mathbb{C}^L$ and $\theta \in
\mathbb{R}$. Let $t_n$ with $n=1, \ldots, N$, $v_\ell$ with $\ell=1,
\ldots, L$ denote the $n$-th and $\ell$-th component of $\tv$ and
$\vv$ respectively. We are interested in evaluating:
$$\frac{d}{d \theta} \inf_{\tv} \sup_{\vv} f(\tv, \vv, \theta)|_{\theta=0}
$$

Let:
$$
[\tv^*(\theta), \vv^*(\theta)] = \arg \inf_{\tv} \sup_{\vv} f(\tv,
\vv, \theta)
$$

then:
\begin{eqnarray}
\frac{d}{d \theta} \inf_{\tv} \sup_{\tv} f(\tv, \vv, \theta) &=&
\frac{d}{d \theta} f(\tv^*(\theta), \vv^*(\theta), \theta) \\ &=&
\sum_{n=1}^N \dot{f}_{t_n}(\tv^*, \vv^*, \theta) \frac{d}{d\theta}
t_{n}^*((\theta)) \\ && + \sum_{\ell=1}^L \dot{f}_{v_{\ell}}(\tv^*,
\vv^*, \theta) \frac{d}{d\theta} v_{\ell}^*((\theta)) +
\dot{f}_{\theta}(\tv^*, \vv^*, \theta)\label{dinnertime}
\end{eqnarray}
with
\begin{eqnarray}
\dot{f}_{t_n}(\tv^*, \vv^*, \theta)& =&  \frac{\partial}{\partial
t_n}
{f}(\tv, \vv, \theta)|_{\tv=\tv^*, \vv=\vv^*}, \\
\dot{f}_{v_n}(\tv^*(\theta), \vv^*(\theta), \theta) &=&
\frac{\partial}{\partial v_n}
{f}(\tv, \vv, \theta)|_{\tv=\tv^*,  \vv=\vv^*},\\
\dot{f}_{\theta}(\tv^*(\theta), \vv^*(\theta), \theta) &=&
\frac{\partial}{\partial \theta} {f}(\tv, \vv, \theta)|_{\tv=\tv^*,
\vv=\vv^*}.
\end{eqnarray}

Under the assumption that the supremum and the infimum are achieved
by $f(\tv, \vv, \theta)$, by Fermat's theorem every local extremum
of a differentiable function is a stationary point hence by their
definition $\tv^*(\theta), \vv^*(\theta)$ are such that for all
$\theta$
\begin{eqnarray}
\dot{f}_{t_n}(\tv^*, \vv^*, \theta)& =&  0, \\
\dot{f}_{v_n}(\tv^*, \vv^*, \theta) &=& 0.
\end{eqnarray}

Hence (\ref{dinnertime}) becomes:
\begin{eqnarray}
\frac{d}{d \theta} \inf \sup f(\tv, \vv, \theta) &=&
\dot{f}_{\theta}(\tv^*(\theta), \vv^*(\theta),
\theta)\label{dinnertime1}
\end{eqnarray}
Consequently
\begin{eqnarray}
\frac{d}{d \theta} \inf \sup f(\tv, \vv, \theta)|_{\theta=0} &=&
\dot{f}_{\theta}(\tv^*(0), \vv^*(0),0) \label{dinnertime2}
\end{eqnarray}
from which it follows that we are allowed to compute the
saddle-point (and hence the fixed-point equation) for $\theta = 0$,
then replace the result in the multivariate function, and
differentiate the result w.r.t. to $\theta$ and  then let $\theta =
0$.

\section{Useful formulas} \label{app:formulas}

This Appendix is devoted to provide methods and explicit formulas to evaluate the quantities appearing 
in the main results. It is worthwhile to notice that the numerical evaluation of the fixed-point equations and
the corresponding free energy is not completely trivial from a numerical stability viewpoint, 
especially for large signal-to-noise ratio $q\Pc_x$ and small sparsity $q$ and sampling rate $p$. Therefore, 
some care must be dedicated to avoid as much as possible brute-force numerical integration. 

We start by considering the calculation of $I(V_0; \sqrt{a} V_0 + Z)$, for $V_0 = X_0 B_0$ Bernoulli-Gaussian, 
and $Z \sim \Cc\Nc(0,1)$, which is instrumental in evaluating (\ref{1l}) and the bounds (\ref{minfo-bound3}) and (\ref{I1LB}), 
for suitable choices of the parameter $a > 0$. We can write
\begin{eqnarray} 
I(V_0; \sqrt{a} V_0 + Z) & = & h(\sqrt{a} V_0 + Z) - h(Z) \nonumber \\
& = & - \EE \left [\log \left ( \frac{q}{1+a\Pc_x} e^{-|Y|^2/(1+a\Pc_x)} + (1 - q) e^{-|Y|^2} \right ) \right ] - \log e\label{IV0}
\end{eqnarray}
where $Y = \sqrt{a} V_0 + Z$. The expectation in (\ref{IV0}), can be calculated by integration in polar coordinates
and, after some algebra, takes on the form
\begin{align}
& q \int_0^\infty \log \left ( \frac{q}{1+a\Pc_x} e^{-r} + (1 - q) e^{-(1+a\Pc_x) r} \right ) e^{-r} dr \nonumber \\
& + (1 - q) \int_0^\infty \log \left ( \frac{q}{1+a\Pc_x} e^{-r/(1+a\Pc_x)} + (1 - q) e^{-r} \right ) e^{-r} dr. 
\end{align}
Finally, both the above integrals can be efficiently and accurately evaluated by using Gauss-Laguerre quadratures. 

Similarly, the MMSE term appearing in (\ref{e:fix-pointeq}) can be calculated as
follows. Letting $Y = V_0 + \eta^{-\frac{1}{2}} Z$,  we have
\[ \EE[V_0|Y] = 
G(|Y|^2;\eta,q,\Pc_x) \frac{\Pc_x \eta}{1 + \Pc_x\eta} Y, \]
where 
\[ G(z;\eta, q, \Pc_x) = \frac{\frac{q}{1+\Pc_x\eta} \exp(-\mu z)}{\frac{q}{1+\Pc_x\eta} \exp(-\mu z) + (1 - q) \exp(-\eta z)} \]
and where $\mu = \eta/(1 + \Pc_x\eta)$.  Notice that for $q = 1$ the observation model becomes jointly Gaussian, and we obtain the usual Gaussian MMSE estimator  $\EE[V_0|Y] = \frac{\Pc_x \eta}{1 + \Pc_x\eta} Y$. 
The  resulting MMSE in the general Bernoulli-Gaussian case is given by 
\begin{eqnarray*}
{\sf mmse} \left ( V_0  | V_0 + \eta^{-\frac12}  Z \right ) & = & \EE\left [|V_0|^2 \right ] - \EE \left [ | \EE[V_0|Y] |^2 \right ] \\
& = & q\Pc_x - \left ( \frac{\Pc_x\eta}{1 + \Pc_x\eta} \right )^2 \EE \left [ G(|Y|^2;\eta,q,\Pc_x)^2  |Y|^2 \right ].
\end{eqnarray*}
Performing integration in polar coordinates and after some algebra we obtain
\[ {\sf mmse} \left ( V_0  | V_0 + \eta^{-\frac12}  Z \right ) = q \left [ \Pc_x - \frac{1}{\eta(1 + \Pc_x\eta)} \Phi \left ( -  (1 + \Pc_x\eta)
\frac{1 - q}{q}, 2, \frac{1}{\Pc_x \eta} \right ) \right ] \]
where $\Phi(a,s,z)$ is known as the Hurwitz-Lerch zeta function \cite{abramowitz}, 
defined as
\[ \Phi(z,s,a) = \frac{1}{\Gamma(s)} \int_0^\infty \frac{t^{s-1} e^{-at}}{1 - z e^{-t}} dt, \]
that can also be efficiently evaluated by Gauss-Laguerre quadratures.  It is immediate to check that for $q = 1$ (jointly Gaussian case) we have 
\[ {\sf mmse} \left ( V_0  | V_0 + \eta^{-\frac12}  Z \right ) = \frac{\Pc_x}{1 + \Pc_x \eta}, \]
as expected. 

In order to evaluate $\Ic_1$ in (\ref{1l}) it is useful to have the integral of the R-transform $\Rc_{\Rm}(-w)$ in closed form. 
For the case of $\Um$ with iid elements, using (\ref{R-transfromGV}) we find, trivially,  
\begin{eqnarray} 
\int_0^\chi \Rc_{\Rm}(-w) dw = p \log(1 + \chi). 
\end{eqnarray}
For the case of Haar-distributed $\Um$, using (\ref{cool2}), we find
\begin{eqnarray} 
\int_0^\chi \Rc_{\Rm}(-w) dw & = & 
\frac{1}{2} (1+\chi - \rho -2p \log(2(1 - p)) + \log(1-p) \nonumber \\
& & - (1 - 2p) \log(1+\chi - 2p+\rho)+ \log(1+\chi (1  - 2p)+\rho)),
\end{eqnarray}
where $\rho = \sqrt{(1+\chi)^2 - 4\chi p}$. 

We conclude by providing the derivation of the closed-form expression of $\EE \left [ |V_0 - \widehat{v}(Y;\xi) |^2 \right ]$ for the Lasso estimator, 
given in (\ref{mse-lasso-closed-form}). 
We have
\begin{eqnarray}  \label{mse-lasso-closed-form-step1}
\EE \left [ |V_0 - \widehat{v}(Y;\xi) |^2 \right ] & = & q\Pc_x  + \EE[|\widehat{v}(Y;\xi) |^2]  - 2 \Re\left \{ \EE[ V_0^* \widehat{v}(Y;\xi)] \right \}. 
\end{eqnarray}
Recalling the expression of $\widehat{v}(Y;\xi)$ in (\ref{scalar-map-lasso}), we have
\begin{eqnarray}
\EE[|\widehat{v}(Y;\xi) |^2] 
& = & \EE\left [ \left | \left [|Y| - \frac{1}{2\xi} \right ]_+ \frac{Y}{|Y|} \right |^2  \right ] \nonumber \\
& = &\int_{|y| > 1/(2\xi)}  \left ( |y| - \frac{1}{2\xi} \right )^2 p_Y(y) dy \nonumber \\
& = &\int_{1/(2\xi)}^\infty  \left ( r - \frac{1}{2\xi} \right )^2  \left [ q \mu e^{-\mu r^2} + (1 - q) \eta e^{-\eta r^2} \right ] 2 r dr  \nonumber \\
& = &2 q \mu \int_{1/(2\xi)}^\infty  \left [ r^3 - r^2/\xi  + r/(4\xi^2) \right ] e^{-\mu r^2} dr \nonumber \\
& & + 2(1 - q) \eta  \int_{1/(2\xi)}^\infty  \left [ r^3 - r^2/\xi  + r/(4\xi^2) \right ] e^{-\eta r^2} dr. \label{zio-lasso}
\end{eqnarray}
In order to solve the integrals in (\ref{zio-lasso}) we use
\begin{eqnarray}
\int_b^\infty 2a x e^{-ax^2} dx & = & e^{-ab^2}  \label{int1} \\
\int_b^\infty 2a x^2 e^{-ax^2} dx & = & b e^{-ab^2} + \frac{\sqrt{\pi} {\rm erfc}(\sqrt{a} b)}{2 \sqrt{a}} \label{int2} \\
\int_b^\infty 2a x^3 e^{-ax^2} dx & = & \frac{(1+ab^2) e^{-ab^2}}{a}. \label{int3} 
\end{eqnarray}
By applying the above integrals  in (\ref{zio-lasso}) and after some manipulation, we obtain
\begin{eqnarray}
\EE[|\widehat{v}(Y;\xi) |^2] & = & 
\frac{q}{\mu} \left [ e^{-\mu'}  -  \sqrt{\pi \mu'} {\rm erfc}(\sqrt{\mu'}) \right ] 
 +  \frac{(1 - q)}{\eta} \left [ e^{-\eta'}  - \sqrt{\pi \eta'} {\rm erfc}(\sqrt{\eta'}) \right ] 
\end{eqnarray}
with $\mu' = \mu/(4\xi^2)$ and $\eta' = \eta/(4\xi^2)$. 

Next, we calculate the expectation $\EE[ V_0^* \widehat{v}(Y;\xi)]$ as follows:
\begin{eqnarray}
\EE[ V_0^* \widehat{v}(Y;\xi)]  & = & \EE \left [ V_0^* \left [|Y| - \frac{1}{2\xi} \right ]_+ \frac{Y}{|Y|} \right ] \nonumber \\
& = & q \EE \left [ \left . X_0^* \left [|Y| - \frac{1}{2\xi} \right ]_+ \frac{Y}{|Y|} \right | B_0 = 1 \right ]. \label{ziobellissimo}
\end{eqnarray}
We notice that $(X_0, Y)$ given $B_0 = 1$ are jointly Gaussian, with mean zero and covariance matrix
\[ {\rm Cov}(X_0,Y) = \left [ \begin{array}{cc}
\Pc_x & \Pc_x \\
\Pc_x & \Pc_x + 1/\eta \end{array} \right ]. \]
Then, $X_0$ given $Y$ is Gaussian with mean 
\[ \EE[X_0 |Y] = \frac{\Pc_x}{\Pc_x + 1/\eta} Y \]
and variance
\[ {\rm Var}(X_0|Y) = \Pc_x - \frac{\Pc_x^2}{\Pc_x + 1/\eta}. \]
Using iterated expectation, we can calculate the expectation in  (\ref{ziobellissimo}) as
\begin{eqnarray} 
\EE \left [ \left . X_0 \left [|Y| - \frac{1}{2\xi} \right ]_+ \frac{Y^*}{|Y|} \right | B_0 = 1 \right ] 
& = & \EE \left [ \left . \EE[X_0 | Y, B_0 = 1]  \left [|Y| - \frac{1}{2\xi} \right ]_+ \frac{Y^*}{|Y|} \right | B_0 = 1 \right ] \nonumber \\
& = & \frac{\Pc_x}{\Pc_x + 1/\eta} \EE \left [ \left . \left [|Y| - \frac{1}{2\xi} \right ]_+ |Y| \right | B_0 = 1 \right ] \nonumber \\
& = & \frac{\Pc_x}{\Pc_x + 1/\eta} \int_{|y| > 1/(2\xi)}  \left (|y| - \frac{1}{2\xi} \right ) |y| p_{Y|B_0=1}(y) dy \nonumber \\
& = & \frac{\Pc_x}{\Pc_x + 1/\eta} \int_{1/(2\xi)}^\infty 2 \mu \left (r - \frac{1}{2\xi} \right ) r^2 e^{-\mu r^2}  dr.
\end{eqnarray}
Using again the integrals (\ref{int1}) -- (\ref{int3}), we obtain
\begin{eqnarray}
\int_{1/(2\xi)}^\infty 2 \mu \left (r - \frac{1}{2\xi} \right ) r^2 e^{-\mu r^2}  dr 
= \frac{1}{\mu} \left [  e^{-\mu'} - \frac{1}{2} \sqrt{\pi \mu'} {\rm erfc}(\sqrt{\mu'}) \right ] .   
\end{eqnarray} 
Finally, replacing all terms in (\ref{mse-lasso-closed-form-step1}), after some simplifications, 
we obtain (\ref{mse-lasso-closed-form}).

\end{document}

%% file: macros.tex
\long\def\comment#1{}

\newfont{\bbb}{msbm10 scaled 700}

\newfont{\bb}{msbm10 scaled 1100}
\newcommand{\CC}{\mbox{\bb C}}

\newcommand{\RR}{\mbox{\bb R}}
\newcommand{\PP}{\mbox{\bb P}}

\newcommand{\EE}{\mbox{\bb E}}

\newcommand{\bv}{{\bf b}}

\newcommand{\ev}{{\bf e}}

\newcommand{\gv}{{\bf g}}

\newcommand{\rv}{{\bf r}}
\newcommand{\sv}{{\bf s}}
\newcommand{\tv}{{\bf t}}
\newcommand{\uv}{{\bf u}}

\newcommand{\vv}{{\bf v}}
\newcommand{\xv}{{\bf x}}
\newcommand{\yv}{{\bf y}}
\newcommand{\zv}{{\bf z}}

\newcommand{\zerov}{{\bf 0}}
\newcommand{\onev}{{\bf 1}}

\newcommand{\Am}{{\bf A}}
\newcommand{\Bm}{{\bf B}}

\newcommand{\Dm}{{\bf D}}
\newcommand{\Em}{{\bf E}}
\newcommand{\Fm}{{\bf F}}
\newcommand{\Gm}{{\bf G}}
\newcommand{\Hm}{{\bf H}}
\newcommand{\Id}{{\bf I}}

\newcommand{\Lm}{{\bf L}}

\newcommand{\Qm}{{\bf Q}}
\newcommand{\Rm}{{\bf R}}
\newcommand{\Sm}{{\bf S}}
\newcommand{\Tm}{{\bf T}}
\newcommand{\Um}{{\bf U}}

\newcommand{\Vm}{{\bf V}}
\newcommand{\Xm}{{\bf X}}

\newcommand{\Cc}{{\cal C}}

\newcommand{\Ec}{{\cal E}}
\newcommand{\Fc}{{\cal F}}
\newcommand{\Gc}{{\cal G}}
\newcommand{\Hc}{{\cal H}}
\newcommand{\Ic}{{\cal I}}

\newcommand{\Nc}{{\cal N}}

\newcommand{\Pc}{{\cal P}}

\newcommand{\Rc}{{\cal R}}

\newcommand{\Vc}{{\cal V}}
\newcommand{\Xc}{{\cal X}}

\newcommand{\Gammam}{\hbox{\boldmath$\Gamma$}}

\newcommand{\Sigmam}{\hbox{\boldmath$\Sigma$}}

\newcommand{\diag}{{\hbox{diag}}}
\renewcommand{\det}{{\hbox{det}}}
\newcommand{\trace}{{\hbox{tr}}}

\renewcommand{\arg}{{\hbox{arg}}}

\newcommand{\SNR}{{\sf SNR}}

\renewcommand{\Re}{{\rm Re}}

\newcommand{\eqdef}{\stackrel{\Delta}{=}}

\newcommand{\herm}{{\sf H}}
\newcommand{\transp}{{\sf T}}